\documentclass[a4paper,USenglish,11pt]{article}
\usepackage{geometry}
\usepackage{graphicx}
\usepackage{comment}
\usepackage[section]{placeins}
\usepackage{booktabs,tabularx}
\usepackage{amsthm}
\usepackage{graphicx,subcaption}
\usepackage{bbm}
\usepackage{tabularx,environ,tikz}
\usetikzlibrary{calc, fit}
\usepackage{url}

\usepackage[pagebackref=true]{hyperref}
\usepackage[utf8]{inputenc} 
\usepackage{amsfonts,amssymb,amsmath} 
\usepackage[linesnumbered,noend]{algorithm2e}
\SetKwInOut{Input}{Input}
\usepackage{natbib}

\usepackage{xspace}
\usepackage[textwidth=20mm,obeyFinal,colorinlistoftodos]{todonotes}
\usepackage[capitalize]{cleveref}

\usepackage{marginnote}
\newcommand{\mytodo}[2]{\xspace}
\newcommand{\myrevtodo}[2]{{
        \let\marginpar\marginnote
        \reversemarginpar
        \renewcommand{\baselinestretch}{0.8}
        }}
\newcommand{\myinlinetodo}[2]{\todo[size=\small, color=#1!50!white, inline, 
caption={}]{#2}\xspace}
\newcommand{\registerAuthor}[3]{
    \expandafter\newcommand\csname #2com\endcsname[1]{\mytodo{#3}{\textsc{#2}: 
    ##1}}
    \expandafter\newcommand\csname 
    #2revcom\endcsname[1]{\myrevtodo{#3}{\textsc{#2}: ##1}}
    \expandafter\newcommand\csname 
    #2inline\endcsname[1]{\myinlinetodo{#3}{\textsc{#2}: ##1}}
    \expandafter\newcommand\csname 
    #2inlineLater\endcsname[1]{\lv{\myinlinetodo{#3}{\textsc{#2}: ##1}}}
}

\makeatletter
\def\moverlay{\mathpalette\mov@rlay}
\def\mov@rlay#1#2{\leavevmode\vtop{
        \baselineskip\z@skip \lineskiplimit-\maxdimen
        \ialign{\hfil$\m@th#1##$\hfil\cr#2\crcr}}}
\newcommand{\charfusion}[3][\mathord]{
    #1{\ifx#1\mathop\vphantom{#2}\fi
        \mathpalette\mov@rlay{#2\cr#3}
    }
    \ifx#1\mathop\expandafter\displaylimits\fi}
\makeatother

\newtheorem{observation}{Observation}

\registerAuthor{Klaus Heeger}{kh}{green}
\registerAuthor{Niclas Boehmer}{nb}{orange}
\registerAuthor{Robert bredereck}{rb}{yellow}

\newcommand{\const}{\textsc{Constructive}\xspace}
\newcommand{\dest}{\textsc{Destructive}\xspace}
\newcommand{\exact}{\textsc{Exact}\xspace}
\newcommand{\ex}{\textsc{Exists}\xspace}

\newcommand{\swap}{\textsc{Swap}\xspace}
\newcommand{\addNoSpace}{\textsc{Add}}
\newcommand{\add}{\addNoSpace\xspace}
\newcommand{\addag}{\text{add}\xspace}
\DeclareMathOperator{\orig}{orig}
\newcommand{\delete}{\textsc{Delete}\xspace}

\newcommand{\accD}{\textsc{DeleteAcceptability}\xspace}

\newcommand{\reor}{\textsc{Reorder}\xspace}

\newcommand{\swapR}{\textsc{SwapRestricted}\xspace}
\newcommand{\uni}{\textsc{Unique}\xspace}

\newcommand{\Iconst}{\textsl{Constructive}\xspace}
\newcommand{\Idest}{\textsl{Destructive}\xspace}
\newcommand{\Iexact}{\textsl{Exact}\xspace}
\newcommand{\Iex}{\textsl{Exists}\xspace}
\newcommand{\Iuni}{\textsl{Unique}\xspace}

\newcommand{\Iswap}{\textsl{Swap}\xspace}
\newcommand{\Iadd}{\textsl{Add}\xspace}
\newcommand{\Idelete}{\textsl{Delete}\xspace}
\newcommand{\IaccD}{\textsl{DeleteAcceptability}\xspace}
\newcommand{\Ireor}{\textsl{Reorder}\xspace}
\newcommand{\IswapR}{\textsl{SwapRestricted}\xspace}
\newcommand{\Mst}{M^*}

\DeclareMathOperator{\ma}{ma}
\DeclareMathOperator{\rank}{rank}
\DeclareMathOperator{\bp}{bp}
\DeclareMathOperator{\dist}{dist}
\DeclareMathOperator{\id}{id}

\newtheorem{theorem}{Theorem}
\newtheorem{lemma}{Lemma}

\newtheorem{proposition}{Proposition}
\newtheorem{example}{Example}

\theoremstyle{remark}
\newtheorem*{remark}{Remark}

\DeclareMathOperator{\pend}{
    \succ 
    \overset{\raise0.3em\hbox{\text{\scriptsize{(rest)}}}}{\ldots}}

\newcommand{\bigO}{\ensuremath{\mathcal{O}}}

\tikzstyle{vertex}=[draw, circle, fill, inner sep = 2pt]
\tikzstyle{squared-vertex}=[draw, fill, inner sep = 2pt]
\usetikzlibrary{arrows.meta}

\usepackage{framed}
\usepackage{tabularx}

\newlength{\RoundedBoxWidth}
\newsavebox{\GrayRoundedBox}
\newenvironment{GrayBox}[1]
{\setlength{\RoundedBoxWidth}{.93\columnwidth}
	\def\boxheading{#1}
	\begin{lrbox}{\GrayRoundedBox}
		\begin{minipage}{\RoundedBoxWidth}}
		{   \end{minipage}
	\end{lrbox}
	\begin{center}
		\begin{tikzpicture}
		\node(Text)[draw=black!20,fill=white,rounded corners,inner sep=2ex,text 
		width=\RoundedBoxWidth]
		{\usebox{\GrayRoundedBox}};
		\coordinate(x) at (current bounding box.north west);
		\node [draw=white,rectangle,inner sep=3pt,anchor=north west,fill=white]
		at ($(x)+(6pt,.75em)$) {\boxheading};
		\end{tikzpicture}
\end{center}}

\newenvironment{defproblemx}[1]{\noindent\ignorespaces
	\FrameSep=6pt
	\parindent=0pt
	\begin{GrayBox}{#1}
		\begin{tabular*}{\columnwidth}{!{\extracolsep{\fill}}@{\hspace{.1em}} 
		>{\itshape} p{1.5cm} p{0.82\columnwidth} @{}}
		}{
		\end{tabular*}
	\end{GrayBox}
	\ignorespacesafterend
}

\newcommand{\defProblemQuestion}[3]{
	\begin{defproblemx}{#1}
		Input: & #2 \\
		Question: & #3
	\end{defproblemx}
}

\begin{document}

\title{\Large \bf {Bribery and Control in Stable Marriage}}

\author{Niclas~Boehmer, Robert~Bredereck\thanks{New affiliation: Institut f\"ur 
		Informatik, 
		Humboldt-Universit\"at zu 
		Berlin, Berlin, Germany.}, Klaus~Heeger, and 
	Rolf~Niedermeier}

\date{
	\small
	TU Berlin, Algorithmics and Computational Complexity\\
	\texttt{\{niclas.boehmer,robert.bredereck,heeger,rolf.niedermeier\}@tu-berlin.de}\\
}

\maketitle 
	
\begin{abstract}
    We initiate the study of external manipulations in 
    \textsc{Stable Marriage} by 
    considering several 
    manipulative actions as well as several 
    manipulation goals. 
    For instance, one goal is to make sure 
    that a given pair of agents is matched in a stable solution, and this 
    may be achieved by the manipulative action of reordering 
    some agents' preference lists. 
    We present a comprehensive study of the computational complexity of all 
    problems arising in this way. We find several polynomial-time 
    solvable 
    cases as well
    as NP-hard ones. For the NP-hard cases, focusing  on the natural 
    parameter ``budget'' (that is, the number of manipulative actions one is 
    allowed to perform),
    we also conduct a parameterized complexity analysis and
    encounter mostly 
parameterized hardness results.
\end{abstract}

\section{Introduction}
In the \textsc{Stable Marriage} problem, we have two sets of 
agents, 
each
agent has preferences over all agents from the other set, and the goal is to 
find a
matching between agents of the one set and agents of the other set such that
no two agents prefer each other to their assigned partners. In this paper, we 
study the manipulation of a \textsc{Stable Marriage} instance by an external 
agent 
that is able to change the set of agents or (parts of) their preferences. 
To motivate our studies, consider the following example.

Due to the high demand from students, Professor X decides to implement a 
central matching scheme using \textsc{Stable Marriage} to assign interested 
students to final-year projects 
offered by X's 
group: Every term X asks each of her group members to propose a 
project. Subsequently, projects are put online and interested
students are asked to submit their preferences over projects, while group 
members are asked to submit their preferences over the
students.
Afterwards, a stable matching of group members/projects 
and students is computed and implemented.\footnote{Notably, there is a 
wide spectrum of literature concerned 
	with finding a (stable) matching of students to final-year 
	projects/supervisors 
	(see the survey of 
	\citet[Chapter 
	5.5]{DBLP:books/ws/Manlove13}). A 
	central 
	mechanism for finding a (stable) matching of students to final-year 
	projects has been 
	implemented at the University of Glasgow \citep{DBLP:books/ws/Manlove13}, 
	the 
	University of York \citep{York}, and several other 
	universities \citep{hussain2019systematic,DBLP:journals/te/AnwarB03}.}  As 
	the submitted preferences are visible to all group members, two 
days before the deadline, one group member realizes that in the currently 
only stable matching he is not matched to the student with whom he already 
started thinking about his proposal. That is why, in order to be matched to his preferred student in at least 
some stable matching, he motivates two so far non-participating students that 
were unsure whether to do their final project this year or next year to register 
already this year. After this change, another group member notices that she is 
now 
matched to her least preferred choice in some of the matchings that are 
currently stable but she already has
an idea how to change this: She quickly 
visits two of her fellow group members and tells a story about how 
great two of the registered students performed in her least year's tutorial. In 
this way, 
she convinces them 
to rank these two students higher in their preferences. On the last day before 
the 
deadline, Professor X realizes that there currently exist multiple 
stable matchings. As X believes that this may cause unnecessary discussions 
on 
which stable matching to choose, X tells the two newest group members who 
anyway know only few of the students how they should change their 
preferences. 

 Another 
application of the non-bipartite version of \textsc{Stable Marriage} (called 
\textsc{Stable Roommates}), which might be 
easily susceptible to external manipulation arises in the context of P2P 
networks. 
This affects,
in particular, the BitTorrent protocol for content 
distribution \citep{DBLP:journals/corr/abs-cs-0612108,DBLP:conf/icdcs/GaiMMR07}.
Here, users rank each other based on some technical data, e.g., 
their download and upload bandwidth and based on the similarity of their 
interests. However, a user may easily add new users by simply 
entering the network multiple times with different accounts, thereby 
influencing 
the computed 
matching. 

Looking at further applications of \textsc{Stable Marriage}
and corresponding generalizations in the context of matching markets, there
is clear 
evidence of external manipulations in modern applications. For instance, 
surveys 
reported that in college 
admission systems in China, Bulgaria, Moldova, and Serbia, bribes have been 
performed 
in 
order to gain desirable admissions~\citep{doi:10.1086/524367,LIU2015104}.
Focusing on the most basic scenario \textsc{Stable Marriage}, we initiate 
a thorough study of manipulative actions (bribery and control) 
from a computational complexity perspective.
Notably, bribery  scenarios have also been used as a motivation in other 
papers
around \textsc{Stable Marriage}, e.g., when finding robust stable 
matchings~\citep{DBLP:conf/ec/ChenSS19} or when studying strongly stable 
matchings in 
the \textsc{Hospitals/Residents} problem
with ties~\citep{DBLP:conf/stacs/IrvingMS03}.

External manipulation may have many faces such as deleting agents,
adding agents, or changing agents' preference lists.
We consider three different manipulation goals and 
five different manipulative actions.

We introduce the manipulation goals \Iconst-\Iex, \Iexact-\Iex, 
and 
\Iexact-\Iuni,
where \Iconst-\Iex  is the least restrictive goal and asks
for modifications such that a desired agent pair is contained in some stable 
matching.
More restrictively, \Iexact-\Iex asks for modifications such that a 
desired matching is stable.
Most restrictively, \Iexact-\Iuni requires that a desired matching becomes the 
only
stable matching.

As manipulative actions, we investigate \Iswap, \Ireor, \IaccD, \Idelete, and 
\Iadd.
The actions \Iswap and \Ireor model bribery through an external agent. 
While a single \Ireor action allows to completely change the preferences of an 
agent (modeling a briber who can ``buy an agent''),
a \Iswap action is more fine-granular and only allows to swap two neighboring 
agents in some agent's preference list (modeling a briber who has to slightly
convince agents where the costs/effort to convince an agent of some 
preferences is larger if the preferences are further away from the agent's true 
preferences).
For both actions,  
the external agent might actually change the true preferences of the influenced 
agent, for example, by advertising some possible partner. However, in settings 
where the agents' preferences serve as an input for a centralized mechanism 
computing a stable matching which is subsequently implemented and cannot be 
changed, it is enough to bribe the agents to cast untruthful preferences.
\Idelete and \Iadd model control of the instance.
They are useful to model an external agent, i.e., the organizer of some matching 
system,
with the power to accept or reject agents to participate or to change the 
participation rules.
While a \Idelete (resp.\ \Iadd) action allows to delete (resp.\ add) an agent 
to the instance,
a \IaccD action forbids for a specific pair of agents the possibility to be 
matched to each other and to be blocking.
The latter can be seen as a hybrid between bribery and control because it can
model an external agent that changes acceptability rules 
or it can model a briber who convinces an agent
that another agent is unacceptable at some cost.

We conduct a complete analysis of all fifteen combinations of 
our five manipulative actions and three manipulation goals. Note, however, that 
for the two actions \Idelete and \Iadd the definition of \Iexact-\Iex and 
\Iexact-\Iuni introduced above cannot be directly applied, which is why we 
propose an adapted definition in \Cref{sub:goals}.

\paragraph*{Related Work.} 
Since its introduction by \citet{GaleS62}, \textsc{Stable
    Marriage} has been intensely studied by researchers from different
disciplines and in many 
contexts~(see, 
e.g., the 
surveys of \citet{DBLP:books/daglib/0066875,Knuth76,DBLP:books/ws/Manlove13}).

A topic related to manipulation in stable matchings is the study of 
\emph{strategic behavior}, which focuses on the question
whether agents can misreport their preferences
to fool a given matching algorithm to match them to a better partner. Numerous 
papers 
have addressed the question of
strategic behavior for different variants of computing stable matchings,
matching algorithms, types of agents' preferences and restrictions on the
agents that are allowed to misreport their preferences
(e.g., 
\citep{DBLP:conf/atal/AzizSW15,DBLP:journals/corr/abs-2012-04518,DBLP:journals/aamas/PiniRVW11,DBLP:journals/mor/Roth82,DBLP:journals/mansci/TeoST01,DBLP:conf/aaai/ShenTD18};
see the survey of \citet[Chapter~2.9]{DBLP:books/ws/Manlove13}). This 
setting is related to ours in the
sense that the preferences of agents are modified to achieve a desired outcome,
while it is fundamentally different with respect to the allowed modifications
and their
goal: In
the context of strategic behavior an agent is only willing to change \emph{its}
preferences if the agent directly benefits from it. 
Notably, in the context of strategic behavior, 
\citet{DBLP:conf/sigecom/Gonczarowski14} investigated how all agents from one 
side 
together can alter their preferences (including to declare some agents 
unacceptable) in order to achieve that a given matching is the unique stable 
matching.
This goal corresponds to our \Iexact-\Iuni setting. However, while 
\citet{DBLP:conf/sigecom/Gonczarowski14} allows for arbitrary changes in the 
preferences of all agents from one side including the deletion of the 
acceptability of agent 
pairs,
in the problems considered in this paper we always allow the briber to 
influence all agents and only either allow for reordering the preferences 
arbitrarily (\Ireor) 
or deleting the acceptability of agent pairs (\IaccD) (and aim to minimize the 
number of such manipulations).

Generally 
speaking, using our manipulative action \Iswap combined with an appropriate 
manipulation
goal, it is possible to model most computational problems where a single agent 
or a  
group of agents wants to fool an algorithm by misreporting their preferences: In 
order to do so, we set the budget such that the preferences of the agents from 
the cheating group can be modified arbitrarily. Further, we introduce dummy 
agents that are always matched among themselves in any stable matching (see
\Cref{lem:dummy-agents}). We add these dummy agents between each pair of 
consecutive agents 
in the preferences of all non-manipulating agents such that the budget is not 
sufficient to swap two non-dummy agents in the preferences of these agents. 
While this reduction shows how \Iswap can model computational problems related 
to strategic behavior of a group of agents, as already mentioned, the 
goals typically considered in 
the study of strategic behavior 
differ significantly from the ones studied in our paper. In the context of 
strategic behavior, the focus of a cheating group usually lies on 
misreporting their preferences such that all of them are matched to a better 
partner in a matching returned by a specific matching 
algorithm (as done, e.g., by 
\citet{DBLP:conf/atal/AzizSW15,DBLP:journals/aamas/PiniRVW11,DBLP:journals/mor/Roth82,DBLP:journals/mansci/TeoST01}),
 while our 
focus lies on making a specific pair or 
matching stable.

While we are interested in finding
ways to influence a profile to change the set of stable matchings, finding
\emph{robust} stable 
matchings~\citep{DBLP:conf/ec/ChenSS19,DBLP:conf/esa/MaiV18,DBLP:journals/corr/abs-1804-05537}
corresponds to finding stable matchings such that a briber cannot easily make 
the matching unstable.
For instance, \citet{DBLP:conf/ec/ChenSS19} introduced the concept of
\emph{$d$-robustness}: A matching is $d$-robust if it is stable in the given 
instance and remains stable even if~$d$~arbitrary swaps in preference lists are 
performed.
One motivation for their study of $d$-robust stable matchings is that $d$-robust stable matchings withstand bribers which may perform up to $d$ swaps.

Conceptually, our work is closely related to the study of bribery and
control in elections
(see the survey of \citet{DBLP:reference/choice/FaliszewskiR16}). 
In election control
problems~\citep{bartholdi1992hard},
the goal is to change the structure of a given election, e.g., by modifying the 
candidate or voter set, such that a
designated candidate becomes the winner/looser of the resulting election. In
bribery 
problems~\citep{DBLP:journals/jair/FaliszewskiHH09}, the briber is
allowed to modify the votes in the election to achieve the goal. Most of the
manipulative actions we consider
are inspired by either some control operation or bribery action
already studied in the context of voting. 

Our manipulation goals are also related to
problems previously studied in the stable matching literature: For example, the 
\Iconst-\Iex problem with given budget zero reduces to the
\textsc{Stable Pair} problem, which aims at deciding whether a given agent pair is contained in at least one
stable
matching. While the problem
is polynomial-time solvable for \textsc{Stable Marriage} instances without 
ties~\citep{Gusfield87}, deciding whether an agent pair is part of a 
``weakly stable'' matching is NP-hard if ties are 
allowed~\citep{DBLP:journals/tcs/ManloveIIMM02}.
This directly implies hardness of the \Iconst-\Iex problem if ties are
allowed even
when the budget is zero.
Similarly, deciding whether there exists a weakly stable matching \emph{not} 
containing a given agent pair is also NP-hard in the presence of 
ties~\citep{DBLP:journals/disopt/CsehH20}.
Moreover, several authors studied sufficient or necessary conditions for a 
\textsc{Stable Marriage} instance to admit a unique stable 
matching~\citep{Clark2006,CONSUEGRA2013468,MR3023042,EECKHOUT20001,10.1007/978-3-642-22427-0_4,Reny21}.

\paragraph*{Our Contributions.}

Providing a complete polynomial-time solvability vs.\ NP-hardness dichotomy,
we settle the computational complexity
of all problems emanating from our manipulation 
scenarios. 
We also conduct a parameterized
complexity analysis of these problems based on the budget parameter~$\ell$, 
that is, the number of elementary manipulative actions that we are allowed to perform.
At some places,
we also consider the approximability of our problems in polynomial time or FPT time.
For instance, we prove that 
\const-\ex-\swap does not admit
an~$\bigO(n^{1-\epsilon})$-approximation in~$f(\ell) n^{\bigO(1)}$ time for any 
$\epsilon >0$ and any computable function $f$ unless FPT=W[1].
\Cref{ta:sum} gives an overview of our results. 
Furthermore, for all problems we observe 
XP-algorithms\footnote{That is, polynomial-time algorithms if $\ell$~is constant.} 
with respect to
the parameter~$\ell$.
The 
    \const-\ex-\reor and 
    \exact-\uni-\reor problem require non-trivial algorithms to show this.
    
\setlength{\tabcolsep}{4.5pt}
\begin{table}[t!]
	\centering
	\resizebox{\linewidth}{!}{\begin{tabular}{l l l l}
			\toprule
			Action/Goal & \Iconst-\Iex & \Iexact-\Iex &
			\Iexact-\Iuni \\
			\midrule
			\Iswap & $\forall \epsilon > 0$: W[1]-hard wrt.~$\ell$ to approx. & 
			P (Th.
			\ref{th:ex-ex-swap}) & NP-c. (Pr. 
			\ref{th:ex-uni-swap})\\
			& within a factor of $\mathcal{O}(n^{1-\epsilon})$ (Th. 
			\ref{th:const-ex-swap})& & \\
			\midrule
			\Ireor & W[1]-h. wrt.~$\ell$ (Th. \ref{th:const-ex-reor})  & P (Pr.
			\ref{th:ex-ex-reor}) & W[2]-h. 
			wrt.~$\ell$
			(Th.
			\ref{th:ex-uni-reor})\\
			& 2-approx
			in
			P (Pr. \ref{th:const-ex-reor-2}) & &\\ \midrule
			\textsl{Delete} & W[1]-h. wrt.~$\ell$ (Th.
			\ref{th:const-ex-reor})& P (Ob.
			\ref{th:ex-ex-accD})&
			P (Th. \ref{thm:exact-uni-accD})\\
			\textsl{Accept.}& 
			& & \\ \midrule
			\Idelete & P (Th. \ref{const-ex-delete})& NP-c. (Pr.
			\ref{th:ex-ex-del}) & NP-c. (Pr.
			\ref{th:ex-ex-del})\\
			& & FPT wrt.
			$\ell$ (Pr.
			\ref{th:ex-ex-delFPT}) & \\
			\midrule
			\Iadd & W[1]-h. wrt.\ $\ell$ (Th.
			\ref{cor:const-ex-add})& P (Pr. \ref{th:ex-ex-add})
			&
			W[2]-h. wrt.~$\ell$
			(Pr.
			\ref{th:ex-uni-add})\\
			& NP-c.\ even if~$\ell=\infty$ (Th.
			\ref{cor:const-ex-add})  & & \\
			\bottomrule
	\end{tabular}} \caption{Overview of our results, where 
		$\ell$~denotes the given budget.
		All stated
		W[1]- and W[2]-hardness results also imply NP-hardness.  See 
		\Cref{se:prel} for definitions of parameterized complexity classes and 
		formal problem definitions.}
	\label{ta:sum}
\end{table}

We highlight the following five results and techniques. 
\begin{itemize}
        \item We develop a quite general framework for constructing 
        parameterized
        reductions from the W[1]-hard graph problem \textsc{Clique} to the 
\Iconst-\Iex problem and design the
        required gadgets for \Iadd, \IaccD, and \Ireor (\Cref{cor:const-ex-add,th:const-ex-reor}).
        \item We design a simple and efficient 
        algorithm for
        \const-\ex-\delete (\cref{const-ex-delete}), based on a non-trivial 
        analysis.
        \item We design a concise parameterized reduction from the W[2]-hard 
    problem \textsc{Hitting
        Set} to \exact-\uni-\reor. Surprisingly, in the constructed instance, to 
        make the target matching the unique stable matching, the preferences of 
some agents need to
        be reordered by swapping \emph{down} their (desired) partner in the given 
        matching~(\cref{th:ex-uni-reor}).
        \item We analyze how the manipulative actions \IaccD and \Ireor can be 
used
        to modify the so-called rotation poset to 
        make a given matching the unique stable
        matching (\cref{thm:exact-uni-accD,pr:ex-uni-reor-XP}).
        \item Our polynomial-time algorithms exhibit surprising connections 
between
        manipulations in \textsc{Stable Marriage} and the classical
        graph problems \textsc{Bipartite Vertex Cover}~(\cref{th:ex-ex-reor}),
        \textsc{Minimum Cut} (\cref{th:ex-ex-swap}), and \textsc{Weighted
        Minimum
        Spanning Arborescence} (\cref{thm:exact-uni-accD}).
\end{itemize}

Comparing the results for the different combinations of
manipulation goals and manipulative actions, we observe a quite
diverse complexity landscape: While for all other manipulative actions the 
corresponding problems are computationally hard, \const-\ex-\delete and 
\exact-\uni-\accD are polynomial-time solvable.
Relating the different manipulation goals to each other, we show that
specifying a complete matching that should be made stable instead of just one 
agent pair
that should be part of some stable matching
makes the problem of finding a successful manipulation significantly easier. 
In contrast to
this, providing even more information about the resulting instance by requiring
that
the given matching is the unique stable matching instead of just one of the
stable matchings makes the problem of finding a successful manipulation again 
harder.

From a high-level perspective, our computational hardness results can be seen as 
a shield against manipulative attacks. Of course, 
these shields are not unbreakable, as they only offer a worst-case protection 
against computing an attack of minimum cost. However, we slightly strengthen 
these (worst-case) shields 
by also 
proving parameterized hardness results for the parameter budget, 
which might be small compared to the number of agents especially in large 
matching markets. In contrast to this, our polynomial-time algorithms suggest 
that 
market makers shall be extra cautious in situations where the corresponding 
manipulative action can be easily performed. 

There also is a more positive interpretation of bribery and control
\citep{DBLP:conf/atal/FaliszewskiST17,DBLP:conf/ijcai/BoehmerBKL20}:
The minimum cost of a successful attack for \Iconst-\Iex can be 
interpreted as a measure for the ``distance from stability''  of the 
corresponding pair. Similarly, the minimum cost for \Iexact-\Iex can be 
interpreted as the ``distance from stability'' of the corresponding matching. 
Both metrics might be particularly interesting in applications where a central 
authority decides on a matching for which stability-related considerations are 
important but perfect stability is not vital. For instance, in our 
introductory example, Professor X might be dissatisfied with the matchings that 
are 
currently stable and therefore could use the ``distance from stability'' 
measure to decide between a few different matchings she deems acceptable. As 
\Iswap can be understood as the most fine-grained of our considered 
manipulative actions, in such situations it might be 
particularly 
appealing to use our polynomial-time algorithm for  \exact-\ex-\swap to compute 
the swap distance from stability of a matching.  Lastly, \Iexact-\Iuni offers a 
measure for the ``distance from 
unique stability'' of a matching. In practice, this distance could, for 
instance, serve as a 
tie-breaker between different stable matchings. 
However, there also exists 
a destructive view on bribery and control problems, where the goal is to 
prevent a given pair/matching from 
being stable. Destructive bribery can thus be interpreted as a distance measure 
from being unstable and can be used to quantify the robustness of the stability 
of a pair or a matching 
\citep{DBLP:conf/atal/FaliszewskiST17,DBLP:conf/ijcai/BoehmerBKL20,DBLP:journals/corr/abs-2010-09678}.
 While we focus on a constructive view, 
some of our results such as the polynomial-time algorithm for 
\const-\ex-\delete carry over to the destructive variant (see 
\Cref{subsub:const-reor}).

\paragraph*{Organization of the Paper.} 
\Cref{se:prel} delivers background on parameterized complexity analysis and the \textsc{Stable Marriage}
problem, and it formally defines the different considered manipulative actions 
and 
manipulation goals.
Afterwards, we devote one section to each of the manipulation goals we analyze.
In \Cref{se:const-ex}, we consider the \Iconst-\Iex setting for all
manipulative actions. We split this section into two parts; in the first part, 
we
prove several W[1]-hardness results, and in the second
part, we present a polynomial-time algorithm for the \Idelete action 
and a polynomial-time 
factor-2-approximation algorithm
for the \Ireor action. In \Cref{se:ex-ex}, we present our results for \Iexact-\Iex. After
considering the manipulative actions \Iswap, \Ireor, and \IaccD in the first
part of this section, we analyze the actions \Iadd and \Idelete in the second 
part. 
In \Cref{ex-uni},
we start by presenting hardness results for the \Iexact-\Iuni setting and then 
derive a polynomial-time for \IaccD as well as one XP algorithm for~\Ireor.
We conclude in \Cref{se:conc}, indicating directions for future research and 
presenting few very preliminary insights from experimental work with our 
algorithms.

\section{Preliminaries and First Observations} \label{se:prel}
In this section, we start by recapping some fundamentals of parameterized 
complexity theory (\Cref{sub:para}) and defining the \textsc{Stable Marriage} 
problem and related concepts (\Cref{sub:mar}). Subsequently, we introduce and 
formally define the five manipulative actions~(\Cref{sub:actions}) and three 
manipulation goals (\Cref{sub:goals}) we study. Lastly, in 
\Cref{se:relManip}, we make some first observations about the relationship of 
the different manipulative actions on a rather intuitive level.

\subsection{Parameterized Complexity} \label{sub:para}
A \emph{parameterized problem} consists of a problem instance~$\mathcal{I}$ 
and a (typically integer)
parameter value~$k$ (in our case the budget $\ell$).\footnote{To
simplify complexity-theoretic matters, by default parameterized problems are framed as decision problems. However, our positive algorithmic results easily extend to the corresponding optimization and search problems.}
It is called 
\emph{fixed-parameter tractable} with respect to
$k$ if it can be solved by an
\emph{FPT-algorithm}, i.e., an algorithm running in 
$f(k)|\mathcal{I}|^{O(1)}$ time for a computable function 
$f$. Moreover, it lies in XP with respect to~$k$ if it can be solved in $|\mathcal{I}|^{f(k)}$ time for some computable function 
$f$.
There is also a theory of hardness of parameterized
problems that includes the notion of W$[t]$-hardness with 
$\text{W}[t]\subseteq \text{W}[t']$ for $t\leq t'$. If a problem is
W[$t$]-hard for a given parameter for any $t\ge 1$, then it is widely believed
not to be
fixed-parameter tractable for this parameter. The usual approach to prove that 
a given parameterized problem is W[$t$]-hard is to describe a parameterized reduction
from a known W[$t$]-hard problem to it.
In our case, we only use the following special case of 
parameterized reductions: Standard many-one
reductions that run in polynomial time and ensure that the parameter of the 
output
instance is upper-bounded by a function of the parameter of the input
instance. 

\subsection{Stable Marriage} \label{sub:mar}
An instance $\mathcal{I}$ of the \textsc{Stable Marriage} (SM) problem consists
of a 
set
$U=\{m_1,\dots m_{n}\}$ of men and a set~$W=\{w_1,\dots,w_{n}\}$ of women, 
together
with a strict \emph{preference list}~$\mathcal{P}_a$ for each~$a\in U \cup 
W$.\footnote{We are well aware of the fact that \textsc{Stable Marriage} can be
	criticized for advocating and transporting outdated role models or 
	conservative marriage concepts. First, we emphasize that we use the old 
	concepts (men matching with women) for notational convenience and for
	being in accordance with the very rich, also recent literature. Second, we 
	remark that in real-world 
	applications \textsc{Stable Marriage} models general two-sided matching 
	markets, 
	which may appear in different scenarios such as matching students with 
	supervisors or matching mines with deposits.}
Note that following conventions from the literature and as this simplifies
discussions in some places, we 
assume that in all considered SM instances  $\mathcal{I}=(U, W, 
\mathcal{P})$, it holds that $|U|=|W|=n$ (for \Iadd, this 
means that after adding all agents to the instance the number of women and men 
is the same).\footnote{Notably, this assumption is crucial for our 
polynomial-time 2-approximation algorithm for \const-\ex-\reor in 
\Cref{subsub:const-reor}.}
Moreover, we call the elements of~$U \cup W$ \emph{agents} and 
$A=U\cup W$ denotes
the set of agents.
The preference list~$\mathcal{P}_a$ of an agent~$a$ is a strict order over the 
agents of the opposite gender.
We denote the preference list of an agent~$a\in A$ by~$a: a_1 \succ a_2 \succ 
a_3 
\succ \dots$, where~$a_1$ is~$a$'s most preferred agent,~$a_2$ is~$a$'s second 
most preferred agent, and so on. For the sake of readability, we sometimes only 
specify parts of the agents' preference relation and end the preferences with 
``$\pend$''. 
In this case, it is possible to complete the given profile by adding the 
remaining agents in an arbitrary order.
We say that~$a$ \emph{prefers}~$a'$ to~$a''$ if $a$ ranks $a'$ above $a''$ in 
its preference list, i.e.,~$a' \succ_a a''$. For two agents
$a,a'\in A$ of opposite gender, let~$\rank(a,a')$ denote the rank of~$a'$ in
the preference relation of $a$, i.e., one plus the number of agents which~$a$ prefers 
to~$a'$.

A \emph{matching}~$M$ is a set of pairs $\{m, w\}$ with $m \in U$ and $w\in W$ 
such that each agent is
contained 
in at most one pair.
An agent is \emph{assigned} in a matching~$M$ if some pair of~$M$ contains this agent, and \emph{unassigned} otherwise.
For a matching~$M$ and an assigned agent~${a\in A}$, we denote by~$M(a)$ the 
agent~$a$ 
is matched to in~$M$, i.e.,~$M(a) = a'$ if~${\{a, a'\}\in M}$.
We slightly abuse notation and write~$a\in M$ for an agent~$a$ if there exists 
some agent~$a'$ such that~$\{a, a'\}\in M$.
A matching is called \emph{complete} if no agent is unassigned.
For a matching~$M$, a pair~$\{m, w\}$ with $m\in U$ and $w\in W$ is \emph{blocking} if 
both~$m$ is 
unassigned or prefers~$w$ to~$M(m)$ and~$w$ is unassigned or prefers~$m$ 
to~$M(w)$.
A matching is \emph{stable} if it does not admit a blocking pair. We denote the set of stable matchings in an SM
instance~$\mathcal{I}$ by
$\mathcal{M}_{\mathcal{I}}$. Given a matching~$M$ and some 
subset~$A'\subseteq A$ of agents, we denote by~$M|_{A'}$ the restriction of~$M$ 
to~$A'$, i.e.,~$M|_{A'}=\{\{u,w\}\in M: u,w\in A'\}$.
A stable matching~$M$ is called \emph{man-optimal} if for every man $m\in U$ 
and every stable matching $M'$ it holds that $m$ does not prefer $M' (m)$ to 
$M(m)$.
Symmetrically, a stable matching~$M$ is called \emph{woman-optimal} if for every 
woman~$w\in W$ and every stable matching~$M'$ it holds that $w$ does not 
prefer~$M' (w)$ to $M(w)$.
\citet{GaleS62} showed that a man-optimal as well as a woman-optimal matching 
always exist.

The \textsc{Stable Marriage with Incomplete 
    Lists} (\textsc{SMI}) problem is a generalization of the \textsc{Stable 
    Marriage} 
problem where each agent~$a$ is allowed to specify 
incomplete preferences of agents of the 
opposite gender.
Then, a pair of agents~$\{m,w \}$ with $m\in U$ and $w\in W$
can only be part of a stable matching~$M$ if they both appear in each others'
preference list.
A pair~$\{m,w \}$ with $m\in U$ and $w\in W$ is blocking if $m$ and $w$ appear on each other's preferences and 
both~$m$ is 
unassigned or prefers~$w$ to~$M(m)$, and~$w$ is unassigned or prefers~$m$ 
to~$M(w)$.
Let~$\ma(M)$ denote the set of agents matched in a stable matching~$M$. 
Note that by the Rural Hospitals Theorem~\citep{roth1986allocation} 
it holds for all stable matchings~$M,M'\in \mathcal{M}_{\mathcal{I}}$ that~
$\ma(M)=\ma(M')$.
Moreover, for an \textsc{SMI} instance $\mathcal{I}$, let~$\ma(\mathcal{I})$ 
denote the set of agents that are matched
in a stable matching in~$\mathcal{I}$.

\subsection{Manipulative Actions} \label{sub:actions}
We introduce five different manipulative actions and 
necessary notation in this subsection. We denote by~$\mathcal{X}\in$
$\{\Iswap,\Ireor,\IaccD,\Idelete,\Iadd\}$ the type of a manipulative action.

\medskip
\noindent \textbf{Swap.} A \Iswap operation changes the order of two 
neighboring 
agents in the 
preference list of an agent.

\begin{example}
Let~$a$ be an agent, and let its preference list be~$a: a_1 \succ a_2\succ 
a_3$.
There are two possible (single)
swaps:
Swapping~$a_1$ and~$a_2$, resulting in~$a: a_2 \succ a_1 \succ a_3$, and 
swapping~$a_2$ and~$a_3$, resulting in~$a: a_1 \succ a_3 \succ a_2$.
\end{example}

\medskip
\noindent \textbf{Reorder.} A \Ireor operation of an agent's preference list 
reorders its preferences 
arbitrarily, i.e., one performs an arbitrary permutation.

\begin{example}
For an agent~$a$ with preference list~$a: a_1 \succ a_2\succ a_3$, there 
are six possible reorderings, resulting in one of the six 
possible strict total orders over $\{a_1,a_2,a_3\}$.
\end{example}

\medskip
\noindent \textbf{Delete Acceptability.} A \IaccD
    operation is understood as deleting the mutual acceptability of a 
    man and a woman.
    This enforces that such a deleted pair cannot be part of any
    stable matching and cannot be a blocking pair for any stable matching. 
    Thus, after applying a \IaccD action, the given SM instance is transformed 
    into an SMI instance. For two agents $a,a'\in A$, we sometimes also say 
that we delete the pair or edge $\{a,a'\}$ if we delete the mutual 
acceptability of the two agents $a$ and $a'$. 

\begin{example}
Let~$m$ be a man with preferences~$m : w_1 \succ w_2 \succ w \succ w_3$, 
and~$w$ be a woman with preferences~$w: m \succ m_1 \succ m_2 \succ m_3$.
Deleting the pair~$\{m, w\}$ results in the following preferences:\\
~$m : w_1 \succ w_2 \succ w_3$ \hspace{0.5cm} and \hspace{0.5cm}~$w: m_1 \succ 
m_2 \succ m_3$.
\end{example}

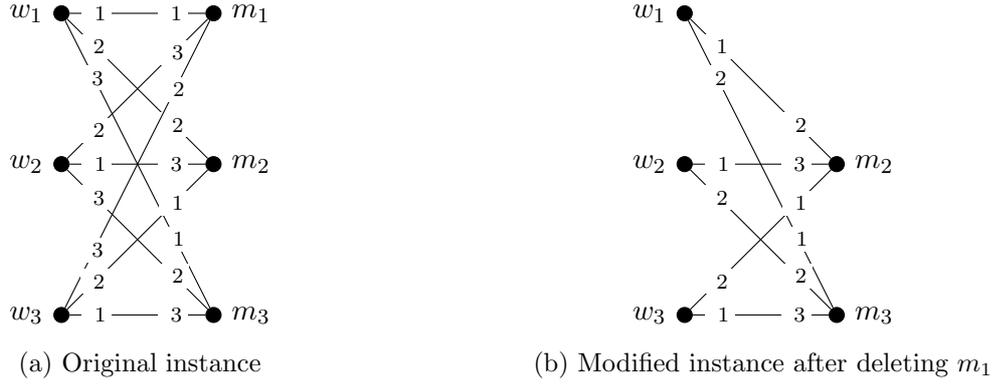
\begin{figure}[t]
    \captionsetup[subfigure]{justification=centering}
    \begin{subfigure}[b]{.45\textwidth}
        \centering
        \begin{tikzpicture}[yscale=2]
        \node[vertex, label=180:$w_1$] (w1) at (0, 0) {};
        \node[vertex, label=180:$w_2$] (w2) at (0, -1) {};
        \node[vertex, label=180:$w_3$] (w3) at (0, -2) {};
        
        \node[vertex, label=0:$m_1$] (m1) at (2, 0) {};
        \node[vertex, label=0:$m_2$] (m2) at (2, -1) {};
        \node[vertex, label=0:$m_3$] (m3) at (2, -2) {};
        
        \draw (w1) edge node[pos=0.2, fill=white, inner sep=2pt] 
        {\scriptsize~$1$}  node[pos=0.76, fill=white, inner sep=2pt] 
        {\scriptsize~$1$} (m1);
        \draw (w1) edge node[pos=0.2, fill=white, inner sep=2pt] 
        {\scriptsize~$2$}  node[pos=0.76, fill=white, inner sep=2pt] 
        {\scriptsize~$2$} (m2);
        \draw (w1) edge node[pos=0.2, fill=white, inner sep=2pt] 
        {\scriptsize~$3$}  node[pos=0.76, fill=white, inner sep=2pt] 
        {\scriptsize~$1$} (m3);
        \draw (w2) edge node[pos=0.2, fill=white, inner sep=2pt] 
        {\scriptsize~$2$}  node[pos=0.76, fill=white, inner sep=2pt] 
        {\scriptsize~$3$} (m1);
        \draw (w2) edge node[pos=0.2, fill=white, inner sep=2pt] 
        {\scriptsize~$1$}  node[pos=0.76, fill=white, inner sep=2pt] 
        {\scriptsize~$3$} (m2);
        \draw (w2) edge node[pos=0.2, fill=white, inner sep=2pt] 
        {\scriptsize~$3$}  node[pos=0.76, fill=white, inner sep=2pt] 
        {\scriptsize~$2$} (m3);
        \draw (w3) edge node[pos=0.2, fill=white, inner sep=2pt] 
        {\scriptsize~$3$}  node[pos=0.76, fill=white, inner sep=2pt] 
        {\scriptsize~$2$} (m1);
        \draw (w3) edge node[pos=0.2, fill=white, inner sep=2pt] 
        {\scriptsize~$2$}  node[pos=0.76, fill=white, inner sep=2pt] 
        {\scriptsize~$1$} (m2);
        \draw (w3) edge node[pos=0.2, fill=white, inner sep=2pt] 
        {\scriptsize~$1$}  node[pos=0.76, fill=white, inner sep=2pt] 
        {\scriptsize~$3$} (m3);
        \end{tikzpicture}  
        \caption{Original instance}
        \label{fig:example-delete-agent-1}
    \end{subfigure}
    \hfill
    \begin{subfigure}[b]{.45\textwidth}
        \centering
        \begin{tikzpicture}[yscale=2]
        \node[vertex, label=180:$w_1$] (w1) at (0, 0) {};
        \node[vertex, label=180:$w_2$] (w2) at (0, -1) {};
        \node[vertex, label=180:$w_3$] (w3) at (0, -2) {};
        
        \node[vertex, label=0:$m_2$] (m2) at (2, -1) {};
        \node[vertex, label=0:$m_3$] (m3) at (2, -2) {};
        
        \draw (w1) edge node[pos=0.2, fill=white, inner sep=2pt] 
        {\scriptsize~$1$}  node[pos=0.76, fill=white, inner sep=2pt] 
        {\scriptsize~$2$} (m2);
        \draw (w1) edge node[pos=0.2, fill=white, inner sep=2pt] 
        {\scriptsize~$2$}  node[pos=0.76, fill=white, inner sep=2pt] 
        {\scriptsize~$1$} (m3);
        \draw (w2) edge node[pos=0.2, fill=white, inner sep=2pt] 
        {\scriptsize~$1$}  node[pos=0.76, fill=white, inner sep=2pt] 
        {\scriptsize~$3$} (m2);
        \draw (w2) edge node[pos=0.2, fill=white, inner sep=2pt] 
        {\scriptsize~$2$}  node[pos=0.76, fill=white, inner sep=2pt] 
        {\scriptsize~$2$} (m3);
        \draw (w3) edge node[pos=0.2, fill=white, inner sep=2pt] 
        {\scriptsize~$2$}  node[pos=0.76, fill=white, inner sep=2pt] 
        {\scriptsize~$1$} (m2);
        \draw (w3) edge node[pos=0.2, fill=white, inner sep=2pt] 
        {\scriptsize~$1$}  node[pos=0.76, fill=white, inner sep=2pt] 
        {\scriptsize~$3$} (m3);
        \end{tikzpicture} 
        \caption{Modified instance after deleting $m_1$}
        \label{fig:example-delete-agent-2}
    \end{subfigure}
    \caption{Visualization of \Cref{ex:Del} for a \Idelete operation.
    The numbers on the edges encode the preferences of the agents:
    The number~$i$ of an edge closer to an agent $a$ means that $a$ considers
    the other endpoint of this edge as its $i$-th best partner.}
    \label{fig:example-delete-agent}
\end{figure}

\medskip
\noindent \textbf{Delete.} A \Idelete operation deletes an agent from the 
instance. Note that we allow for deleting an unequal number of men and women. 

\begin{example} \label{ex:Del}
Let $\mathcal{I}$ be an SM instance with agents~$U= \{m_1, m_2, m_3\}$ and
~$W =\{w_1, w_2, w_3\}$, and the following preferences:
\begin{align*}
m_1 & : w_1 \succ w_3 \succ w_2, & m_2& : w_3 \succ w_1 \succ w_2, & m_3 & : 
w_1 \succ w_2 \succ w_3,\\
w_1 & : m_1 \succ m_2 \succ m_3, & w_2&: m_2 \succ m_1 \succ m_3, & w_3 &: 
m_3 
\succ m_2 \succ m_1.
\end{align*}
The instance 
$\mathcal{I}$ is visualized in \Cref{fig:example-delete-agent-1}.
Deleting agent~$m_1$ results in the instance~$\mathcal{I}'=(U', W', 
\mathcal{P}')$ with~$U' = \{m_2, m_3\}$,~$W' = \{w_1, w_2, w_3\}$, and 
preferences:
\begin{align*}
m_2& : w_3 \succ w_1 \succ w_2, & m_3 & : 
w_1 \succ w_2 \succ w_3,& &\\
w_1&: m_2 \succ m_3, & w_2&: m_2 \succ m_3, & w_3&: m_3 \succ m_2.
\end{align*}
The modified instance is visualized in \Cref{fig:example-delete-agent-2}.
\end{example}
Given an SM 
    instance~$\mathcal{I}=(U,W,\mathcal{P})$
    and 
    a subset of agents~$A'\subseteq A$, we define~$\mathcal{I}\setminus A'$ to
    be the instance that results from deleting the 
    agents~$A'$ 
    from~$\mathcal{I}$.
    
\medskip

\noindent \textbf{Add.} An \Iadd operation adds an agent from a predefined set 
    of agents to the 
    instance.
    Formally, the input for a computational problem considering the 
    manipulative 
    action \Iadd consists 
    of an SM instance~$(U, W, \mathcal{P})$ together with two 
    subsets~$U_{\addag}
    \subseteq
    U$ and~$W_{\addag} \subseteq W$. The sets $U_{\addag}$ and $W_{\addag}$ contain 
    agents 
    that are not initially present and can be added to the original instance. 
    All other men $U_{\orig}:=U\setminus 
    U_{\addag}$ and women  $W_{\orig}:=W\setminus W_{\addag}$ are already 
    initially present and part of the 
    original instance.
    Adding a set of agents~$X_A=X_U \cup X_W~$ with~$X_U \subseteq U_{\addag}$ 
    and~$X_W 
    \subseteq W_{\addag}$ results in the instance~$(U_{\orig} \cup X_U, 
    W_{\orig} 
    \cup X_W, 
    \mathcal{P}')$, where~$\mathcal{P}'$ is the restriction of~$\mathcal{P}$ to 
    agents from~$U_{\orig} \cup X_U\cup W_{\orig} \cup X_W$. Note that as 
    mentioned in \Cref{sub:mar}, we
    require that $|U|=|W|$ but do not impose constraints on~$|U_{\addag}|$ and 
    $|W_{\addag}|$ or whether the same number of men and women is added to the 
    instance by the manipulation.

\begin{example}
    An example for the application of an \Iadd operation is depicted in 
    \Cref{fig:example-add-agent}. The instance consists of three 
    men~$U= \{m_1, m_2, m_3\}$ and three women~$W =\{w_1, w_2, w_3\}$ 
    with~$U_{\addag} 
    = \{m_2,m_3\}~$ and~$W_{\addag}=\{w_2,w_3\}$. That is, only the agents 
    $m_1$ and $w_1$ are 
    initially present in the instance and all other agents can be added by a 
    manipulative action. After adding~$X_U = \{m_2\}$ and~$X_W = \{w_2, w_3\}$, 
    the agent sets change to~$U^*= \{m_1, m_2\}$ and~$W^* =\{w_1, w_2, w_3\}$, 
    resulting in the preference profile shown in 
    \Cref{fig:example-add-agent-2}.
\end{example}
\begin{figure}[t]
	\captionsetup[subfigure]{justification=centering}
	\begin{subfigure}[b]{.45\textwidth}
		\centering
		\begin{tikzpicture}[yscale=2]
		\node[vertex, label=180:$w_1$] (w1) at (0, 0) {};
		\node[vertex, label=180:$w_2$, gray] (w2) at (0, -1) {};
		\node[vertex, label=180:$w_3$, gray] (w3) at (0, -2) {};
		
		\node[fit=(w1), circle, draw, inner sep = 0.4 cm, dashed, 
		label=90:$W_{\orig}$] (c) {};
		
		\node[vertex, label=0:$m_1$] (m1) at (2, 0) {};
		\node[vertex, label=0:$m_2$, gray] (m2) at (2, -1) {};
		\node[vertex, label=0:$m_3$, gray] (m3) at (2, -2) {};
		
		\node[fit=(m1), circle, draw, inner sep = 0.4 cm, dashed, 
		label=90:$U_{\orig}$] (d) {};
		
		\draw (w1) edge node[pos=0.2, fill=white, inner sep=2pt] 
		{\scriptsize~$1$}  node[pos=0.76, fill=white, inner sep=2pt] 
		{\scriptsize~$1$} (m1);
		\draw (w1) edge[gray] node[pos=0.2, fill=white, inner sep=2pt] 
		{\scriptsize~$2$}  node[pos=0.76, fill=white, inner sep=2pt] 
		{\scriptsize~$2$} (m2);
		\draw (w1) edge[gray] node[pos=0.2, fill=white, inner sep=2pt] 
		{\scriptsize~$3$}  node[pos=0.76, fill=white, inner sep=2pt] 
		{\scriptsize~$1$} (m3);
		\draw (w2) edge[gray] node[pos=0.2, fill=white, inner sep=2pt] 
		{\scriptsize~$2$}  node[pos=0.76, fill=white, inner sep=2pt] 
		{\scriptsize~$3$} (m1);
		\draw (w2) edge[gray] node[pos=0.2, fill=white, inner sep=2pt] 
		{\scriptsize~$1$}  node[pos=0.76, fill=white, inner sep=2pt] 
		{\scriptsize~$3$} (m2);
		\draw (w2) edge[gray] node[pos=0.2, fill=white, inner sep=2pt] 
		{\scriptsize~$3$}  node[pos=0.76, fill=white, inner sep=2pt] 
		{\scriptsize~$2$} (m3);
		\draw (w3) edge[gray] node[pos=0.2, fill=white, inner sep=2pt] 
		{\scriptsize~$3$}  node[pos=0.76, fill=white, inner sep=2pt] 
		{\scriptsize~$2$} (m1);
		\draw (w3) edge[gray] node[pos=0.2, fill=white, inner sep=2pt] 
		{\scriptsize~$2$}  node[pos=0.76, fill=white, inner sep=2pt] 
		{\scriptsize~$1$} (m2);
		\draw (w3) edge[gray] node[pos=0.2, fill=white, inner sep=2pt] 
		{\scriptsize~$1$}  node[pos=0.76, fill=white, inner sep=2pt] 
		{\scriptsize~$3$} (m3);
		\end{tikzpicture}
		\caption{Original instance with only agents $w_1$ and $m_1$ present 
		(to which agents $w_2$, $m_2$, $w_3$, and $m_3$ can be added).}
		\label{fig:example-add-agent-1}
	\end{subfigure}
	\hfill
	\begin{subfigure}[b]{.45\textwidth}
		\centering
		\begin{tikzpicture}[yscale=2]
		\node[vertex, label=180:$w_1$] (w1) at (0, 0) {};
		\node[vertex, label=180:$w_2$] (w2) at (0, -1) {};
		\node[vertex, label=180:$w_3$] (w3) at (0, -2) {};
		
		\node[vertex, label=0:$m_1$] (m1) at (2, 0) {};
		\node[vertex, label=0:$m_2$] (m2) at (2, -1) {};
		
		\draw (w1) edge node[pos=0.2, fill=white, inner sep=2pt] 
		{\scriptsize~$1$}  node[pos=0.76, fill=white, inner sep=2pt] 
		{\scriptsize~$1$} (m1);
		\draw (w1) edge node[pos=0.2, fill=white, inner sep=2pt] 
		{\scriptsize~$2$}  node[pos=0.76, fill=white, inner sep=2pt] 
		{\scriptsize~$2$} (m2);
		\draw (w2) edge node[pos=0.2, fill=white, inner sep=2pt] 
		{\scriptsize~$2$}  node[pos=0.76, fill=white, inner sep=2pt] 
		{\scriptsize~$3$} (m1);
		\draw (w2) edge node[pos=0.2, fill=white, inner sep=2pt] 
		{\scriptsize~$1$}  node[pos=0.76, fill=white, inner sep=2pt] 
		{\scriptsize~$3$} (m2);
		\draw (w3) edge node[pos=0.2, fill=white, inner sep=2pt] 
		{\scriptsize~$2$}  node[pos=0.76, fill=white, inner sep=2pt] 
		{\scriptsize~$2$} (m1);
		\draw (w3) edge node[pos=0.2, fill=white, inner sep=2pt] 
		{\scriptsize~$1$}  node[pos=0.76, fill=white, inner sep=2pt] 
		{\scriptsize~$1$} (m2);
		\end{tikzpicture}
		\caption{Modified instance after adding $m_2$, $w_2$, and $w_3$.}
		\label{fig:example-add-agent-2}
	\end{subfigure}
	\caption{Example for the application of manipulative action \Iadd.}
	\label{fig:example-add-agent}
\end{figure}
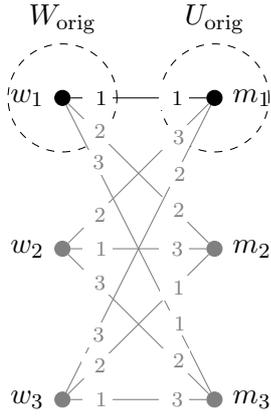
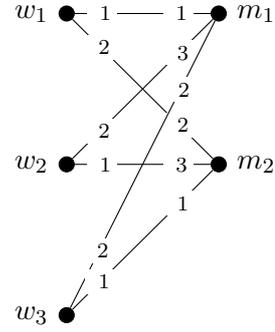

\subsection{Manipulation Goals} \label{sub:goals}
In the \Iconst-\Iex setting, the goal is to modify a given SM instance using
manipulative actions of some given type such that a designated man-woman pair 
is 
part of some stable matching.
For $\mathcal{X}\in \{\Iswap, \IaccD, \Idelete, \Iadd\}$, the formal definition 
of the problem is presented below. 
\defProblemQuestion{\const-\ex-$\mathcal{X}$}
{Given an SM instance~$\mathcal{I}=(U,W,\mathcal{P})$, a man-woman
    pair~$\{m^*,w^*\}$, and a budget~$\ell\in \mathbb{N}$.}
{Is it possible to perform at most~$\ell$ manipulative actions of type~$\mathcal{X}$
    such that 
    $\{m^*,w^*\}$ is part of at least one matching that is stable in the altered 
    instance?}

If one applies the above problem formulation to \Ireor, then the resulting 
problem always allows a trivial solution of size two by reordering the 
preferences of $m^*$ and~$w^*$ such that they are each other's top-choice.
Hence, for \const-\ex-\reor, we forbid the reordering of the preferences of 
$m^*$ and 
$w^*$, 
resulting in the
following problem formulation.
\defProblemQuestion{\const-\ex-\reor}
{Given an SM instance~$\mathcal{I}=(U,W,\mathcal{P})$, a man-woman
pair~$\{m^*,w^*\}$, and a budget~$\ell\in \mathbb{N}$.}
{Is it possible to perform at most~$\ell$ reorderings of the preferences of agents other 
than~$m^*$ and $w^*$
such that
$\{m^*,w^*\}$ is part of at least one stable matching in the altered instance?}

\medskip
In the \Iexact setting, in contrast to the \Iconst setting, we are given a 
complete 
matching. Within this setting, we consider two different computational 
problems. 
First, we consider the \exact-\ex problem where the goal is to modify a given SM
instance such that the given matching is stable in the instance. Second, we 
consider the \exact-\uni problem where the goal is to modify a given 
SM instance such that the given matching is the \emph{unique} stable 
matching.   
\defProblemQuestion{\exact-\ex (\uni)-$\mathcal{X}$}
{Given an SM instance~$\mathcal{I}=(U,W,\mathcal{P})$, a complete
    matching~$\Mst$, and budget~$\ell\in \mathbb{N}$.}
{Is it possible to perform at most~$\ell$ manipulative actions of type~$\mathcal{X}$ 
    such that $\Mst$ is a (the unique) stable matching in the altered instance?}
    
    For manipulative actions \Idelete and \Iadd, the definitions of 
\exact-\ex-$\mathcal{X}$ and \exact-\uni-$\mathcal{X}$ are not 
directly 
applicable, as the set of agents changes by applying \Idelete or \Iadd 
operations. That is 
why, for these actions, we need to slightly adapt the definitions from above. 
The general idea 
behind the proposed adaption is that we specify a complete matching on all 
agents (including those from $U_{\addag}\cup W_{\addag}$ for \Iadd), and 
require that the restriction of the specified matching~$M^*$ to the agents 
contained in the manipulated instance 
should be the (unique) stable matching in the manipulated instance. The reasoning 
behind this is that we only want to allow pairs that we approve to be part of 
a (the) stable matching. This results in the following definition for \exact-\ex(\uni)-\delete: 
\defProblemQuestion{\exact-\ex(\uni)-\delete}
{Given an SM instance~$\mathcal{I}=(U,W,\mathcal{P})$, a complete
matching~$\Mst$, and budget~$\ell\in \mathbb{N}$.}
{Is it possible to delete at most~$\ell$ agents from~$U\cup W$ such that there exists some 
$M'\subseteq \Mst$ which is a
(the unique) stable matching in the altered instance?}
Similarly, we get the following definition for \exact-\ex(\uni)-\add:
\defProblemQuestion{\exact-\ex(\uni)-\add}
{Given an SM instance~$\mathcal{I}=(U,W,\mathcal{P})$ together
with subsets~$U_{\addag}\subseteq U$ and~$W_{\addag}\subseteq W$, a complete
matching~$\Mst$, and budget~$\ell\in \mathbb{N}$.}
{Is it possible to add at most~$\ell$ agents~$X_A \subseteq U_{\addag}\cup W_{\addag}$ 
such that there exists some 
$M'\subseteq \Mst$ which is a
(the unique) stable matching in the altered instance?}

There also exist natural optimization 
variants of all considered decision problems which ask for the minimum number 
of manipulative actions that are 
necessary to alter a given SM instance to achieve the specified goal.
The (in-)approximability results in \Cref{th:const-ex-swap,th:const-ex-reor-2} refer to the optimization variants of these problems.

\subsection{Relationship Between Different Manipulative Actions} 
\label{se:relManip} 
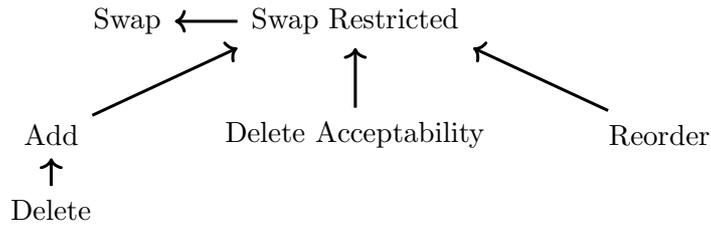
\begin{figure}[bt]
    \begin{center}
      \begin{tikzpicture}
        \node[inner sep=5pt] (sw)  at (3,4.5) {Swap Restricted};
        \node[inner sep=5pt] (usw)  at (0,4.5) {Swap};
        \node[inner sep=5pt] (dA)  at (3,3) {Delete Acceptability};
        \node[inner sep=5pt] (re)  at (7,3) {Reorder};
        \node[inner sep=5pt] (A)  at (-1,3) {Add};
        \node[inner sep=5pt] (D)  at (-1,2) {Delete};
        
        \draw[draw=black, line width=0.4mm,->]     (D) edge node  {} (A);
        \draw[draw=black, line width=0.4mm,->]     (A) edge node  {} 
(sw.south west);
        \draw[draw=black, line width=0.4mm,->]     (dA) edge node  {} (sw);
        \draw[draw=black, line width=0.4mm,->]     (re) edge node  {} (sw.south 
east);
        \draw[draw=black, line width=0.4mm,->]     (sw) edge node  {} (usw);
      \end{tikzpicture}

    \end{center}
    \caption{Relationship between manipulative actions as 
described in \Cref{se:relManip}. An arc from an action $\mathcal{X}$ to an 
action $\mathcal{Y}$ indicates that we present a ``simple'' 
description how to model action $\mathcal{X}$ by action 
$\mathcal{Y}$. Note that an arc does not imply that a computational 
problem for $\mathcal{X}$ is always reducible to the same problem for 
$\mathcal{Y}$.}\label{fig:aov}
  \end{figure}

In this section, to get an overview of the different manipulative actions, we 
analyze how they relate to each other.
This section, however, does not aim at introducing formal relationships in the
sense of a general notion of reducibility of 
two manipulative actions~$\mathcal{X}$ and $\mathcal{Y}$, as different manipulation goals require different properties of a reduction. 
Instead, we present high-level ideas how it is
possible to simulate one action with another action.
We mainly aim at giving an 
intuition for the relationships between the actions that helps to 
understand, 
relate, and classify the results we present in the paper. Note that most of the 
sketched relationships are applicable when relating computational 
problems around the \Iconst-\Iex goal for the different manipulative actions to 
each other (and not so much for \Iexact-\Iex and 
\Iexact-\Iuni). For an overview of 
the relations see \Cref{fig:aov}.

\medskip
\noindent \textbf{\Idelete via \Iadd.}
It is possible to model \Idelete actions by \Iadd actions by adapting the 
considered SM instance as follows.  We keep all agents $a\in A$ from the
original instance and introduce for each of them a new designated 
\textit{binding 
    agent}~$a'$ of opposite gender that ranks~$a$ first and all other agents 
in an arbitrary ordering afterwards. Moreover, we add~$a'$ at the 
first position of the preferences of~$a$. The set of 
agents 
that can be added to the instance are the binding agents.\\ 
Then, adding
the binding agent~$a'$ in the modified instance corresponds to deleting the
corresponding non-binding agent~$a$ in the original instance, as in this case 
the
agent and its binding agent 
are their mutual top-choices and thereby always matched in a stable matching.

\medskip
\noindent \textbf{Restricted \Iswap via \Iswap.} 
Before we describe how different manipulative actions can be modeled by \Iswap 
actions, we first sketch how it is possible to model a variant of \Iswap 
where for a given set of pairs consisting of adjacent positions swapping 
agents on each 
such pair of neighboring positions is forbidden and swapping 
the first and 
the second element in the preference relation of an agent may have some 
specified non-unit cost. To model this variant, we
introduce $n(\ell+1)$ 
dummy men and $n(\ell+1)$ dummy 
women each ranking all dummy agents of the opposite gender before
the other agents. Due to the preferences of the dummy agents,
regardless of which $\ell$ swaps are performed in the preferences of dummy 
agents,
in all stable matchings all dummy agents are matched to dummy agents and a 
dummy agent is 
never part of a blocking pair together with a non-dummy agent (see 
\Cref{lem:dummy-agents}). Now, for each
preference 
list of a non-dummy agent $a$, if we want to restrict that agent~$a'$  at
rank $i$ cannot be swapped with agent~$a''$ ranked directly behind~$a'$ in the 
preferences of~$a$, we place the 
$\ell+1$ dummy agents with indices $(i-1)(\ell+1)+1$ to $i(\ell+1)$ of opposite 
gender between~$a'$ and~$a''$ in the preference list of~$a$. Thus, the given 
budget never suffices to swap agents $a'$ and~$a''$ in the preference 
list of~$a$. Moreover, if we want to introduce a non-unit cost $2\leq c\leq \ell$ of 
swapping 
the first and the second agent in some agent's preference list, then we put the 
dummy agents $1$ to $c-1$ of opposite gender between its most preferred and 
second-most preferred agent in its preference list. 
We use this variant of \Iswap (and a more restrictive version of it) in the 
proofs of 
\Cref{th:const-ex-swap} and \Cref{th:ex-uni-swap}.

\medskip
\noindent \textbf{\Iadd via Restricted \Iswap.}
It is possible to model \Iadd actions by (restricted) \Iswap actions by 
modifying a
given SM instance $(A=U\cup W, \mathcal{P})$ with sets $U_{\addag}\subseteq U$ 
and ${W_{\addag}\subseteq W}$ as follows. We keep all agents 
$a\in A$ from
the original instance and introduce for each agent~$a\in  U_{\addag}\cup 
W_{\addag}$ one agent~$a'$ of
the
opposite gender and one agent~$a''$ of the same gender with preferences 
$a':a\succ a'' \pend $ and $a'': a' \pend$. Moreover, we put for each~$a\in
 U_{\addag}\cup W_{\addag}$ agent~$a'$ as the top-choice of~$a$.
All preference lists are completed by appending the remaining agents in an arbitrary order at the end.
Now, we introduce dummy agents such that the only allowed swaps 
are swapping $a$ with $a''$ in $a'$'s preference list for some~$a\in 
U_{\addag}\cup W_{\addag}$. Not adding an agent~$a\in 
 U_{\addag}\cup W_{\addag}$ corresponds to leaving the preferences of $a'$ 
unchanged, which 
results in $\{a',a\}$ being part of every stable matching. Adding an agent~$a\in 
 U_{\addag}\cup W_{\addag}$ corresponds to modifying the preferences of~$a'$ by 
swapping~$a$ 
and~$a''$, which results in~$\{a',a''\}$ being part of every stable matching. 
In this case, $a$ is now able to pair up with other agents from the original 
instance. We prove the correctness of this transformation in 
\Cref{cor:const-ex-swapr} and use it in the proof of
\Cref{th:const-ex-swap}.

\medskip
\noindent \textbf{\Ireor via Restricted \Iswap.}
It is possible to model \Ireor actions by (restricted) \Iswap actions. To do 
so, we 
construct 
a new instance from a given SM instance with agent set $A$ and a given budget 
$\ell$ as 
follows. 
First of all, we keep all agents $a\in
A$. Moreover, for each agent $a\in A$, we add a copy $a'$ with the same 
preferences as 
well as a \textit{binding agent} $\widetilde{a}$ of opposite gender with 
preferences $\widetilde{a}: a'\succ a \pend$. 
We adjust the budget to~$\ell' = \ell \cdot (4n^2+2n)$. For each $a\in A$, we 
modify the 
preferences 
of agents~$a$ and $a'$ by inserting $\widetilde{a}$ as their top-choice. 
Moreover, for 
each agent $b$ of the opposite gender of $a$ in the constructed instance, we 
insert the copy $a'$ 
directly after 
the agent $a$ in the preferences of $b$. All preference lists are completed 
arbitrarily.
 We  now
only allow to swap the first two agents in the preferences of
$\widetilde{a}$ at cost $4n^2$ and to swap all agents except $\widetilde{a}$ in 
the preferences of $a'$ at unit cost. The general idea of the construction is 
that only one 
of $a$ and $a'$ can be 
free to pair up with a non-binding agent, as the other is matched to the 
corresponding binding agent in all stable matchings. We cannot change the 
preferences of $a$ 
at all, while we can change the preferences of $a'$ at cost 
$2n$ arbitrarily (except its top-choice). Initially, $a'$ is always matched to 
the binding agent but can be 
``freed'' by modifying the preferences of $\widetilde{a}$ at cost~$4n^2$. 
Reordering the preferences of an agent $a$ in the 
original instance then corresponds to freeing $a'$ and reordering the 
preferences 
of $a'$ arbitrarily (except 
its top-choice which is irrelevant in this case).  
Overall, we can free at most $\ell$ agents $a'$, while for each of them 
we can fully reorder the relevant part of their preferences. We use a similar 
construction in the proof of \Cref{th:ex-uni-swap}.

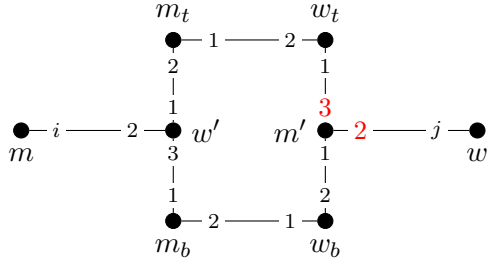
\begin{figure}[bt]
    \begin{center}
        \begin{tikzpicture}[xscale =2 , yscale = 1.2]
        \node[vertex, label=270:$m$] (m) at (-1, 0) {};
        \node[vertex, label=0:$w'$] (wf) at (0, 0) {};
        \node[vertex, label=90:$m_t$] (mt) at (0, 1) {};
        \node[vertex, label=270:$m_b$] (mb) at (0, -1) {};
        \node[vertex, label=90:$w_t$] (wt) at (1, 1) {};
        \node[vertex, label=270:$w_b$] (wb) at (1, -1) {}; \node[vertex, 
        label=180:$m'$] (mf) at (1, 0) {};
        \node[vertex, label=270:$w$] (w) at (2, 0) {};
        \draw (m) edge node[pos=0.2, fill=white, inner sep=2pt] {\scriptsize 
            $i$}  node[pos=0.76, fill=white, inner sep=2pt] {\scriptsize $2$} 
        (wf);
        \draw (mf) edge node[pos=0.2, fill=white, inner sep=2pt] {\small  
            $\textcolor{red}{2}$}  node[pos=0.76, fill=white, inner sep=2pt] 
        {\scriptsize $j$} (w);
        \draw (mf) edge node[pos=0.2, fill=white, inner sep=2pt] {\small  
            $\textcolor{red}{3}$}  node[pos=0.76, fill=white, inner sep=2pt] 
        {\scriptsize $1$} (wt);
        \draw (mf) edge node[pos=0.2, fill=white, inner sep=2pt] {\scriptsize 
            $1$}  node[pos=0.76, fill=white, inner sep=2pt] {\scriptsize $2$} 
        (wb);
        \draw (wf) edge node[pos=0.2, fill=white, inner sep=2pt] {\scriptsize 
            $3$}  node[pos=0.76, fill=white, inner sep=2pt] {\scriptsize $1$} 
        (mb);
        \draw (wf) edge node[pos=0.2, fill=white, inner sep=2pt] {\scriptsize 
            $1$}  node[pos=0.76, fill=white, inner sep=2pt] {\scriptsize $2$} 
        (mt);
        \draw (wt) edge node[pos=0.2, fill=white, inner sep=2pt] {\scriptsize 
            $2$}  node[pos=0.76, fill=white, inner sep=2pt] {\scriptsize $1$} 
        (mt);
        \draw (wb) edge node[pos=0.2, fill=white, inner sep=2pt] {\scriptsize 
            $1$}  node[pos=0.76, fill=white, inner sep=2pt] {\scriptsize $2$} 
        (mb);
        \end{tikzpicture}
        
    \end{center}
    \caption{Gadget to model \IaccD by \Iswap.
    Swapping $w$ and $w^t$ in the preferences of $m'$ (marked in red) corresponds to deleting edge $\{m, w\}$.
    }\label{fig:dellA-swap}
\end{figure}
\medskip
\noindent \textbf{\IaccD via Restricted \Iswap.}
It is possible to model deleting the acceptability of two agents by performing 
(restricted)
swaps. To do so, we modify the given SM instance by introducing for each 
man-woman
pair $\{m,w\}$ with $m\in U$ and $w\in W$ where~$m$ ranks~$w$ at position~$i$ 
and $w$~ranks
$m$ at position~$j$ the gadget depicted in \Cref{fig:dellA-swap} (note that 
this gadget was introduced by \citet{CechlarovaF05} to model parallel edges in 
a \textsc{Stable Marriage} instance).
Moreover, we only allow swapping $w$ and $w_t$ in the preferences of some 
man~$m'$ (as 
indicated in \Cref{fig:dellA-swap}). 
Matching $m$ to $w$ in the original instance corresponds to matching $m$ to
$w'$ and $m'$ to $w$ in the modified instance. Note that it is never possible
that only one of $\{m,w'\}$ and $\{m',w\}$ is part of a stable matching. 
Deleting the acceptability of a pair $\{m,w\}$ now corresponds to 
swapping $w$ and $w_t$ in $m'$'s preference relation, as in this case neither 
$\{m,w'\}$ nor $\{m',w\}$ can be part of any stable matching.

\smallskip
Summarizing, we conclude that performing swaps is, in some sense, 
the most 
powerful manipulative action considered, as all other actions can be modeled 
using it. However, this does not imply that if one of our computational 
problems is computationally hard for some manipulative action, then it is also 
hard 
for \Iswap since, for example, for the \Iexact-\Iex setting a modified problem 
definition is used for the manipulative actions \Iadd and \Idelete.

    \section{\Iconst-\Iex} \label{se:const-ex}

In this section, we analyze the computational complexity of
\const-\ex-$\mathcal{X}$.
In \Cref{se:Const-Ex-W1}, we start with showing intractability for $\mathcal{X} 
\in\{$\Iadd, \Iswap,
\IaccD, \Ireor{}$\}$.
We complement these intractability results with an XP-algorithm for 
\const-\ex-\reor (for the other manipulative actions, an XP-algorithm is 
trivial) in \Cref{section:const-ex-reor}. Subsequently, in  
\Cref{section:const-ex-poly}, we
show that \const-\ex-\delete is solvable in~$\mathcal{O}(n^2)$ time, and 
\const-\ex-\reor admits a 2-approximation with the same running time.

\subsection{A Framework for Computational Hardness} \label{se:Const-Ex-W1}

All our W[1]-hardness results for \Iconst-\Iex essentially follow from the same 
basic idea for a
parameterized reduction.
We now explain the general framework of the reduction, using the manipulative
action \Iadd as an example.
The modifications needed to transfer the approach to the
manipulative actions \Iswap, \IaccD, and \Ireor are described afterwards.

We construct a parameterized reduction from \textsc{Clique}, where given an 
undirected graph~$G = (V, E)$ 
and an
integer~$k$,
the question is whether $G$ admits a size-$k$ clique, i.e., a set of 
$k$~vertices
that are pairwise adjacent. Parameterized by $k$, \textsc{Clique} is W[1]-hard 
\citep{DBLP:series/mcs/DowneyF99}.
Fix an instance~$(G = (V, E), k)$ of \textsc{Clique} and denote
the set of vertices by~$V = \{v_1, \dots, v_{|V|}\}$ and
the set of edges by~$E = \{e_1, \dots, e_{|E|}\}$.
Let~$d_v$ denote the degree of vertex~$v$.
Moreover, let~$e^v_1,\dots, e^v_{d_v}$ be a list of all edges
incident to~$v$.

The high-level idea is
as follows.
We start by introducing two agents~$m^*$ and~$w^*$, and the edge~$\{m^*, w^*\}$
is the edge
which shall be contained in a stable matching.
Furthermore, we add~$q:= \binom{k}{2}$ women~$w_1^\dagger, \dots, w_q^\dagger$,
which we call \emph{penalizing women}.
The idea is that~$m^*$ prefers every penalizing woman to~$w^*$, and thereby, a
stable matching containing the edge~$\{m^*, w^*\}$ can only exist if every
penalizing woman~$w_j^\dagger$ is matched to a man she prefers to~$m^*$, as
otherwise~$\{m^*,
w_j^\dagger\}$ would be a blocking pair for any matching containing~$\{m^*, w^*\}$.
In addition, we introduce one vertex gadget for every
vertex and one edge
gadget for every
edge;
these differ for the different manipulative actions. Each vertex gadget
includes a \textit{vertex woman} and each edge gadget an \textit{edge man}:
A penalizing woman can only be matched to an edge man in a stable matching 
containing $\{m^*,w^*\}$.
However, an edge man can only be matched to a
penalizing woman if the gadgets corresponding to the endpoints of the edge and
the gadget corresponding to the edge itself are manipulated. Thus, one has to 
perform manipulations in at least $\binom{k}{2}$ edge gadgets
and in all vertex gadgets corresponding to the endpoints of these edges.
In this way, a budget of $\ell = k+\binom{k}{2}$ suffices if and only if $G$~contains 
a 
clique of size $k$.
\subsubsection{\Iadd} \label{se:Const-Ex-Add}
We now give the details of the parameterized reduction (following the general 
approach sketched before) for the
manipulative action \Iadd. For each vertex~$v\in V$, we introduce a vertex 
gadget consisting of
one vertex woman~$w_v$ and two men~$m'_v$ and~$m_v$.
For each edge~$e\in E$, we introduce an edge gadget consisting of an edge 
man~$m_e$,
one man~$m'_e$, and one woman~$w_e$.
Additionally, we introduce a set of~$k$ women~$\widetilde{w}_1,\dots,
\widetilde{w}_k$.
The agents that can be added are~$U_{\addag}:=\{m_v': v\in V\}\cup \{m'_e: e\in
E\}$ and~$W_{\addag}
:= \emptyset$, while all other agents are part of the original instance.
We set the budget~$\ell:= k + \binom{k}{2}$. (Note that we will show that 
the reduction, in fact,
works even if~$\ell=\infty$.)

In this reduction, adding the man~$m_v'$ for some
$v\in V$
corresponds to
manipulating the corresponding vertex gadget, whereas adding~$m'_e$ for 
some~$e\in E$ corresponds
to manipulating the corresponding edge gadget. We call the constructed 
\const-\ex-\add instance~$\mathcal{I}_{\text{add}}$.

For each vertex $v\in V$ that is incident to edges $e^v_1,\dots, e^v_{d_v}$, the 
preferences of the agents from the corresponding vertex gadget are as follows:
\begin{align*}
    w_v&: m'_v\succ m_{e^v_1}\succ \dots \succ m_{e^v_{d_v}} \succ m_v \pend,
    & 
    m'_v 
&: w_v \pend,\\
    m_v  &: w_v\succ \widetilde{w}_1 \succ \dots
    \succ
    \widetilde{w}_k \succ w^* \pend.
    &
\end{align*}
For each edge $e=\{u,v\}\in E$, the agents from the corresponding edge gadget 
have the following preferences: 
\begin{align*}
      w_e &: m'_e\succ m_e \pend,
    &
    m_e&: w_e\succ w_{u}\succ w_v \succ w^\dagger _1 
\succ\dots\succ  w^\dagger
    _q \pend,\\
    m'_e &: w_e \pend. & 
\end{align*}
Lastly, the agents $m^*$ and $w^*$ and, for $i\in [q]$ (recall that $q = \binom{k}{2}$) and 
$t\in [k]$, the agents~$w^\dagger_i$ and $\tilde{w}_t$ have the
following preferences:  
\begin{align*}
   \widetilde{w}_t& \colon m_{v_1}\succ \dots \succ
    m_{v_{|V|}} \pend,
    & 
    w^\dagger_i & : m_{e_1}\succ \dots \succ m_{e_{|E|}}\succ m^* \pend, \\
    m^* & : w^\dagger_1 \succ \dots\succ w^\dagger_q \succ w^* \pend, &
    w^* &: m_{v_1} \succ \dots \succ m_{v_{|V|}} \succ  m^* \pend.
\end{align*}

Note that in all stable matchings containing~$\{m^*, w^*\}$, every penalizing 
woman $w^\dagger _i$ is matched to a man she prefers to~$m^*$ and every 
man~$m_v$ is
matched to a woman which he prefers to
$w^*$, as otherwise the matching is blocked by~$\{m^*,w^\dagger _i\}$ or~$\{m_v,
w^*\}$.
This ensures that at most~$k$ men~$m_v'$ can be added to the instance (which will be used later to show that the reduction also works if $\ell = \infty$), as
there exist only~$k$ women~$\widetilde{w}_i$ that can be matched to some~$m_v$
from a manipulated vertex gadget. Parts of the
construction are visualized
in \Cref{fig:const-ex-add}.

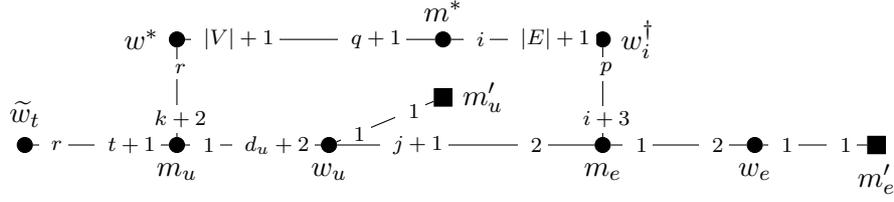
\begin{figure}[bt]
\begin{center}
    \begin{tikzpicture}[xscale =2 , yscale = 0.7]
    \node[vertex, label=270:$w_u$] (wu) at (0, 0) {};

    \node[vertex, label={[xshift=0.cm, yshift=0.cm]0:$w_i^\dagger$}] (wp) at 
(1.8, 2) {};

    \node[squared-vertex, label=0:$m_u'$, minimum size=0.23cm] (mu) at (0.75,
    0.9) {};
    \node[vertex, label=270:$m_u$] (mup) at (-1, 0) {};
    \node[vertex, label=270:$m_e$] (me) at ($(wp)+ (0, -2)$) {};
    \node[squared-vertex, label=270:$m_e'$, minimum size=0.23cm] (mep) at (3.6,
    0) {};
    \node[vertex, label=270:$w_e$] (we) at (2.8, 0) {};
    \node[vertex, label=90:$\widetilde{w}_t$] (wt) at (-2., 0) {};

    \node[vertex, label={90:$m^*$}] (ms) at (0.75, 2) 
{};
    \node[vertex, label=180:$w^*$] (ws) at (-1, 2) {};

    \draw (wu) edge node[pos=0.35, fill=white, inner sep=2pt] {\scriptsize
    ~$d_u+2$}  node[pos=0.85, fill=white, inner sep=2pt] {\scriptsize~$1$}
    (mup);
    \draw (wu) edge node[pos=0.2, fill=white, inner sep=2pt]
    {\scriptsize~$1$}  node[pos=0.76, fill=white, inner sep=2pt]
    {\scriptsize~$1$} (mu);
    \draw (me) edge node[pos=0.2, fill=white, inner sep=2pt]
    {\scriptsize~$1$}  node[pos=0.76, fill=white, inner sep=2pt]
    {\scriptsize~$2$} (we);
    \draw (we) edge node[pos=0.2, fill=white, inner sep=2pt]
    {\scriptsize~$1$}  node[pos=0.76, fill=white, inner sep=2pt]
    {\scriptsize~$1$} (mep);
    \draw (wu) edge node[pos=0.3, fill=white, inner sep=2pt]
    {\scriptsize~$j+1$}  node[pos=0.76, fill=white, inner sep=2pt]
    {\scriptsize~$2$} (me);

    \draw (me) edge node[pos=0.2, fill=white, inner sep=2pt]
    {\scriptsize~$i+3$}  node[pos=0.76, fill=white, inner sep=2pt]
    {\scriptsize~$p$} (wp);
    \draw (mup) edge node[pos=0.3, fill=white, inner sep=2pt] {\scriptsize~$t
    + 1$}  node[pos=0.85, fill=white, inner sep=2pt] {\scriptsize~$r$} (wt);

    \draw (ms) edge node[pos=0.18, fill=white, inner sep=2pt]
    {\scriptsize~$i$}  node[pos=0.7, fill=white, inner sep=2pt]
    {\scriptsize~$|E|+1$} (wp);
    \draw (mup) edge node[pos=0.2, fill=white, inner sep=2pt]
    {\scriptsize~$k+2$}  node[pos=0.76, fill=white, inner sep=2pt]
    {\scriptsize~$r$} (ws);

    \draw (ms) edge node[pos=0.25, fill=white, inner sep=2pt] {\scriptsize
    ~$q+ 1$}  node[pos=0.8, fill=white, inner sep=2pt] {\scriptsize~$|V| +1$}
    (ws);
    \end{tikzpicture}
\end{center}
\caption{A vertex gadget and an edge gadget for the hardness reduction for
\add, where
~$e = e_j^u = e_p$ and~$u = v_r$.
The squared vertices are the vertices from~$U_{\addag}$ that can be added
to the instance.
In the figure, we only exemplarily
depict one penalizing woman~$w^\dagger_i$ for some~$i \in [q]$ and one woman~$\widetilde{w}_{t}$ for some~$t\in [k]$.
For an edge $\{x, y\}$, the number on this edge closer to $x$ indicates the 
rank of~$y$ in $x$'s preference order.
}\label{fig:const-ex-add}
\end{figure} 

 Note that in the instance as described above there are $2|V| + 2|E|+1$ 
 men and $|V| + |E| + k+q+1$ women.
 However, our definition of \textsc{Stable Marriage} requires the instance to 
 have the same number of men and women.
 This can be achieved by adding $|V| + |E| - k-q$ \emph{filling women} 
 (to 
 $W_{\orig}$).
 These filling women have arbitrary preferences and every man prefers any 
 non-filling woman to any filling woman.
 The presence or absence of filling women does not change the existence of a 
 stable matching containing~$\{m^*, w^*\}$ because every stable matching in the 
 presence of filling women can be transformed into a stable matching in the 
 absence of filling women by deleting all edges incident to a filling woman. 
 Moreover, 
 every stable matching~$M$ in the absence of filling women can be transformed 
 to 
 a stable matching in the presence of filling women by adding to $M$ a stable 
 matching 
 between the set of men unassigned by~$M$ and filling women (note that 
 such a stable matching has to exist since every \textsc{Stable Marriage} 
 instance admits a stable matching).
 In order to keep the proof of correctness of the reduction simpler, we 
 will 
 ignore all filling women and assume they are not part of the instance.

\begin{lemma} \label{le:const-ex-add-L1}
If~$G$ contains a clique of size~$k$, then~$\mathcal{I}_{\mathrm{add}}$ is a
YES-instance.
\end{lemma}

\begin{proof}
Let~$C$ be a clique in~$G$. For an edge $e=\{u,v\}\in E$, we write
$e\subseteq C$ to express that~$e$ lies in $C$, i.e., $u\in C$ and $v\in C$.
Further, let~$C[i]$ denote the vertex with~$i$-th
lowest index in~$C$ and~$D[i]$ the edge with~$i$-th lowest index in~$C$.
We add the $\ell$~agents~$\{m_v' : v\in C\}$ and~$\{m_e': e \subseteq C\}$,
and claim that
\begin{align*}
    M:= & \{\{m^*, w^*\}\} \cup \{\{m_v', w_v\} : v\in C\} \cup
    \{\{m_{C[i]}, \widetilde{w}_i\} : i \in [k]\} \cup\\ & \{\{m_v, w_v\} : v\in V
    \setminus C\} \cup \{ \{m_e, w_e\} : e\not \subseteq C\}
    \cup \\ & \{\{m_e', w_e\} :
    e
    \subseteq C\} \cup \{\{m_{D[i]}, w_i^\dagger\} : i\in [q]\}
\end{align*}
is a
stable matching, containing $\{m^*,w^*\}$. We now iterate over all agents
present in the instance after adding $\{m_v' : v\in C\}$ and~$\{m_e': e
\subseteq C\}$ and argue why they cannot be part of a blocking pair.
Let $A':= U_{\orig} \cup W_{\orig} \cup \{m_v' : v\in C\} \cup \{m_e': e 
\subseteq 
C\}$ be the agents contained in the instance arising through the addition of 
$\{m_v' : v\in C\} \cup \{m_e': e \subseteq C\}$.

First, note that for each~$e\not\subseteq C$ the agents~$m_e$
and~$w_e$ are matched to their top-choice in the instance and, therefore,
cannot be part of a blocking pair.

For each vertex~$v\in V\setminus C$, man~$m_v$ is matched to his first
choice and thus is not part of a blocking pair.
Since $v \in V\setminus C$, every edge $e $ incident to $v$ is not contained 
in the clique, and consequently, $m_e$ is not part of a blocking pair.
As all agents in $A'$ which $w_v$ prefers to~$m_v$ are not part of a blocking 
pair, also~$w_v$ is not part of
a blocking pair.

For each vertex~$v\in C$, the agents~$m_v'$ and~$w_v$ are matched to their
top-choice and thus are not part of a blocking pair. 
Consider a man~$m_v$ for some $v\in C$.
This man is matched to woman~$\widetilde{w}_j$ for some $j\in [k]$.
Edge~$\{m_v, w_v\}$ is not blocking, as $w_v$ is not part of a blocking pair.
Moreover, there cannot exist a blocking pair of the
form~$\{m_v, \widetilde{w}_i\}$ for some~$i\in [k]$, as $m_v$ only
prefers women $\widetilde{w}_i$ with $i<j$. However, all women
$\widetilde{w}_i$ with~$i<j$ prefer their current partner to $m_v$, as they
are all assigned a man corresponding to a vertex with a smaller index than
$v$.

Recall that there cannot exist a blocking pair involving an agent from 
$\{w_v: v\in V\}$. Thus,  since all agents from~$\{m_e: e\subseteq C\}$
have the same preferences over
the penalizing women, and the penalizing women prefer each man from~$\{m_e:
e\subseteq C\}$ to~$m^*$, there is no blocking pair involving agents
from~$\{m_e: e\subseteq C\} \cup \{w_i^\dagger : i\in
[q]\}$. 

Finally, neither $m^*$ nor $w^*$ are part of a blocking pair, as all agents
which they prefer to each other (i.e., penalizing women $w_i^\dagger$ or men $m_v$)
are not contained in a blocking pair.
Thus,~$M$ is stable. Note that if the $k$ vertices from $C$ were not to form a 
clique, then the set $\{m_e : e\subseteq C\}$ would consist of less than 
$\binom{k}{2}$ men. Thus, not all penalizing women are matched to an edge 
man in $M$ which implies that $m^*$ and a
penalizing woman form a blocking pair for $M$.
\end{proof}

We now proceed with the backward direction.

\begin{lemma} \label{le:const-ex-add-L2}
If there exists a set $X_A$ of agents (of arbitrary size) such that 
after their addition there exists a stable matching containing $\{m^*, w^*\}$, 
then~$G$
contains a clique of size~$k$.
\end{lemma}

\begin{proof}
Let $M$ be a stable matching containing $\{m^*, w^*\}$.
Since the edges~$\{m^*,
w_i^\dagger\}$ are not blocking, all
penalizing women are matched to an edge man~$m_e$ for some~$e\in E$.
This
requires that~$m'_e\in X_A$, as otherwise~$\{m_e,w_e\}$ is a blocking pair.
Moreover, for each such edge~$e = \{u, v\}$, the vertex women~$w_u$ and~$w_v$
have to be either matched to other edge men or to the men~$m_u'$ or~$m_v'$.
Note that in both cases, the corresponding agents~$m_u$ and~$m_v$ are matched 
to one of the women~$\widetilde{w}_i$, as otherwise~$\{m_u, w^*\}$
or~$\{m_v, w^*\}$ is a blocking pair. Thus, there exist at most~$k$ vertices~$v\in
V$ where $w_v$ is matched to an edge men or to~$m_v'$.

Since there are~$\binom{k}{2}$ penalizing women, and each of them is matched
to an edge man, it follows that these edge men correspond to the edges 
in the clique formed by the $k$ vertices~$v\in V$ where $w_v$ is either matched 
to an edge men or $m_v'$.
\end{proof}

From \Cref{le:const-ex-add-L1} and \Cref{le:const-ex-add-L2}, we 
conclude
that there exists a parameterized reduction from
\textsc{Clique} parameterized by~$k$ to \const-\ex-\add parameterized
by~$\ell$, implying the following.

\begin{theorem}\label{cor:const-ex-add}
Parameterized by the budget~$\ell$, it is W[1]-hard  to decide whether
\const-\ex-\add has
a solution with at most~$\ell$ additions or has no solution with an arbitrary
number of additions, even if we are only allowed to add agents of one gender.
\end{theorem}

\subsubsection{\IaccD and \Ireor}

The reduction for \Iadd presented in \Cref{se:Const-Ex-Add} cannot be directly 
applied for \IaccD and \Ireor. Instead, new vertex and edge gadgets need to be 
constructed.
One reason for this is that for \Iadd (and \Iswap), we could ensure that the
penalizing women are not manipulated.
However, they can be manipulated by \IaccD and \Ireor operations, and,
therefore, in our construction for \Iadd, there
exists an easy solution with~$q$~manipulations (recall that $q= \binom{k}{2}$), which just deletes
the acceptabilities~$\{m^*, w_i^\dagger\}$ in the case of \IaccD or moves~$m^*$
to
the end of~$w_i^\dagger$'s preferences for each~$i\in [q]$ in the case
of \Ireor.
To avoid such solutions, we modify the construction for \Iadd as follows. We 
add~$q$ additional penalizing men~$m_1^\dagger,
\dots,
m_{q}^\dagger$, and one manipulation of an edge gadget will now allow to match
both a penalizing woman and a penalizing man to this edge gadget.
We assume without loss of generality that $k \ge 6$,
as this makes the proof of \Cref{le:const-ex-reor-L2,le:const-ex-reor-L3} 
easier.

We now describe the details of the construction of the SM instance that we want
to manipulate, which is the same for manipulative actions \Ireor and \IaccD.
Given an instance of \textsc{Clique} consisting of a graph $G=(V,E)$ and an 
integer $k$, for
each vertex~$v\in V$ we introduce
a gadget consisting of one vertex woman
$w_v$ and one vertex man~$m_v$ together with one woman~$w'_v$ and one
man~$m_v'$.
The preferences are as follows:
\begin{align*}
w_v&: m_v' \succ m_{e^v_1}\succ \dots \succ m_{e^v_{d_v}} \succ m_v \pend,
&  w_v' &: m_v' \succ m_v \pend,   \\
m_v &: w'_v\succ w_{e^v_1} \succ \dots \succ w_{e^v_{d_v}} \succ w_v \pend,
&   m'_v &: w'_v\succ w_v \pend.
\end{align*}
For each edge $e=\{u,v\}\in E$, we introduce a gadget consisting of one edge 
man~$m_e$, one edge
woman~$w_e$ together with two men~$m_e'$ and
$m_e''$, and two women~$w_e'$ and~$w_e''$.
The preferences are as follows:
\begin{align*}
m_e&: w_e'\succ w_{u}\succ w_v \succ w^\dagger _1 \succ\dots\succ  w^\dagger
_q\pend,  &   m_e' &: w_e''\succ  w_e \pend,\\
w_e&: m_e'\succ m_{u}\succ m_v \succ m^\dagger _1 \succ\dots\succ  m^\dagger
_q  \pend, &   w_e' &: m_e''\succ  m_e \pend,\\
m_e'' &: w_e'' \succ w_e' \pend,  & w_e'' &: m_e'' \succ m_e' \pend.
\end{align*}
See \Cref{fig:const-ex-del-acc} for an example of a vertex gadget and an edge
gadget.
\begin{figure}[bt]
\begin{center}
    \begin{tikzpicture}[xscale =2 , yscale = 1.2]
    \node[vertex, label=90:$w_u$] (wu) at (0, 0.5) {};
    \node[vertex, label=270:$w_u'$] (wup) at (-1, -1.5) {};

    \node[vertex, label=270:$m_u$] (mu) at (0, -1.5) {};
    \node[vertex, label=90:$m_u'$] (mup) at (-1, 0.5) {};
    \node[vertex, label=90:$m_e$] (me) at (2, 0) {};
    \node[vertex, label=90:$m_e''$] (mep) at (4, 0) {};
    \node[vertex, label=90:$w_e'$] (we) at (3, 0) {};
    \node[vertex, label=270:$w_e$] (xe) at (2, -1) {};
    \node[vertex, label=270:$m_e'$] (ne) at (3, -1) {};
    \node[vertex, label=270:$w_e''$] (xep) at (4, -1) {};

    \node[vertex, label=270:$m^*$] (ms) at (4.5, -2) {};
    \node[vertex, label=90:$w^*$] (ws) at (4.5, 1) {};
    \node[vertex, label=270:$m_i^\dagger$] (mp) at (2.8, -2) {};
    \node[vertex, label=90:$w_i^\dagger$] (wp) at (2.8, 1) {};

    \draw (wu) edge node[pos=0.2, fill=white, inner sep=2pt]
    {\scriptsize~$1$}  node[pos=0.76, fill=white, inner sep=2pt]
    {\scriptsize~$2$} (mup);
    \draw (mup) edge node[pos=0.2, fill=white, inner sep=3pt]
    {\scriptsize~$1$}  node[pos=0.76, fill=white, inner sep=3pt]
    {\scriptsize~$1$} (wup);
    \draw (wup) edge node[pos=0.2, fill=white, inner sep=2pt]
    {\scriptsize~$2$}  node[pos=0.76, fill=white, inner sep=2pt]
    {\scriptsize~$1$} (mu);
    \draw (me) edge node[pos=0.2, fill=white, inner sep=2pt]
    {\scriptsize~$1$}  node[pos=0.76, fill=white, inner sep=2pt]
    {\scriptsize~$2$} (we);
    \draw (we) edge node[pos=0.2, fill=white, inner sep=2pt]
    {\scriptsize~$1$}  node[pos=0.76, fill=white, inner sep=2pt]
    {\scriptsize~$2$} (mep);
    \draw (mep) edge node[pos=0.2, fill=white, inner sep=2pt]
    {\scriptsize~$1$}  node[pos=0.76, fill=white, inner sep=2pt]
    {\scriptsize~$1$} (xep);
    \draw (xep) edge node[pos=0.2, fill=white, inner sep=2pt]
    {\scriptsize~$2$}  node[pos=0.76, fill=white, inner sep=2pt]
    {\scriptsize~$1$} (ne);
    \draw (ne) edge node[pos=0.2, fill=white, inner sep=2pt] {\scriptsize
    ~$2$}  node[pos=0.76, fill=white, inner sep=2pt] {\scriptsize~$1$} (xe);
    \draw (mu) edge node[pos=0.2, fill=white, inner sep=2pt] {\scriptsize
    ~$d_v+2$}  node[pos=0.76, fill=white, inner sep=2pt] {\scriptsize~$d_v+2$}
    (wu);

    \draw (wu) edge[dotted] node[pos=0.2, fill=white, inner sep=2pt]
    {\scriptsize~$j+1$}  node[pos=0.76, fill=white, inner sep=2pt]
    {\scriptsize~$2$} (me);
    \draw (mu) edge[dotted] node[pos=0.2, fill=white, inner sep=2pt]
    {\scriptsize~$j+1$}  node[pos=0.76, fill=white, inner sep=2pt]
    {\scriptsize~$2$} (xe);

    \draw (ms) edge node[pos=0.25, fill=white, inner sep=2pt] {\scriptsize
    ~$q+ 1$}  node[pos=0.8, fill=white, inner sep=2pt] {\scriptsize
    ~$q+1$} (ws);
    \draw (ms) edge node[pos=0.25, fill=white, inner sep=2pt] {\scriptsize
    ~$i$}  node[pos=0.8, fill=white, inner sep=2pt] {\scriptsize
    ~$|E|+1$} (mp);
    \draw (ws) edge node[pos=0.25, fill=white, inner sep=2pt] {\scriptsize
    ~$i$}  node[pos=0.8, fill=white, inner sep=2pt] {\scriptsize~$|E|+1$} (wp);

    \draw (me) edge node[pos=0.4, fill=white, inner sep=2pt] {\scriptsize
    ~$i+3$}  node[pos=0.8, fill=white, inner sep=2pt] {\scriptsize~$p$} (wp);
    \draw (xe) edge node[pos=0.4, fill=white, inner sep=2pt] {\scriptsize
    ~$i+3$}  node[pos=0.8, fill=white, inner sep=2pt] {\scriptsize~$p$} (mp);
    \end{tikzpicture}

\end{center}
\caption{Visualization of the reduction for \IaccD and \Ireor. For some edge~$e_p=\{u,v\}\in E$, the edge gadget corresponding to $e_p$ and the vertex
gadget corresponding to~$u$ (assuming that $e_p=e_j^u$) are included as well
as $m^*$ and $w^*$ together with penalizing agents $m^\dagger_i$ and
$w^\dagger_i$ for some arbitrary $i\in[q]$.
Edges between the vertex and the edge gadget are dotted.}\label{fig:const-ex-del-acc}
\end{figure}
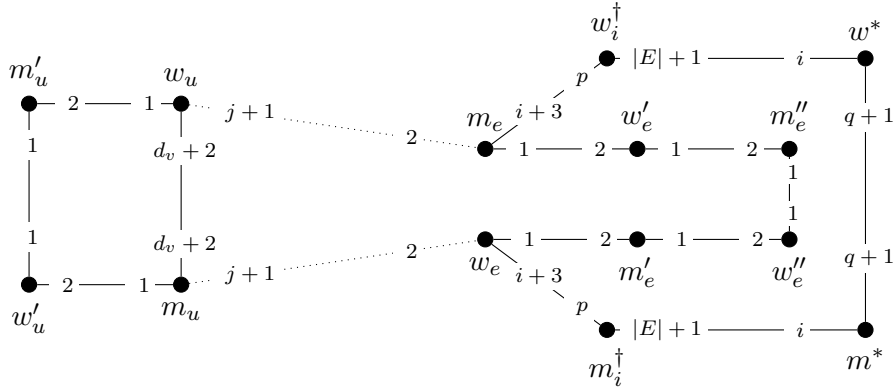
The preferences of
the agents~$m^*$,~$w^*$ and the penalizing agents are as follows:
\begin{align*}
    w^\dagger_i & : m_{e_1}\succ \dots \succ m_{e_{|E|}}\succ m^* \pend, &
    m^* & : w^\dagger_1 \succ \dots \succ w^\dagger_q \succ w^* \pend,\\
    m^\dagger_i & : w_{e_1}\succ \dots \succ w_{e_{|E|}}\succ w^* \pend,
    &   w^* &: m^\dagger_1 \succ \dots \succ m_q^\dagger \succ m^* \pend.
\end{align*}
Finally, we set~$\ell:= \binom{k}{2} + k$.
By $\mathcal{I}_{\text{del}}$ we denote the resulting instance of
\const-\ex-\accD
and by $\mathcal{I}_{\text{reor}}$ the resulting instance of
\const-\ex-\reor. In
the following, in \Cref{le:const-ex-reor-L1}, we prove the forward 
direction of the reduction for both \Ireor and \IaccD. In
\Cref{le:const-ex-reor-L2}, we prove the correctness of the backward direction 
for
\Ireor and in \Cref{le:const-ex-reor-L3} for \IaccD.
\begin{lemma} \label{le:const-ex-reor-L1}
If~$G$ contains a clique of size~$k$, then~$\mathcal{I}_\mathrm{del}$
and~$\mathcal{I}_\mathrm{reor}$ are YES-instances.
\end{lemma}

\begin{proof}
Let~$C\subseteq V$ be a clique.
We denote by~$D[i]$ the edge with~$i$-th lowest index in the clique and for an 
edge $e=\{u,v\}\in E$, we write $e\subseteq C$ if $u\in C$ and $v\in C$.
In~$\mathcal{I}_\text{del}$, we delete the following~$\ell$ edges.
For each~$v\in C$, we delete~$\{m_{v}', w_v'\}$, and for each edge~$e\subseteq
C$, we delete~$\{m_e'', w_e''\}$.
In~$\mathcal{I}_\text{reor}$, we manipulate the following~$\ell$ agents.
For each~$v\in C$, we change the preferences of~$m_v'~$ to~$m_v': w_v
\pend$.
For each edge~$e\subseteq C$, we change the preferences of~$m_e''$ to~$m_e'': w_e' \pend$.

We claim that the matching
\begin{align*}
    M:= &\{\{m^*, w^*\}\} \cup \{\{m_v, w_v'\}, \{m_v', w_v\} : v\in
    C\}
    \cup \{\{m_v', w_v'\}, \{m_v,w_v\} : v\in V\setminus C\} \cup\\ & \{\{m_e, w_e'\}
    ,
    \{m_e', w_e\},
    \{m_e'', w_e''\}: e \not \subseteq C\} \cup\{\{m_e', w_e''\}, \{m_e'',
    w_e'\} : e
    \subseteq C\}\cup\\ & \{\{m_{D[i]}, w_i^\dagger\}, \{m_i^\dagger, w_{D[i]}\} :
    i\in
    [q]\},
\end{align*} which contains $\{m^*,w^*\}$, is stable in both modified instances.

To see this, first note that for each~$e\not\subseteq C$, the
agents~$m_e$,~$m_e''$,~$w_e$, and~$w_e''$ are matched to their top-choices,
and therefore are not part of a blocking pair.
As a consequence, also~$m_e'$ and~$w_e'$ cannot be part of a blocking pair.

For each agent~$v\in V\setminus C$, both~$m_v'$ and~$w_v'$ are matched to
their top-choices and thus not part of a blocking pair.
The only agents which $m_v$ prefers to $w_v$ are $w_v'$ and $w_e$ for every edge 
$e = \{v, u\}\in E$ incident to $v$;
however, we showed for all these agents that they are not part of a blocking pair.
Thus, $m_v$ is not part of a blocking pair.
Symmetrically, the only agents which $w_v$ prefers to $m_v$ are $m_v'$ and $m_e$ 
for every edge $e = \{v, u\} \in E$ incident to~$v$, and also these agents are 
not contained in a blocking pair.
	Therefore, also~$w_v$ is not part of a blocking pair.

For each agent~$v\in C$, the agents~$m_v$,~$m_v'$, and~$w_v$ are
matched to their top-choices and thus are not part of a blocking pair (note that 
for both \Ireor and \IaccD, we have modified the preferences of $m_v'$ such 
that $w_v$ is his top-choice).
In \IaccD, also woman~$w_v'$ is matched to her top-choice.
In \Ireor, woman~$w_v'$ is matched to her second choice, while her top-choice $m_v'$ is 
not part of a blocking pair.
Thus, $w_v'$ is also not part of a blocking pair.

Since all agents from~$\{m_e: e\subseteq C\}$ have the same preferences over
the penalizing women, and the penalizing women prefer each man from~$\{m_e:
e\subseteq C\}$ to~$m^*$, there is no blocking pair involving only agents
from~$\{m_e: e\subseteq C \} \cup \{m^*\} \cup \{w_i^\dagger: i \in [q]\}$.
Symmetrically, it follows that no blocking pair involves only agents from
$\{w_e: e\subseteq C\} \cup \{w^*\} \cup \{m_i^\dagger: i \in [q]\}$.
Thus,~$M$ is stable.  Note that if the $k$ vertices from $C$ were not to form a 
clique, then the set $\{m_e : e\subseteq C\}$ would consist of less than 
$\binom{k}{2}$ men. Thus, not all penalizing women are matched to an edge 
men in $M$ which implies that $m^*$ and a
penalizing women form a blocking pair for $M$.
\end{proof}

We now prove the backward direction for \Ireor.

\begin{lemma} \label{le:const-ex-reor-L2}
    If~$\mathcal{I_\mathrm{reor}}$ is a YES-instance, then~$G$ contains a
    clique of size~$k$.
\end{lemma}

\begin{proof}
Let~$X_{\text{reor}}$ be the set of at most~$\ell$ agents whose preferences
have been
reordered, and let~$M$ be a stable matching containing~$\{m^*, w^*\}$. In the
following, we call a vertex agent $a_v$ \emph{unhappy} if it is neither contained in $X_{\text{reor}}$ nor matched
to
one of its $d_v+1$ most preferred partners, i.e., it prefers all edge agents of 
edges incident to $v$ to its current
partner.
Note that for each edge gadget for an edge~$e\in E$ such that no agent from the
edge gadget is contained in~$X_{\text{reor}}$, every stable matching contains
the edges~$\{m_e, w_e'\}$,~$\{m_e', w_e\}$, and~$\{m_e'',w_e''\}$.
If an edge agent~$a_e\in \{m_e, w_e\}$ is the only agent from this edge gadget
contained in~$X_{\text{reor}}$, then matching~$M$
contains the edge~$\{m_e, w_e'\}$ if~$a_e = w_e$ and~$\{w_e, m_e'\}$ if~$a_e =
m_e$.
Furthermore, each penalizing agent~$a_i^\dagger$ needs to be contained
in~$X_{\text{reor}}$ or
matched to an edge agent~$a_e$ since, otherwise, $\{m_i^{\dagger}, w^*\}$ if $a_i^\dagger = m_i^{\dagger}$
	or $\{m^*, w_i^{\dagger}\}$ if $a_i^\dagger = w_i^{\dagger}$ blocks $M$.

Let $p$ be the number of agents from~$X_{\text{reor}}$ which are  
edge agents 
matched to a penalizing agent or are penalizing agents.
As there are~$2 \binom{k}{2}$ penalizing agents, at least $2 \binom{k}{2}-p$
of them
need to be
matched to edge agents which are not in $X_{\text{reor}}$, as otherwise at least one
penalizing agent forms a blocking pair together with $m^*$ or $w^*$.
Since, as argued above, matching an edge agent
to a penalizing agent requires that at least one agent from the corresponding edge gadget is
part of $X_{\text{reor}}$ (and it is not sufficient to reorder the preferences of the other edge agent from this gadget), it follows that at
least~$\frac{2\binom{k}{2} -p}{2}  + p= \binom{k}{2} + \frac{p}{2} \le \ell$ agents from $X_{\text{reor}}$ are
agents from edge gadgets or penalizing agents.
As $\ell = \binom{k}{2} + k$, it follows that $p \le 2k$.

For every vertex agent~$a_v\in \{m_v, w_v\}$ it holds that if none of the
vertices from its vertex gadget is contained in~$X_{\text{reor}}$, then
agent~$a_v$ is either
matched to an edge agent or unhappy.
As all but at most $k - \frac{p}{2}$ agents from $X_{\text{reor}}$ are either 
a penalizing agent or contained in an edge gadget, at 
least $2\binom{k}{2} - p$ edge agents
are matched to penalizing agents.
Because an edge agent may only be matched to an agent outside its edge gadget  
if the preferences of at least one agent from the gadget are modified, it 
follows that there can be at
most~$2k- p$ happy vertex agents.
Thus, without loss of generality, there are at most~$k - \frac{p}{2}$ happy
vertex men (otherwise we apply the following argument for happy vertex
women).

Note that at least~$\binom{k}{2} - p~$ penalizing men are matched to edge women 
which
are not contained in~$X_{\text{reor}}$ in $M$.
These~$\binom{k}{2} - p$ edge women~$w_e$ for some $e\in E$ prefer the vertex
men corresponding
to the endpoints of~$e$ to each penalizing man.
Since only the at most~$k - \frac{p}{2}$ happy vertex men may not prefer to be 
matched to an edge women corresponding to an incident edge, it 
follows that the~$\binom{k}{2}-p$ edges to which the $\binom{k}{2} - p$ edge 
women~$w_e$ correspond have at most~$k -
\frac{p}{2}$ endpoints.
This is only possible if~$\binom{k}{2} - p \le \binom{k -\frac{p}{2}}{2}$,
which is equivalent to~$0 \le p \cdot (p - 4k + 10)$, implying (as~$p \ge 0$) that~$p=
0$ or~$p \ge 4k -10$.
However, the latter case combined with our previous observation that $p \le 2k$ 
implies that $2k \ge p \ge 4
k - 10$, a contradiction for $k \ge 6$.
It follows that $p= 0$. This implies that indeed $k$ vertex men are happy,
no penalizing woman has been manipulated, and that in $\binom{k}{2}$ edge gadgets
exactly one agent, but no edge agent has been manipulated.
Thus, each penalizing woman is matched to an edge man $m_e$ such that both
endpoints of $e$ are happy.
It follows that there are~$\binom{k}{2}$~edges whose endpoints are among
the $k$ vertices whose vertex man $m_v$ is happy.
These~$k$~vertices clearly form a clique.
\end{proof}

In a similar way, we can show the analogous statement
for~$\mathcal{I}_{\text{del}}$.

\begin{lemma} \label{le:const-ex-reor-L3}
If~$\mathcal{I_\mathrm{del}}$ is a YES-instance, then~$G$ contains a clique
of
size~$k$.
\end{lemma}

\begin{proof}
Let~$X_{\text{del}}$ be the set of at most~$\ell$ pairs which have been
deleted, and let~$M$ be a stable matching containing~$\{m^*, w^*\}$.
Note that, for each edge gadget for an edge~$e$ such that no pair from the edge
gadget is contained in~$X_{\text{del}}$, any stable matching contains the
edges~$\{m_e, w_e'\}$,~$\{m_e', w_e\}$, and~$\{m_e'', w_e''\}$.
For each penalizing agent~$a_i^\dagger$, either the edge~$\{a_i^\dagger, a^*\}$
where~$a^*$ is the agent from~$\{m^*, w^*\}$ of opposite gender is contained
in~$X_{\text{del}}$ or~$a_i^\dagger$ is matched to an edge agent~$a_e$.
Let~$p$ be the number of pairs in~$X_{\text{del}}$ containing a penalizing
agent and an
agent from~$\{m^*, w^*\}$.
As there are~$2 \binom{k}{2}$ penalizing agents, at
least~$\frac{2\binom{k}{2} -p}{2}  + p= \binom{k}{2} + \frac{p}{2} \le 
\ell=\binom{k}{2} + k$
deletions happen where either both involved agents are from the same edge 
gadget or one of the involved agents is a penalizing agent and the other 
one is from~$\{m^* , w^*\}$.
From this, it follows that $p \le 2k$.

For every vertex agent~$a_v$ it holds that if no pair in the corresponding
vertex gadget is contained in~$X_{\text{del}}$, then the agent~$a_v$ is either
matched to an edge
agent or unhappy (where $a_v$ is unhappy if none of the edges incident to it is deleted and it is not matched to one of its first~$d_v +1$ most preferred partners).
Since, as argued above, all but $k - \frac{p}{2}$ of the deleted pairs do not 
involve a vertex agent,
and at least $2\binom{k}{2} - p$ edge agents are matched to penalizing agents,
there can be at
most~$2k- p$ happy vertex agents.
The rest of the argument is the same as in the proof of
\Cref{le:const-ex-reor-L2}.
\end{proof}

The following theorem follows directly from
\Cref{le:const-ex-reor-L1,le:const-ex-reor-L2,le:const-ex-reor-L3}.

\begin{theorem} \label{th:const-ex-reor}
Parameterized by the budget~$\ell$, \const-\ex-\accD is W[1]-hard.
Parameterized by~$\ell$, \const-\ex-\reor is W[1]-hard, and this also holds if one is only
allowed to reorder the preferences of agents of one gender.
\end{theorem}

Note that the presented construction, in contrast to the reduction for \Iadd
and \Iswap, does not have implications in terms of inapproximability of \Ireor and
\IaccD. In particular, as described in the beginning of this section, there
always exists a trivial solution of cost $2q$. In fact, we show in the next
section that \Ireor admits a factor-2 approximation.
Note further that the presented construction is also a valid parameterized
reduction
from \textsc{Clique} to \const-\ex-\swap;
however, we will derive a stronger hardness result for this problem (yielding also
FPT-inapproximability) in \Cref{sec:const-ex-swap}.

\subsubsection{\Iswap}\label{sec:const-ex-swap}
\Cref{cor:const-ex-add} showed that it is W[1]-hard to distinguish whether it 
is possible to make an edge~$e^* = \{m^*  , w^*\}$ part of a stable matching or 
no set of agents whose addition makes~$e^*$ part of a 
stable matching exists.
We now use this
W[1]-hardness to derive
an FPT-inapproximability result for
\const-\ex-\swap by a reduction from \const-\ex-\add.
We achieve this result in two steps.
First, we consider a variant of \Iswap,
which
we call \IswapR: Here, a subset~$\hat{A}$ of
agents is given, and one is only allowed to swap the first two agents in the
preference lists of agents from $\hat{A}$, while the preferences of all agents 
$A\setminus \hat{A}$ need to remain unmodified.
By reducing from \const-\ex-\add, we show that parameterized by the 
budget~$\ell$, it is W[1]-hard to decide whether
\const-\ex-\swapR has a solution with at most~$\ell$~swaps or has no solution
with an arbitrary number of allowed swaps.
Second, we derive our FPT-inapproximability result for \const-\ex-\swap by 
reducing from \const-\ex-\swapR.

We start by showing W[1]-hardness of \const-\ex-\swapR.

\begin{lemma}\label{cor:const-ex-swapr}
Parameterized by the budget~$\ell$, it is W[1]-hard to decide whether
\const-\ex-\swapR has a solution with at most $\ell$ swaps or has no solution
with an arbitrary number of allowed swaps.
\end{lemma}

\begin{proof}
We reduce from \const-\ex-\add, for which it is W[1]-hard to distinguish whether
there is a solution with at most~$\ell$ additions or no solution for any 
number of additions 
(\Cref{cor:const-ex-add}).
Let $(\mathcal{I}, U_{\addag}, W_{\addag}, \ell)$ be an instance of
\const-\ex-\add.
We form an equivalent instance of \const-\ex-\swapR by adding for each man
$m\in U_{\addag}$ a woman $w_m$ and a man $m'$, where~$w_m$ is added to $m$'s
preferences as the top-choice.
The man $m'$ has $w_m$ as his top-choice and the woman $w_m$ has~$m$ as her
top-choice and $m'$ as her second most preferred man. All other agents follow
in an arbitrary order in the preferences of $w_m$ and $m'$.
All other agents add $w_m$ or $m'$ at the end of their preferences.
Similarly, we add for each woman~$w\in W_{\addag}$ two agents~$m_w$ and~$w'$ 
whose preferences are constructed analogously.
We set~$\hat{A}:= \{m_w : w\in W_{\addag}\} \cup \{w_m : m\in U_{\addag}\}$.

It remains to show that the instances are equivalent.
From a solution to \const-\ex-\add, consisting of the agents $X_U \cup X_W$,
one can get a solution to \const-\ex-\swapR  by swapping the two
most preferred men~$m$ and~$m'$ in the preferences of
$w_m$ for each~$m\in X_U$ and modifying the preferences of $m_w$ for 
each $w\in X_W$ analogously.
Vice versa, from a solution to \const-\ex-\swapR modifying the preferences of
a set of agents~$X$ one can get a solution to \const-\ex-\add by adding $w\in
W_{\addag}$ if $m_w \in X_U$ and $m\in U_{\addag}$ if $w_m \in X_W$.
It is now straightforward to verify the correctness.
\end{proof}

We continue by reducing \IswapR to \Iswap.
The basic idea behind this reduction is that we can introduce a set of dummy agents and use these agents to make the swaps not allowed in the \IswapR instance
too expensive.
First,
we show that
we can add a set of $r > \ell$ men $m_1^d$, \dots, $m_r^d$
and~$r$~women~$w_1^d$, \dots, $w_r^d$ to an SM instance such that, in any
instance
arising through at most $\ell$ swaps, any stable matching contains the edges
$\{m_i^d, w_i^d\} $ for all $i\in [r]$ regardless where the newly inserted
agents are placed in the preferences of the other agents.
This allows us to make swapping two neighboring non-dummy agents $a$ and~$a' $ 
in the
preference list of some non-dummy agent $b$ very expensive, as we 
can
insert (some of)
the newly added dummy agents in between $a$ and~$a'$ in~$b$'s preference list.

Given an SM instance $\mathcal{I}$ and a list of swap operations~$S$,
we denote by~$\mathcal{I}[S]$
the SM instance resulting from applying
the swaps from~$S$ to $\mathcal{I}$.

\begin{lemma}\label{lem:dummy-agents}
Let $\mathcal{I}$ be a \textsc{Stable Marriage} instance containing~$r$ 
men~$m_1^d, \dots, m_r^d$ and~$r$~women $w_1^d, \dots, w_r^d$ such that, for 
all $i\in [r]$, the
preferences of $w_i^d$ and $m^d_i$ match the following pattern (where all
indices are taken modulo $r$):
\begin{align*}
    m_i^d &: w_i^d \succ w^d_{i+1} \succ w^d_{i+2} \succ \dots \succ w^d_{r +
    i-1} \pend, \\
    w^d_i &: m^d_i \succ m^d_{i+1} \succ m^d_{i+2} \succ \dots \succ m^d_{r +
    i-1}\pend.
\end{align*}
For any list $S$ of at most $r-1$ swap operations, any stable matching in the
instance~$\mathcal{I}[S]$ contains the edges $\{m_i^d, w_i^d\}$ for all $i\in
[r]$.
\end{lemma}

\begin{proof}
Let $S$ be any list of at most $r-1$ swap operations. For the sake of
contradiction,
assume that there exists an~$i_1\in [r]$ and a stable matching $M\in
\mathcal{M}_{\mathcal{I}[S]}$ such that~$\{m_{i_1}^d, w_{i_1}^d\}\notin M$.
This implies that there either exist $s$ indices $i_1,\dots i_s$ such 
that~$\{m^d_{i_1},w^d_{i_s}\}$ and~$\{m_{i_{j+1}}^d, w_{i_j}^d\} \in M$ for 
$j\in[s-1]$ or there exist some~$w\notin \{w_1^d, \dots, w_r^d\}$ and~$m\notin 
\{m_1^d, \dots, m_r^d\}$
together with~$s$ indices $i_1,\dots i_s$ such that~$\{m_{i_1}^d, w\} \in M$
and~$\{m, w_{i_s}^d\}\in M$ and~$\{m_{i_{j+1}}^d, w_{i_j}^d\} \in M$
for~$j\in[s-1]$.

In the first case, for every $j \in [s]$, at least one of~$w_{i_j}^d$ and
$m_{i_{j}}^d$ needs to prefer his or her partner in~$M$
to~$m_{i_j}^d$ and $w_{i_{j}}^d$, since $\{m_{i_j}, w_{i_j}\}$ does not block 
$M$.
Thus, we
assume without loss of generality that~$w_{i_1}^d$ prefers~$m_{i_{2}}^d$
to~$m_{i_{1}}^d$.
Furthermore, there exists no $i_j$ such that both $w_{i_j}^d$ and
$m_{i_{j+1}}^d$ prefer their partners in~$M$ to~$m_{i_j}^d$ and $w_{i_{j+1}}^d$,
respectively, as this would
already
require $r $ swaps.
It follows that $w_{i_j}^d$ prefers $m_{i_{j+1}}^d$ to $m_{i_j}^d$ for every $j \in [s]$ (where $i_{s+1} := i_s$).
Define $\dist (p, p'):= p' - p \bmod r $.
To make  $w_{i_j}^d$ prefer  $m_{i_{j+1}}^d$ to $m_{i_j}^d$, one needs to
perform at least $\dist (i_j, i_{j+1})$ swaps in the
preference list of $w_{i_j}^d$ (where~$i_{s+1} = i_1$).
	Summing over the number of swaps for $w_{i_1}, \dots, w_{i_s}$, we get that at least $r $ swaps
have been performed, a contradiction.

In the second case, we may assume without loss of generality by the same 
argument as in the first
case that all women $w_{i_j}^d$ prefer $m_{i_{j+1}}^d$ to $m_{i_j}^d$,
but
$m_{i_{j+1}}^d$ does not prefer~$w_{i_{j}}^d$ to~$w_{i_j+1}^d$ for $j \in
[s-1]$.
However, as otherwise $\{w_{i_s}^d, m_{i_s}^d\}$ forms a blocking pair, this 
implies that~$w_{i_s}^d$ prefers $m$ to $m_{i_s}^d$, which needs
at least~$r $ swaps, a contradiction.
\end{proof}

We now use
\Cref{lem:dummy-agents} to model
\IswapR by \Iswap by adding sufficiently many agents with preferences as in 
\Cref{lem:dummy-agents} such that any swap not allowed in the \IswapR instance drastically
exceeds the budget in the \Iswap instance:

\begin{lemma}\label{lem:swapr}
Let $(\mathcal{I}=(U, W, \mathcal{P}), \hat{A})$ be a \IswapR instance with 
$n$ men and $n $
women, and~$4\le c\in \mathbb{N}$.
One can create a \textsc{Stable Marriage} instance $\mathcal{I}'$ by
adding~$2n^{c} (n-1)$ dummy agents~$m_1^d, \dots, m_{n^{c} (n-1)}^d$, $w_1^d,
\dots, w^d_{n^{c} (n-1)}$ such that
\begin{itemize}
\item[(a)] all swaps which are allowed in $\mathcal{I}$ are also possible in 
$\mathcal{I}'$ (i.e., for every agent~$a$ which is allowed to swap agents $b$ 
and $b'$ in $\mathcal{I}$, there is no agent between $b$ and $b'$ in the 
preferences of agent~$a$ in $\mathcal{I}'$),
\item[(b)] for any list $S$ of allowed swap operations in $\mathcal{I}$, it
holds that
$\mathcal{M}_{\mathcal{I}'[S]} = \{M \cup \{\{m_i^d, w_i^d\} : i\in [n^c (n-1)]\}
:
M \in \mathcal{M}_{\mathcal{I}[S]}\}$, and 
\item[(c)] for any list $S'$ of at most $n^c$ swap operations in the instance
$\mathcal{I}'$, it holds that $\mathcal{M}_{\mathcal{I}[S]} = \{ M|_{U\cup
W} : M\in \mathcal{I}[S']\}$, where $S$ is the sublist of~$S'$ containing the swaps in $S'$
between the two most preferred agents of some agent $a\in \hat{A}$ in
$\mathcal{I}$.
Furthermore, any $M \in \mathcal{I}[S']$ contains the edge $\{m_i^d,
w_i^d\}$ for~$i\in[n^c(n-1)]$.
\end{itemize}

\end{lemma}

\begin{proof}
Let $(\mathcal{I}=(U, W, \mathcal{P}),\hat{A})$ be an instance of \IswapR.
We modify the given SM instance~$\mathcal{I}$ to obtain a new SM instance
$\mathcal{I}'$ by adding $n^{c } (n-1)$ additional men~$m_1^d, \dots, m^d_{
n^{c}(n-1)}$ and~$n^{c}(n-1) $ additional women $w^d_1, \dots,
w^d_{n^{c}(n-1)}$.
For the rest of the proof, all indices are taken modulo $n^{c}(n-1)$.
The preferences of these men and women are as follows:
\begin{align*}
    m_i^d &: w_i^d \succ w^d_{i+1} \succ w^d_{i+2} \succ \dots \succ w^d_{n^{c}
    (n-1) + i-1}\pend,\\
    w^d_i &: m^d_i \succ m^d_{i+1} \succ m^d_{i+2} \succ \dots \succ m^d_{n^{c}
    (n-1) + i-1}\pend.
\end{align*}

For a man $m\in U\setminus \hat{A}$ with $m: w_1 \succ w_2 \succ \dots \succ
w_n$ and a woman $w\in W\setminus \hat{A}$ with~$w : m_1 \succ m_2 \succ
\dots\succ m_n$, their modified preferences look as
follows:
$$m : w_1 \succ w_1^d \succ  w_2^d \succ \dots \succ w_{n^c}^d
    \succ w_2 \succ w_{n^c + 1}^d 
    \succ w^d_{n^c+2} \succ \dots \succ
    w_{2n^c}^d \succ w_3 \succ \dots \succ w_n,$$
    $$
    w : m_1 \succ m_1^d \succ m_2^d \succ \dots \succ m_{n^c}^d
    \succ m_2 \succ m_{n^c + 1}^d \succ m_{n^c + 2}^d \succ \dots \succ
    m_{2n^c}^d \succ m_3 \succ \dots \succ m_n.
    $$
    For a man $m\in U \cap \hat{A}$ with $m: w_1 \succ w_2 \succ \dots \succ
    w_n$ (and analogously for a woman~$w\in W\cap \hat{A}$), his modified
    preferences look as follows:
    \begin{multline*}
    m : w_1
    \succ w_2 \succ w_{n^c + 1}^d \succ w^d_{n^c+2} \succ \dots \succ
    w_{2n^c}^d 
    \succ w_3 \succ w_{2n^c+1} \succ \dots \succ w_{3n^c} \\
    \succ w_4
    \succ \dots \succ w_n
    \succ w_1^d \succ w_2^d \succ \dots \succ w_{n^c}^d.
    \end{multline*}

We now prove that the constructed SM instance $\mathcal{I}'$ indeed fulfills 
the three conditions from the lemma. Part (a) is obvious.
Part (b) is also clearly fulfilled, as $m_i^d$ and $w_i^d$ are their mutual top 
choices in $\mathcal{I}'[S]$ for
$i\in [n^c(n-1)]$ and therefore are matched together by any stable matching in  
$\mathcal{I}'[S]$ and can never form a blocking pair with another agent
(dummy agents can never be part of a swap from $S$, as they are not part of the 
instance $\mathcal{I}$).

To prove part (c), let $S'$ be any set of at most $n^c$ swap operations
in~$\mathcal{I}'$.
By \Cref{lem:dummy-agents}, any stable matching in~$\mathcal{I}'[S']$
contains the edges $\{m_i^d, w_i^d\}$ for $i\in [r]$.
As any change which does not involve at least one agent~$m_i^d$ or~$w_i^d$
swaps the two most preferred agents of some agent $a\in \hat{A}$, the lemma
follows.
\end{proof}

Combining the hardness result from \Cref{cor:const-ex-swapr} and the
construction from \Cref{lem:swapr} modeling \IswapR by \Iswap, we now prove
the following inapproximability result.
Note that the inapproximability is tight in the sense that there is always a
solution using at most~$2(n-1)$ swaps which, given the edge~$\{m^*, w^*\}$ that shall be contained in a stable matching, just swaps $m^*$ to be the most
preferred man of $w^*$ and swaps~$w^*$ to be the most preferred woman of $m^*$.

\begin{theorem} \label{th:const-ex-swap}
Unless FPT = W[1], \const-\ex-\swap does not admit
an~$\bigO(n^{1-\epsilon})$-approximation in~$f(\ell) n^{\bigO(1)}$ time for 
any $\epsilon>0$ and any computable function $f$. This also holds if
one is only allowed to swap the
first two men in the preferences of women.
\end{theorem}

\begin{proof}
Fix~$\epsilon > 0$.
We assume without loss of generality that $\epsilon = 4c^{-1}$ for some $c \in \mathbb{N}$ (if this is not the case, then we can choose a smaller $\epsilon$ fulfilling this condition).

We reduce from \const-\ex-\swapR, for which it is W[1]-hard to distinguish 
whether
there is a solution with at most~$\ell$ manipulations or no solution for any 
number of allowed swaps 
	(\Cref{cor:const-ex-swapr}).
Let $(\mathcal{I}, \hat{A}, \ell)$ be an instance of \const-\ex-\swapR with at most~$n_R$ men and at most~$n_R$ women.
By \Cref{lem:swapr}, we can construct an instance of \const-\ex-\swap which
contains a solution of size at most~$\ell$ if and only if the corresponding
instance of \const-\ex-\swapR contains a solution of size at most~$\ell \le n_R^3$, and
otherwise any solution has size at least $n^c_R$.
Let $n = n_R^c (n_R^c - 1) + n_R^c = \Theta (n_R^{c+1} )$ be the number of men in the \const-\ex-\swap
instance.
Thus, given an instance of 
\const-\ex-\swapR, an~$\bigO(n^{1-\epsilon})$-approximation for 
\const-\ex-\swap allows to
decide whether there
exists a solution of size~$\bigO(\ell n ^{1-\epsilon}) \le \bigO(n_R^{(c+1)
(1-\epsilon) + 3}) = \bigO( n_R^{ ( \frac{4}{\epsilon} +1)(1 -\epsilon) +3} ) = \bigO(
n_R^{\frac{4}{\epsilon} - 4 + 1 - \epsilon +3}) = \bigO(n_R^{\frac{4}{\epsilon} -
\epsilon}) = \bigO(n_R^{c-\epsilon}) <
n_R^{c}$ in the given instance (the last inequality holds 
only for sufficiently 
large $n_R$).
Therefore, assuming that we have an $\bigO(n^{1-\epsilon})$-approximation algorithm for \const-\ex-\swap running in~$f(\ell) n^{\bigO(1)}$ time for some computable function~$f$ allows to decide \const-\ex-\swapR in $f(\ell) n_R^{\bigO(1)}$ time, which implies by
\Cref{cor:const-ex-swapr} that FPT = W[1].
\end{proof}

\subsection{An XP Algorithm for Constructive-Exists-Reorder}
\label{section:const-ex-reor}
Observe that for all manipulative actions except \Ireor, the membership of  
\const-\ex-$\mathcal{X}$ parameterized by $\ell$ in XP follows from a straightforward 
brute-force algorithm. We now show that \const-\ex-\reor also lies in XP 
parameterized by the budget $\ell$ using a simple algorithm.
\begin{proposition}\label{pr:const-ex-reorXP}
    \const-\ex-\reor can be solved in $\mathcal{O}(2^\ell n^{2\ell + 2})$ time.
\end{proposition}
\begin{proof}
    We guess the set $X$ of $\ell $ agents whose preferences are modified, and 
    for
    each agent~${a\in X}$, we guess an agent $T(a)$ to which $ a$ is matched in a
    stable matching containing $\{m^*,w^*\}$.
    For each $a \in X$, we modify the preferences of $a$ such that $T(a)$ is 
    its top-choice.
    We check whether the resulting instance contains a stable matching 
    containing the edge $\{m^*, w^*\}$.
    If any guess results in a stable matching containing $\{m^*, w^*\}$, then we 
    return YES, and NO otherwise.
    
    The running time of the algorithm is $\mathcal{O}(2^\ell n^{2\ell +2})$, as 
    we
    guess a set of $\ell$ out of $2n $ agents whose preference are modified, and
    for each of these $\ell$ agents we guess to which agent it is matched to.
    Checking whether there is a stable matching containing $\{m^*, w^*\}$ can 
    be done in $\bigO(n^2)$~time for each guess~\citep{Gusfield87}.
    
    It remains to show the correctness.
    If the algorithm returns YES, then the instance is clearly a YES-instance.
    To prove the other direction, assume that there exists a set $Y$ of at 
most $\ell$ agents such that in an
    instance $\mathcal{I}_Y$ arising through reordering the preferences 
    of~$Y$
    there exists a stable matching $M$ containing the edge $\{m^*, w^*\}$. Then,
    there exists a guess~$(X, T)$ for our algorithm such that $Y\subseteq X$ 
and $T(a) = M(a)$ for
    all $a\in X$.
    We claim that for the instance constructed by the algorithm using this 
    guess,
    the matching $M$ is stable, and therefore the algorithm returns YES.
    Note that no agent from $X$ can be contained in a blocking pair, as they 
    are all matched to their top-choice.
    However, any blocking pair containing no agent from $X$ would also be a
    blocking pair for $M$ in $\mathcal{I}_Y$, and therefore, no blocking pair
    exists.
\end{proof}

\subsection{Polynomial-Time Algorithms}\label{section:const-ex-poly}
Having encountered computational hardness for all 
manipulative actions but \Idelete, we now give
a polynomial-time algorithm for \const-\ex-\delete.
We then use this algorithm to derive a 2-approximation algorithm for 
\const-\ex-\reor.
The polynomial-time algorithm for
\Idelete might be particularly surprising because of the strong hardness result
that we derived for the related manipulative action \Iadd. However, the reason
for this is that, in a \const-\ex-$\mathcal{X}$ instance, we can compute a set 
of ``conflicting agents'' that prevent that the pair $\{m^*,w^*\}$ is part of 
a stable matching, and one
\Idelete operation can only decrease the number of conflicting agents by at 
most one (and, in fact, by deleting a conflicting agent, we can also always 
decrease the number of conflicting agents by one).
For \Iadd, however, there are
presumably multiple ways or no way how one could resolve a conflicting 
agent. Similarly, we
derived an inapproximability result for
\Iswap, while the seemingly similar manipulative action \reor admits a factor-2
approximation in the \Iconst-\Iex setting. The reason for this is that the
\Ireor operation is more
powerful, allowing to create in some sense trivial solutions again by 
identifying a set of ``conflicting agents''.
These trivial solutions then consist of reordering the preferences of all 
conflicting agents and we prove that this is a 2-approximation. Thus,
compared to the \Iswap setting where one can save quite some manipulative
actions by choosing the swaps to perform, here it does not make a fundamental
difference how
the agent's preferences are reordered. We start by presenting the polynomial-time 
algorithm for \Idelete, and then describe the factor-2 approximation algorithm 
for \Ireor. 

\subsubsection{\Idelete}\label{section:const-ex-del}

In sharp contrast to the hardness results for all other considered manipulative
actions,
there is a simple algorithm solving a given instance
of \const-\ex-\delete consisting of an SM instance $\mathcal{I}=(U, W, 
\mathcal{P})$ together with 
a 
man-woman pair $\{m^*,w^*\}$ and an integer~$\ell$ in time linear in the size of the input.
The algorithm is based on the following observation.
Let $W^*$ be the set of women preferred by $m^*$ to $w^*$, and $U^*$ the set of 
men preferred by $w^*$ to $m^*$.
In every stable matching~$M$ which includes
$\{m^*,w^*\}$, every woman in $W^*$ needs to be matched to a man
whom she prefers to~$m^*$, or she needs to be deleted. Analogously, every man
in $U^*$ needs to be matched to a woman which he prefers to $w^*$, or he needs to 
be deleted.
Consequently, all pairs consisting of an agent $a\in U^* \cup W^*$ and 
an agent~$a' $ which~$a$ does not prefer to $w^*$ or $m^*$ cannot be part of 
any stable matching.
After deleting these pairs and the agents $m^*$ and $w^*$, consider a stable 
matching~$M$ in the resulting instance.
If we add the pair $\{m^*,w^*\}$ to $M$, then in the original instance 
$\mathcal{I}$ each woman $w\in W^*$ not assigned by~$M$ will always 
create a blocking pair $\{m^*, w\}$ (independent of how and whether we extend 
$M$). Similarly,  each unassigned man $m\in U^*$ will always create a  
blocking pair $\{m , w^*\}$. Thus, the agents from $U^* \cup W^*$ that are 
unassigned in $M$ will be the set of ``conflicting agents''.
This observation motivates the following algorithm:
\begin{enumerate}
  \item Let $U^*$ be the set of agents which $w^*$ prefers to $m^*$.\\
  Let $W^*$ be the set of agents which $m^*$ prefers to $w^*$.
  \item Delete all edges
  $\{m, w\}$ such that $m \in U^*$ and $m$ does not prefer $w$ to $w^*$.\\
  Delete all edges $\{m, w\}$ such that $w \in W^*$ and $w$ does not prefer $m$ to $m^*$.\\
  Delete the agents $m^*$ and $w^*$.\\
  Call the resulting SMI instance~$\hat{\mathcal{I}}$.
  \item Compute a stable matching~$\hat M$ in $\hat{\mathcal{I}}$.
  \item Return the set~$A'$ of agents from $U^* \cup W^*$ which are unassigned 
in~$\hat M$ as the agents to be deleted.
\end{enumerate}

To show the correctness of this algorithm, we first investigate the influence of
a \Idelete operation on the set of agents matched in a stable matching.

\begin{lemma} \label{le:SMI-del}
    Let~$\mathcal{I}'$ be an \textsc{SMI} instance and~$a\in A$ some agent.
    Then,
    there exists at most one agent~$a'\in
    A$ which was unassigned in~$\mathcal{I}'$, i.e.,~$a'\notin \ma(\mathcal{I}')$, and is matched 
    in~$\mathcal{I}'\setminus\{a\}$, i.e.,~$a'\in
    \ma(\mathcal{I}'\setminus\{a\})$.
\end{lemma}
\begin{proof}
    Let~$M'$ be a stable matching in~$\mathcal{I}'$ and~$M^a$ be a stable
    matching in~$\mathcal{I}'\setminus\{a\}$.
    We prove that~$|\ma(M^a)\setminus \ma(M')|\leq 1$,
    which clearly implies the lemma.
    To do so, we examine the symmetric
    difference~$S$ of the two matchings~$M'$ and~$M^a$.
    The set~$S$ is a union of paths and even-length cycles. We do not care 
about cycles, as all agents in such a
    cycle are matched in both~$M'$ and~$M^a$.
    Therefore, we turn to paths.

    First, of all, we claim that all paths include~$a$.
    Assume that there exists a path~$P=(a_1,\dots,a_{k+1})$ for some $k\in
    \mathbb{N}$ not
    including~$a$,
    starting
    without loss of generality with a pair from~$M'$. Then, as~$\{a_1,a_2\}$
    is not a
    blocking
    pair in~$M^a$ and $a_1$ is unassigned in $M'$, it holds 
that~$a_3\succ_{a_2} a_1$. As
    $\{a_2,a_3\}$ is not a blocking pair in~$M'$, it holds that~$a_4\succ_{a_3}
    a_2$. Consequently, it holds that~$a_{k+1}\succ_{a_{k}} a_{k-1}$. If~$k$ is
    even, then~$a_{k+1}$ is unassigned in~$M'$, and thus
    $\{a_{k},a_{k+1}\}$ is a blocking pair in~$M'$, contradicting the stability
    of~$M'$. Otherwise,~$a_{k+1}$ is unassigned in~$M^a$, and~$\{a_{k},a_{k+1}\}$
    is blocking in~$M^a$.

    From this it follows that there exists at most one maximal path~$P$, which 
needs to involve agent~$a$.
    As agent~$a$ cannot be part of~$M^a$, it needs to be one of
    the endpoints of~$P$.
    Consequently, the only agent that could be
    matched in~$M^a$ but not matched in~$M'$ is the other endpoint of~$P$.
\end{proof}

\noindent Using \Cref{le:SMI-del}, we now show the correctness of the 
algorithm.
\begin{theorem} \label{const-ex-delete}
    \const-\ex-\delete is solvable in~$\mathcal{O}(n^2)$ time.
\end{theorem}

\begin{proof}
    Since a stable matching in an SMI instance can be computed in $\bigO(n^2)$ 
    time~\citep{GaleS62}, the set $A'$ clearly can be 
    computed in $\bigO(n^2)$.
    We claim that the given \const-\ex-\delete instance is a YES-instance if 
and only if 
    $|A'|\le \ell$.

    First assume $|A'| \le \ell$.
    Let $\hat M$ be the stable matching in $\hat{\mathcal{I}}$ computed by the algorithm.
    We define $M :=\hat M \cup \{\{m^*, w^*\}\}$, and claim that $M$ is a stable 
    matching in 
    $\mathcal{I}\setminus A'$, showing that the given \const-\ex-\delete 
instance is a YES-instance.
    For the sake of a contradiction, assume that there exists a blocking 
    pair~$\{m, w\}$.
    Note that~$\{m, w\}$ contains neither $m^*$ nor $w^*$ since all agents which
    $m^*$ and~$w^*$ prefer to each other are either deleted or matched to an 
    agent they
    prefer to $m^*$ and $w^*$.
    Since $\{m, w\}$ is not a blocking pair in $\hat{\mathcal{I}}\setminus A'$, it 
    contains an agent~$a$ from $(U^* \cup W^*)\setminus A'$, and $a$ 
prefers 
$w^*$ to~$w$ if~$a= m$ or $a$ prefers $m^* $ to $m$ if $a = w$.
	We assume without loss of generality that~$a = m$.
    As~$m$ is matched in $\hat M$, he prefers $\hat M(m) = M(m)$ to $w^*$.
    Thus, $m$ prefers~$M(m)$ to $w$, a contradiction to $\{m,w\}$ being blocking
    for $M$.

    Now assume that $|A'| > \ell$.
    For the sake of contradiction, assume that there exists a
    set of agents~$B'=\{b_1,\dots,b_k\}$ with $k\le \ell$ such that 
$\mathcal{I}\setminus B'$ admits a stable matching~$M$ containing~$\{m^*, w^*\}$.
    For each~$i\in \{0, 1, \dots, k\}$, let~$\hat{M}_i$ be a stable matching
    in~$\hat{\mathcal{I}}\setminus \{b_1, \dots, b_i\}$.
    By the definition of $A'$, all agents from $A'$ are unassigned in $\hat{M}_0$.
    Note that each agent~$a\in A' \subseteq U^* \cup W^*$ is either part 
of~$B'$ or 
prefers $M(a)$ to $m^*$ or $w^*$ due to the stability of $M$; in particular, 
$a$ is either contained in $B'$ or matched in $M$.
    Since~$k\le \ell < |A'|$, there exists
    an~$i$ such that there exist two
    agents~$a,a'\in A'$ which are unassigned in~$\hat{M}_{i-1}$ and not contained
    in~$\{b_1, \dots, b_{i-1}\}$ but matched in~$\hat{M}_i$ or contained in~$
    \{b_1,\dots b_{i}\}$.
    It follows from \Cref{le:SMI-del} that it is not possible that both~$a$
    and~$a'$ are unassigned in~$\hat{M}_{i-1}$ but matched in~$\hat{M}_i$. Consequently,
    without loss of generality it needs to hold that~$a = b_i$ with~$a$ being
    unassigned in $\hat{M}_{i-1}$, and~$a'\in
    \ma(\hat{\mathcal{I}}\setminus
    \{b_1,\dots b_{i}\}) \setminus \ma(\hat{\mathcal{I}}\setminus
    \{b_1,\dots b_{i-1}\})$.
    However, deleting an agent that was previously unassigned does not change the set
    of matched agents, i.e.,~$\ma(\hat{\mathcal{I}}\setminus
    \{b_1,\dots b_{i-1}\})=\ma(\hat{\mathcal{I}}\setminus
    \{b_1,\dots b_{i}\})$, since a matching that is stable 
in~$\hat{\mathcal{I}}\setminus
    \{b_1,\dots b_{i-1}\}$ is also stable in~$\ma(\hat{\mathcal{I}}\setminus
    \{b_1,\dots b_{i}\})$ (deleting unassigned agents cannot
    create new blocking pairs).
This contradicts~$a'\in \ma(\hat{\mathcal{I}}\setminus
    \{b_1,\dots b_{i}\}) \setminus \ma(\hat{\mathcal{I}}\setminus
    \{b_1,\dots b_{i-1}\})$.
\end{proof}

The algorithm of \Cref{const-ex-delete} 
directly extends to the setting where an arbitrary number
of edges is given and the goal is to delete agents such that there exists a
stable matching which is a superset of the given set of edges.

\subsubsection{Factor-2 Approximation for \Ireor} \label{subsub:const-reor}

We now follow  a similar approach as for \const-\ex-\delete to 
construct a
factor-2 approximation for the optimization version of \const-\ex-\reor 
(notably, this algorithms crucially relies on our assumption that $|U|=|W|$).
We construct an instance $\hat{\mathcal{I}}$ identically as in the case of 
\const-\ex-\delete:
Let $W^*$ be the set of women preferred by $m^*$ to $w^*$, and $U^*$ 
be the set of
men preferred by~$w^*$ to~$m^*$.
In every stable matching~$M$ which contains
$\{m^*,w^*\}$, every woman which~$m^*$ prefers to~$w^*$ needs to be matched to a 
man
which she prefers to~$m^*$, or her preference list needs to be reordered. Analogously, every man
which~$w^*$ prefers to~$m^*$
needs to be matched to a woman which he prefers to $w^*$ or his preferences need to be reordered.
Consequently, all pairs consisting of an agent $a\in U^* \cup W^*$ and
an agent~$a' $ which $a$ does not prefer to $w^*$ or~$m^*$ cannot be part of
any stable matching. This observation motivates a transformation of the
given \textsc{SM} instance~$\mathcal{I}$ to a \textsc{SMI} 
instance~$\hat{\mathcal{I}}$ through the deletion of all such pairs.
We also delete~$w^*$ and $m^*$ from $\hat{\mathcal{I}}$ and compute a stable
matching~$M$ in the resulting instance.

We observe
analogously to \Cref{le:SMI-del} that by reordering the preferences of an
agent,
at most two previously unassigned agents become matched in an
\textsc{SMI} instance:
\begin{lemma}\label{le:const-ex-reor-2}
    Let~$\mathcal{I}'$ be an \textsc{SMI} instance, $a^*\in A$ some agent, and
    let~${\mathcal{I}}^{a^*}$ denote the instance arising from~$\mathcal{I}'$ by
    reordering and
    extending
    $a^*$'s
    preferences arbitrarily. Then, there
    exists at most one man~$m\in U$ and at most one woman~$w\in W$ who are
    unassigned in
    $\mathcal{I}'$, i.e.,~$m,w\notin \ma(\mathcal{I}')$, and who are matched in
    $\mathcal{I}^{a^*}$, i.e.,~$m,w\in
    \ma({\mathcal{I}}^{a^*})$.
\end{lemma}
\begin{proof}
    The proof proceeds analogously to the proof of
    \Cref{le:SMI-del}. Let~$M'$ be some stable matching in
    $\mathcal{I}'$ and~$M^{a^*}$ a stable matching 
in~$\mathcal{I}^{a^*}$. 
    We examine again the symmetric
    difference of $M'$ and ${M}^{a^*}$ and
    conclude that only the unique maximal path including~$a^*$ can change the
    set of matched agents. As this path has only two endpoints, at
    most two agents which are not matched in~$M'$ can become matched in~${M}^{a^*}$.
    Assuming that the path has an even number of edges, only one of the
    endpoints can be matched
    in~${M}^{a^*}$. 
    Assuming that the path has an odd number of edges, one of
    the endpoints needs to correspond to a woman and one to a man.
\end{proof}
Now, similarly to \Cref{const-ex-delete}, it is possible to construct a
straightforward solution, which matches one agent from $U^* \cup W^*$
which is currently
not matched in any stable matching in $\hat{\mathcal{I}}$ by reordering the
preferences of one agent.
Using \Cref{le:const-ex-reor-2}, we now show that this approach yields a
factor-2 approximation of the optimal solution:
\begin{proposition}\label{th:const-ex-reor-2}
    One can compute a factor-2 approximation of the optimization
    version of
    \const-\ex-\reor in
    $\mathcal{O}(n^2)$ time.
\end{proposition}
\begin{proof}
    Given an instance of \const-\ex-\reor consisting of an SM
    instance~$\mathcal{I}=(U,W,\mathcal{P})$, budget~$\ell$, and the pair~$\{m^*, 
w^*\}$, we construct an \textsc{SMI} instance~$\hat{\mathcal{I}}$ and
    define~$U^*$ and $W^*$ as described above. The
    2-approximation algorithm proceeds as
    follows (recall that $\ma(M)$ is the set of agents assigned by matching 
    $M$):
    \begin{itemize}
        \setlength\itemsep{0em}
        \item Compute a stable matching~$\hat{M}$ in~$\hat{\mathcal{I}}$.
        \item Let $A^* := (U^* \cup W^*) \setminus \ma (\hat{M})$ and $A' := A 
\setminus \ma (\hat{M})$.
        Reorder the preferences of each agent from~$A^*$ such that all agents from $A'$ of opposite gender are in the beginning of the preferences.
    \end{itemize}

    Let $N$ be a stable matching on the instance restricted to agents from $A'$ 
    (with modified preferences).
    Note that the number of men which $\hat{M}$ leaves unassigned equals the number
    of women which $\hat{M}$ leaves unassigned (here we use that $|U| = |W|$), and thus, $N$ matches every agent 
    from $A'$. 
    The correctness of the solution returned by the algorithm and the
    approximation factor can be proven in a way similar to the proof of
    \Cref{const-ex-delete}:

    First, we show that~$M' := \hat{M} \cup N \cup \{\{m^*, w^*\}\}$ is a stable
    matching after the described reorderings of the preferences.
    No agent from~$A^*$ is part of a blocking pair, as $N$ is stable and every agent from $A^*$ prefers every agent from $A'$ of opposite sex to every agent from~$A\setminus A'$.
    No agent from $A' \setminus A^*$ is part of a blocking pair, as $N$ is stable and no agent from $A\setminus A'$ prefers an agent from $A'$ to its partner in $\hat M$.
    By the same arguments as in \Cref{const-ex-delete},
    no
    agent from $A \setminus A'$ is part of a blocking pair.

    It remains to show that at least~$q:=\frac{|A^*|}{2}$ reorderings are needed.
    Assume that less than~$q$ reorderings are needed, and let~$B = \{b_1, 
    \dots, b_k\}$ with $k< q$ be the set
    of agents whose preferences have been reordered;
    we refer to the SM instance arising through these reorderings as 
$\mathcal{I}^*$.
    Let $M^*$ be a stable matching containing $\{m^*,w^*\}$ in the instance~$\mathcal{I}^*$.
    Let $\hat{M}_i$ be a stable matching in the instance $\hat{\mathcal{I}}_i$ arising 
    from $\hat{\mathcal{I}}$ by replacing the preference list of~$b_1, \dots, b_i$ 
    by their reordered, complete preferences (this includes adding for each
    agent~$a\in A\setminus \{b_1, \dots, b_i\}$
    such that $\{a, b_j\}$ is not contained in $\hat{\mathcal{I}}$ the agent $b_j$
    in $a$'s preferences at its position in $a$'s preferences in 
$\mathcal{I}$).
    Note that $M^*$ is a stable matching in~$\hat{\mathcal I}_k$ (if $M^*$ contained 
    an edge $\{m, w\}$ not contained in $\hat{\mathcal{I}}_k$, then $m\in U^*$ and 
    $m$ would prefer~$w^*$ to~$M^* (m)$, implying that $\{m, w^*\}$ is a 
    blocking pair, or $w\in W^*$ and $w $ would prefer $m^*$ to~$M^* (w)$, 
    implying 
    that $\{m^*, w\}$ is a blocking pair).
    Then, every agent from $A$ is matched in~$\hat{M}_k$.
    By the definition of $A^*$, all agents from $A^*$ are unassigned in $\hat{M}_0$.
As it holds that~$k < \frac{|A^*|}{2}$ and all agents from $A^*$ are 
unassigned in a stable matching $\hat{\mathcal{I}}_0$, there needs to exist 
some~$j\in 
\{0, 1, \dots, k-1\}$ such that three agents from $A^*$ that are unassigned in 
a stable 
matching in  $\hat{\mathcal{I}}_{j}$ are matched in a stable matching in  
$\hat{\mathcal{I}}_{j+1}$. This contradicts \Cref{le:const-ex-reor-2}.
    It follows that the assumption~$|B| < \frac{|A^*|}{2}$ is wrong, and thus, 
    the algorithm computes a 2-approximation.
\end{proof}

Note that the above approach does not directly carry over to the case $|U | \neq |W|$.
The problem is that the matching $M'$ constructed in the proof of 
\Cref{th:const-ex-reor-2}
is 
not 
necessarily stable in this case. The reason for this is that $|U \cap A^*| > |W 
\cap A'|$ (or $|W \cap 
A^*| > |U \cap 
A'|$) might hold. In this case, there remains at least one unassigned man 
from $U \cap A^*$ (or at least one unassigned woman from $W \cap A^*$) in $M'$ 
which 
then 
forms a 
blocking pair with~$w^*$ (or~$m^*$) for $M'$.
In fact, an optimal solution might be much larger than $|A^*|$, showing that better lower bounds are needed to design a constant-factor approximation for the general case.
For example, consider an instance of \const-\ex-\reor with~$|U| = |W| +1$, 
where 
$w^*$ prefers all but one man $m'$ to $m^*$, and every other woman has~$m'$ as 
her top-choice.
Furthermore, every man but $m^*$ has $w^*$ as his last choice, while~$m^*$ has 
$w^*$ as his top-choice.
Then $\hat{\mathcal{I}}$ arises through the deletion of $m^*$ and $w^*$, and every 
stable matching in~$\hat{\mathcal{I}}$ leaves exactly one man~$m \in 
U\setminus\{m^*, 
m'\}$ unassigned.
Consequently, we have~$A^* = \{m\}$.
However, every stable matching has to match every man from $U\setminus \{m^*, 
m'\}$ to a woman from $W \setminus \{w^*\}$, implying that $m'$ will be the 
only man unassigned in a stable matching.
Therefore, any optimal solution has to reorder the preferences of every woman 
from~$W\setminus \{w^*\}$.

\begin{remark}[Destructive-Exists] \label{remark} As already briefly discussed 
in the introduction, instead of considering \Iconst-\Iex it is also 
possible to 
consider \Idest-\Iex where given a SM instance and a man-woman pair 
$\{m^*,w^*\}$, we want to 
alter 
the SM
instance such that~$\{m^*,w^*\}$ is not part of at least one stable matching. 
Notably, polynomial-time algorithms for \Iconst-\Iex carry over to 
\Idest-\ex, as we can solve the latter problem by running the algorithm for 
\Iconst-\Iex for all man-woman pairs involving one of $m^*$ and $w^*$. 
Moreover, it is also possible to adapt our hardness 
reduction for \const-\ex-\add to \dest-\ex-\add: We introduce an additional 
man~$m^{**}$ which has $w^{*}$ as his top-choice and change the preferences of 
$w^*$ to~$w^* \colon m^*\succ m^{**}\succ \dots$ and set $\{m^{**},w^*\}$ to 
the pair that 
we want to exclude from some stable matching. Notably, a stable matching does 
not include $\{m^{**},w^*\}$ if and only if it includes $\{m^*,w^*\}$. After 
the described modifications, the reduction no longer has any implications 
concerning inapproximability, as we need to use that our budget is $\ell=k+ 
\binom{k}{2}$ in 
the proof of correctness of the backward direction of the 
reduction (\Cref{le:const-ex-add-L2}), which needs to be slightly adapted. 
Thus, we 
can conclude that \dest-\ex-\add parameterized by $\ell$ is W[1]-hard. By 
modeling \Iadd by \Iswap as described in 
\Cref{sec:const-ex-swap}, we can also conclude that \dest-\ex-\swap 
parameterized by $\ell$ is W[1]-hard. 
\end{remark}

\section{\Iexact-\Iex} \label{se:ex-ex}
In this section, we aim to make a given matching in an SM
instance stable by 
performing some manipulative actions. The difference to the 
\Iconst-\Iex setting considered in the previous section is that now
instead of one edge that should be included in some stable matching, a complete 
matching is given which shall be made stable. At first sight, it is unclear 
whether making the goal more specific in the sense of providing the complete 
matching instead of just one edge in some matching makes the problem easier or 
harder. In this section, we prove 
that providing this additional information makes the problem usually 
easier, as 
for all manipulative actions~$\mathcal{X}$ 
for which we showed hardness in the previous 
section, \exact-\ex-$\mathcal{X}$ becomes polynomial-time 
solvable.\footnote{The only manipulative action for which \Iexact-\Iex is 
harder than \Iconst-\Iex is \Idelete. However, recall that we came up with a 
modified definition of \Iexact-\Iex for \Idelete. As this definition makes 
\exact-\ex-\delete quite different from the problem for the other actions and 
also from \const-\ex-\delete, this observation does not contradict our previous 
claim.}
The intuitive reason for this difference is that the problem of making a 
given matching $\Mst$ stable simplifies to ``resolving'' all pairs that are 
blocking 
for 
$\Mst$, which turns out to be solvable in polynomial-time. In contrast to 
this, for  
\const-\ex-$\mathcal{X}$, we also need to decide which matching including the 
given pair we want to make stable, a
task which turned out to be hard for most manipulative actions.  

Next, we start by 
considering the actions \IaccD, \Ireor, and \Iswap before turning to the 
manipulative actions \Idelete and \Iadd for which we need to adapt the problem 
definition.

\subsection{Polynomial-Time Algorithms for \IaccD, \Ireor,  and \Iswap}
\label{se:ex-ex:ap}
As argued above, in the \Iexact-\Iex setting, we need to 
``resolve'' all pairs 
that are blocking for the given matching $\Mst$. This requires manipulating the 
preferences
of at least one agent~$a$ in each blocking pair such that it no longer prefers
the 
other agent in the blocking pair to $\Mst(a)$. For a 
matching~$M$ and an SM instance~$\mathcal{I}$, we denote by~$\bp(M,\mathcal{I})$
the set of all blocking pairs of~$M$ in~$\mathcal{I}$. For a blocking 
pair~$\beta=\{m,w\}\in 
\bp(M,\mathcal{I})$ and an agent~$a\in \beta$, we denote by~$\beta(a)$ the 
other 
agent in
the 
blocking pair.

The optimal solution for an instance of \exact-\ex-\accD is to 
delete the 
acceptability of all blocking pairs. The set of blocking pairs can be computed 
in~$\mathcal{O}(n^2)$ time. This solution is optimal since it is always necessary 
to delete the acceptability 
of all blocking pairs and, by doing so, no new pairs will become blocking. 

\begin{observation} \label{th:ex-ex-accD}
    \exact-\ex-\accD is solvable in~$\mathcal{O}(n^2)$ time. 
\end{observation}

We now turn to \Ireor.
Since it is possible to resolve all blocking pairs containing an agent
by manipulating this agent, \exact-\ex-\reor corresponds to finding a 
minimum-cardinality
subset~$A'\subseteq U\cup W$ such 
that~$A'$ covers~$\bp(\Mst,\mathcal{I})$, i.e., each pair that 
blocks~$\Mst$ in 
the given SM 
instance~$\mathcal{I}$ contains an agent from $A'$.
From this observation, the following proposition directly follows. 
\begin{proposition} \label{th:ex-ex-reor}
    \exact-\ex-\reor reduces to finding a vertex cover in a bipartite
    graph and is, hence, solvable in~$\mathcal{O}(n^{2.5})$ time.
\end{proposition}
\begin{proof}
    Given an instance of \exact-\ex-\reor consisting of an SM instance
    $\mathcal{I}=(U, W, \mathcal{P})$ together with a matching~$\Mst$ and 
    budget $\ell$, we 
    construct a bipartite 
    graph 
    as follows. For each~$a\in A$, we introduce a vertex~$v_a$ and we connect 
    two vertices~$v_a,v_{a'}$ if~$\{a,a'\}\in \bp(\Mst,\mathcal{I})$. Note that
    the resulting graph is bipartite, as there cannot exist a blocking pair 
    consisting of two men or two women. We compute a minimum vertex cover 
    $X$ in the graph, i.e., a subset of vertices such that all edges are 
    incident to at least one vertex in $X$, using the Hopcroft–Karp algorithm 
    in 
    $\mathcal{O}(n^{2.5})$~time~\citep{DBLP:conf/focs/HopcroftK71}. For 
    each~$v_a\in X$, we reorder~$a$'s 
    preferences such that~$\Mst(a)$ becomes~$a$'s top-choice. After these
    reorderings,~$\Mst$ is stable, as for each blocking pair~$\beta$, for at 
    least
    one involved agent~$a\in \beta$, the agent~$\Mst(a)$ is now~$a$'s 
    top-choice, and no new blocking
    pairs are created by this procedure. Moreover, the computed solution is 
    optimal. For the sake of contradiction, let us assume that there exists a 
    smaller solution. Then, there exists an~$\{m,w\}\in \bp(\Mst,\mathcal{I})$ 
    where
    neither~$m$'s nor~$w$'s preferences have been modified. However, this 
    implies that~$\{m,w\}$ still blocks~$\Mst$.
\end{proof}

Now, we turn to the manipulative action \Iswap. Here, the 
cost of resolving a blocking pair~$\{m,w\}$ by manipulating $m$'s preferences is 
the number of swaps needed 
to swap~$\Mst(m)$ 
with $w$ in the preferences of $m$, and the cost of resolving the pair by manipulating $w$'s preferences 
is the number of swaps 
needed to swap $\Mst(m)$ with $w$ in the preferences of~$w$. This observation could lead to the conjecture 
that it is optimal to determine for each blocking pair the agent with the lower 
cost and then resolve the pair by performing the corresponding swaps.  
However, this approach is not optimal, as by resolving some blocking pair 
involving an agent also another blocking pair involving this 
agent might be resolved, as we observe it in the following example:
\begin{example} \label{ex:exact-ex-swap}
    Consider an instance of \exact-\ex-\swap consisting of the 
    following SM instance together
    with budget~$\ell=3$ and 
    $\Mst=\{\{m_1,w_3\}, \{m_2,w_2\}, \{m_3,w_1\}\}$:
    \begin{itemize}
        \setlength\itemsep{0em}
        \item~$m_1: w_1 \succ w_2\succ w_3$,
        \item~$m_2: w_2 \succ w_3\succ w_1$,
        \item~$m_3: w_2 \succ w_3\succ w_1$,
        \item~$w_1: m_1 \succ m_2\succ m_3$,
        \item~$w_2: m_1 \succ m_3\succ m_2$, and
        \item~$w_3: m_1 \succ m_2\succ m_3$.
    \end{itemize} The set of blocking pairs of~$\Mst$ is: 
    $\bp(\Mst,\mathcal{I})=\{\{m_1,w_1\}, \{m_1,w_2\}, \{m_3,w_2\}\}$.  In this 
    example, the cost to resolve the pair~$\{m_1,w_1\}$ by modifying~$m_1$'s 
    preference list is two, i.e., swap 
    $w_3$ 
    and $w_2$ and subsequently~$w_3$ and~$w_1$. However, by 
    doing this, also the other blocking pair $\{m_1,w_2\}$ including 
    $m_1$ is resolved. In fact, these swaps are part of the unique optimal 
    solution 
    to make $\Mst$ stable, which is to 
    swap~$w_3$ and~$w_2$ and subsequently~$w_3$ and~$w_1$ in~$m_1$'s preference
    relation and to swap~$m_3$ and~$m_2$ in~$w_2$'s preference relation.
\end{example}

In the following, we 
describe how on instance of \exact-\ex-\swap can be solved in time cubic in the
number of agents by reducing it to an instance of \textsc{Minimum Cut}. In 
the \textsc{Minimum Cut} problem, we are given a directed graph 
$G=(V,E)$, a cost function $c : E \rightarrow \mathbb{N} \cup \{\infty\}$,
two distinguished 
vertices~$s$ and $t$ and an integer $k$, and the question is
to 
decide whether there exists a subset of arcs of total weight at most~$k$ such 
that every $(s,t)$-path in~$G$ includes at least one of these arcs.

Before describing the reduction, let us introduce some notation. For two 
agents~${a,a'\in A}$, 
let~$c(a,a')$ denote the 
number of swaps needed such that~$a$ prefers~$\Mst(a)$ to~$a'$, i.e.,
$c(a,a') := \max\big(\rank(a,\Mst(a))-\rank(a,a'),0\big)$, where $\rank(a,a')$ is 
one plus the number of agents which $a$ prefers to $a'$.
Moreover, for each~$a\in A$, let~$q_a$ denote the number of blocking 
pairs 
involving~$a$ and let~$\beta^a_1,\dots, \beta^a_{q_{a}}$ be a list of these 
blocking 
pairs ordered decreasingly by the number of swaps in $a$'s preferences 
needed 
to resolve the blocking pair, 
i.e.,~$c(a,\beta^a_1(a))\geq 
c(a,\beta^a_2(a))\geq 
\dots \ge c(a, \beta^a_{q_a}(a))$.
For a blocking pair~$\beta\in U\times W$ with~$a\in \beta$, we denote by~$\id(a,\beta)$ the position of blocking pair $\beta$ in $a$'s list of blocking 
pairs, that 
is, $\id(a,\beta) =i$ if $\beta=\beta^a_i$.
Using this 
notation, we now
prove the following:

\begin{theorem} \label{th:ex-ex-swap}
    \exact-\ex-\swap is solvable in~$\mathcal{O}(n^4)$ time. 
\end{theorem}
\begin{proof}
	Assume we are given an instance of \exact-\ex-\swap consisting of an SM 
	instance $\mathcal{I}=(U, W, \mathcal{P})$, a matching~$\Mst$, and budget $\ell$.
	Let $A = \{a_1, \dots, a_{2n}\}$.
    We first show that there is always an optimal solution 
	which, for each agent~$a\in A$,  only swaps
    $M^*(a)$ (upwards) 
    in the preference of~$a$ (i.e., $M^*(a)$ becomes more preferred by~$a$). 
    Let $S$ be a list of swap operations of minimum cardinality such that after 
    performing the swap operations~$S$, the given matching~$\Mst$ is stable. 
    Let $a\in A$ be some agent and $i\in \mathbb{N}$ the number of swap 
    operations 
    from~$S$ modifying the preferences of $a$. Then, it is possible to 
    replace these~$i$ swap operations by swapping~$\Mst(a)$ by~$i$ positions to 
    the top in $a$'s 
    preference list. The resulting list of swap operations~$S'$ consists of the 
    same number of swaps and  makes $\Mst$ still stable, as there is no agent 
    which~$a$ prefers to~$\Mst(a)$ after the swap operations in~$S'$ but not 
    after the 
    swap operations in~$S$. 
    As a consequence, it is 
    enough to 
    consider the solutions to the given \exact-\ex-\swap instance that 
    correspond to a tuple
    $(d_{a_1},\dots, d_{a_{2n}})$, where~$d_a$ encodes the number of times
    $\Mst(a)$ is swapped with its left neighbor in~$a$'s  preference relation.
    Note that~$(d_{a_1},\dots, d_{a_{2n}})$ is a valid solution to the
    problem if for each blocking pair~$\{m,w\}\in \bp(\Mst,\mathcal{I})$ it
    holds that 
    $d_m\geq c(m,w)$ or that~$d_w\geq c(w,m)$. Now, we are ready to 
    reduce the given \exact-\ex-\swap instance to an
    instance of the \textsc{Minimum Cut} problem. 
    
    \smallskip 
    \textbf{Reduction to \textsc{Minimum Cut}.}
    We start by constructing a weighted directed graph~$G= (V, E)$ as follows: For each man~$m$, 
    we introduce one vertex for each blocking pair $m$ is part of:~$u^m_1,\dots,
    u^m_{q_m}$.
    Similarly, for each woman~$w\in W$, we introduce one vertex for each 
    blocking pair $w$ is part of:~$u^w_1,\dots,
    u^w_{q_w}$.
    Moreover, we add a source~$s$ and a sink~$t$.
    
    Turning to the 
    arc set, for each~$m\in 
    U$ that is included in at least one blocking pair, we introduce an arc 
    from~$s$ to~$u^m_1$ of cost~$c(m,\beta^m_1(m))$.
    Moreover, for each~$i\in [q_m-1]$, we introduce an arc from~$u^m_i$ 
    to~$u^m_{i+1}$ of cost~$c(m,\beta^m_{i+1}(m))$. For each woman~$w\in W$ 
    that is included in at least one blocking pair, we
    introduce an arc from~$u^w_1$ to~$t$ of 
    cost~$c(w,\beta^w_1(w))$.
    Moreover, for each~$i\in [q_w-1]$, we introduce an arc from~$u^w_{i+1}$ to 
    $u^w_i$ of cost~$c(w,\beta^w_{i+1}(w))$. For each blocking 
    pair~$\beta=\{m,w\}\in 
    \bp(\Mst,\mathcal{I})$, we
    introduce an arc from~$u^m_{\id(m,\beta)}$ to~$u^w_{\id(w,\beta)}$ of 
    infinite 
    cost.
    Finally, we set $k := \ell$.
    We 
    visualize the described reduction in \Cref{fig:exact-ex-swap} where
    the graph corresponding to \Cref{ex:exact-ex-swap} is displayed.  
    
    \begin{figure}[bt]
        \begin{center}
            \begin{tikzpicture}[xscale =2 , yscale = 1.2]
            \node[vertex, label=90:$s$] (sig) at (0, 1) {};
            
            \node[vertex, label=-90:$u_{1}^{m_1}$] (m11) at (0.7, 0) {};
            \node[vertex, label=90:$u_{2}^{m_1}$] (m12) at (0.7, 1) {};
            \node[vertex, label=90:$u_{1}^{m_3}$] (m31) at (0.7, 2) {};
            
            \node[vertex, label=-90:$u_{1}^{w_1}$] (w11) at (1.7, 0) {};
            \node[vertex, label=0:$u_{1}^{w_2}$] (w21) at (1.7, 1) {};
            \node[vertex, label=90:$u_{2}^{w_2}$] (w22) at (1.7, 2) {};
            \node[vertex, label=0:$t$] (tau) at (2.4, 0) {};
            
            \draw[-{Latex[length=2.5mm]}] (sig) edge node[pos=0.5, fill=white, 
            inner 
            sep=2pt]
            {\scriptsize
                $2$}  (m11);
            \draw[-{Latex[length=2.5mm]}] (sig) edge node[pos=0.5, fill=white, 
            inner sep=2pt]
            {\scriptsize
                $2$}  (m31);
            \draw[-{Latex[length=2.5mm]}] (m11) edge node[pos=0.5, fill=white, 
            inner sep=2pt]
            {\scriptsize~$1$}  (m12);
            \draw[-{Latex[length=2.5mm]}] (m11) edge node[pos=0.5, fill=white, 
            inner sep=2pt]
            {\scriptsize~$\infty$}  (w11);
            \draw[-{Latex[length=2.5mm]}] (m12) edge node[pos=0.5, fill=white, 
            inner sep=2pt]
            {\scriptsize~$\infty$}  (w21);
            \draw[-{Latex[length=2.5mm]}] (m31) edge node[pos=0.5, fill=white, 
            inner sep=2pt]
            {\scriptsize~$\infty$}  (w22);
            \draw[-{Latex[length=2.5mm]}] (w22) edge node[pos=0.5, fill=white, 
            inner sep=2pt]
            {\scriptsize~$1$}  (w21);
            \draw[-{Latex[length=2.5mm]}] (w11) edge node[pos=0.5, fill=white, 
            inner sep=2pt]
            {\scriptsize
                $2$}  (tau);
            \draw[-{Latex[length=2.5mm]}] (w21) edge node[pos=0.5, fill=white, 
            inner sep=2pt]
            {\scriptsize
                $2$}  (tau);
            
            \end{tikzpicture}
            
        \end{center}
        \caption{\textsc{Min-Cut} graph constructed to solve 
            \Cref{ex:exact-ex-swap}.
            The number on an edge denotes its weight.
            Note that~$\beta_1^{m_1}=\{m_1,w_1\}$,~$\beta_2^{m_1}=\{m_1,w_2\}$ 
            with
            $c(m_1,\beta_1^{m_1}(m_1))=2$ 
            and~$c(m_1,\beta_2^{m_1}(m_1))=1$;~$\beta_1^{m_3}=\{m_3,w_2\}$ with
            $c(m_3,\beta_1^{m_3}(m_3))=2$;~$\beta_1^{w_1}=\{m_1,w_1\}$ with
$c(w_1,\beta_1^{w_1}(w_1))=2$;~$\beta_1^{w_2}=\{m_1,w_2\}$,~$\beta_2^{m_2}=\{m_3,
w_2\}$
            with
            $c(w_2,\beta_1^{w_2}(w_2))=2$ and~$c(w_2,\beta_2^{w_2}(w_2))=1$. 
            The unique minimum
            $(s,t)$-cut is
            $E':=\{(s,u_1^{m_1}), (u_2^{w_2},u_1^{w_2})\}$.
        }\label{fig:exact-ex-swap}
    \end{figure}
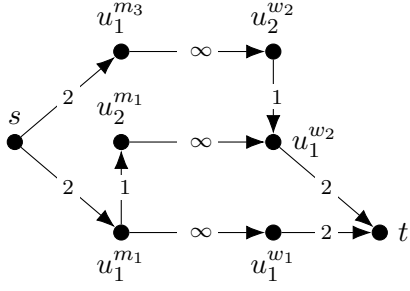
    
    The general idea of the construction is that cutting an arc incident to 
    some vertex~$u_i^a$
    of an agent~$a\in A$ of cost~$c$ is equivalent to swapping 
    up~$\Mst(a)$ in
    $a$'s preference list~$c$ times. Thereby, all blocking pairs with
    costs at most~$c$ for~$a$ are resolved (all paths visiting the corresponding
    vertices are cut) and we encode for each agent~$a$ the entry in the solution 
    tuple 
    by the arc incident to one of its vertices~$u_i^a$ contained in the cut (where no arc contained in the cut corresponds to performing no swaps). For each blocking pair
    the involved woman or the 
    involved man needs to resolve the pair, as otherwise there still exists an 
    $(s,t)$-path. 
    
    Formally, we compute a minimum cut~$E'\subseteq 
    E$  of the constructed graph, which can be done in $\mathcal{O}(|V| \cdot |E|) = \mathcal{O} (n^4)$ time 
    \citep{DBLP:journals/jal/KingRT94,DBLP:conf/stoc/Orlin13}.
    Note that for each agent~$a$ at most one arc to one of 
    $u^a_1,\dots, u^a_{q_a}$ is contained in~$E'$.
    For the sake of contradiction, let us assume that there exist two vertices 
    $u^a_i$ and~$u^a_j$ with~$i<j$ such that 
    both the arc to~$u^a_i$ and~$u^a_j$ have been cut.
    Then, already cutting the arc to~$u^a_i$ destroys all~$(s,t)$-paths 
    visiting~$u^a_j$, contradicting the minimality of the cut.
    
    Using~$E'$, we construct a
    solution tuple as follows. For each agent~$a\in A$, we set~$d_a=0$ if no 
arc to a 
    vertex from
    $u^a_1,\dots, u^a_{q_a}$ has been cut. Otherwise, let~$u^a_i$ be the
    destination of the arc in the cut. We set 
    $d_a=c(a,\beta^a_i(a))$. Note that the cost 
    of the cut corresponds to the cost of the constructed solution.
    
    \textbf{Correctness.} It remains to prove that the solution~$(d_{a_1},\dots,
    d_{a_{2n}})$ computed by our algorithm is indeed a 
    solution to the given \const-\ex-\swap instance and that no solution of 
    smaller cost exists.
    To prove the first part, for the sake of contradiction, let us assume that 
    there exists a pair 
    $\beta=\{m,w\}\in \bp(\Mst,\mathcal{I})$ that is still blocking, i.e., it
    holds that~$d_m< c(m,w)$ and ~$d_w< c(w,m)$.
    However, this implies that the graph~$G'$ arising from $G$ through the deletion of the edges from~$E'$ still contains a path 
from~$s$ to 
    $u^m_{\id(m,\beta)}$ 
    and from~$u^w_{\id(w,\beta)}$ to~$t$: No arcs on the unique path from 
    $s$ to~$u^m_{\id(m,\beta)}$ were cut by~$E'$, as they were all of cost 
    greater 
    than~$c(m,w)$.
    A symmetric argument shows that no edges on the path from~$u^w_{\id(w,\beta)}$ 
    to~$t$ are contained in~$E'$.
    Moreover, there exists an arc of infinite cost from~$u^m_{\id(m,\beta)}$ 
    to~$u^w_{\id(w,\beta)}$
    which implies the existence of an~$(s,t)$-path in $G'$. This leads to 
a 
    contradiction.
    
    To prove the second part, let us assume that there exists a 
    solution~$(d'_{a_1},\dots, d'_{a_{2n}})$ of smaller cost. However, one can 
    construct from this a cut $E''$ of smaller cost than the 
computed minimum cut $E'$,
    a contradiction. For each 
    agent~$a\in A$, we include in the cut~$E''$ the arc to the 
vertex
    from~$u^a_1,\dots, u^a_{q_a}$ of maximum index $i$ such that ${c( a, 
\beta_i^a(a)) \le 
    d'_a}$. 
    Clearly, $E''$ has the same cost as $(d'_{a_1}, \dots, d'_{a_{2n}})$ and is therefore cheaper than the computed minimum cut $E'$, so it remains
    to show that $E''$ is indeed an $(s,t)$-cut, i.e., after
    deleting all arcs from~$E''$, there is no
    $(s,t)$-path.
    For the sake of contradiction, let us assume that there exists an
    $(s,t)$-path after deleting the arcs from $E''$. Then there exist~$m\in U$ 
and~$w\in W$
    with~$i\in [q_m]$ and 
    $j\in [q_w]$ such that this path includes the arc~$(u^m_i, u^w_j)$.
    However, as it needs to hold that~$d'_m\geq c(m,\beta^m_i(m))$ or~$d'_w\geq 
    c(w,\beta^w_j(w))$ (as otherwise $\{m, w\}$ would block $M^*$ after the 
bribery), either
    an arc from the unique path from~$s$ to~$u^m_i$ 
    or an arc from the unique path from~$u^w_j$ to~$t$ is part of~$E''$. 
    This 
    leads to a contradiction.
\end{proof}

Having seen that \exact-\ex-\swap/\accD/\reor are solvable in polynomial time, 
it is natural to ask whether these tasks remain tractable 
if we instead of specifying a full matching only specify a set of edges that 
should be made part of a stable matching. 
Note that we have seen in \Cref{se:Const-Ex-W1} that 
\const-\ex-\swap/\accD/\reor, where the goal is to make just one edge  
part of a stable matching, are NP-hard. Looking now at cases in between these 
two extremes, i.e., an arbitrary number 
$j$ of edges that should be included in some stable matching is given, it is 
straightforward to come up with an FPT-algorithm with respect to the parameter~$n-j$.
\begin{proposition}
    For a given SM instance
    $\mathcal{I}$ and partial matching~$\widetilde{M}\subseteq U\times W$ of 
    size~$j$, one can decide in $(n-j)!n^{\mathcal{O}(1)}$ time whether 
	it is possible to modify~$\mathcal{I}$ 
    using \Iswap/ \IaccD/ \Ireor actions such that~$\widetilde{M}$ is part of 
    some 
    stable 
    matching. 
\end{proposition}
\begin{proof}
    Let $\mathcal{X} \in \{\Iswap, \IaccD, \Ireor\}$.
    The idea is to brute-force over all possibilities~$\widetilde{M}'$ of 
    matching the 
    remaining~$2(n-j)$ agents not included in $\widetilde{M}$ to each other. 
There are~$(n-j)!$ such 
    possibilities (fix an ordering of men and iterate over all possible 
    orderings of women and match two agents at the same position in the 
    orderings to each other). For each possibility, we employ
	the algorithm for \exact-\ex-$\mathcal{X}$ to decide whether the complete matching
    $\widetilde{M}\cup \widetilde{M}'$ can be 
    made stable using at most~$\ell$ manipulative actions of type~$\mathcal{X}$.
\end{proof}

\subsection{\Idelete and \Iadd}
In this section, we turn to the manipulative actions \Idelete and \Iadd for 
which we needed to adapt the definition of \Iexact-\Iex. Recall that in the 
context of these manipulative actions we are given a 
complete matching~$\Mst$ involving all agents in the instance and the goal 
is 
to modify the instance 
such that there exists a stable matching $M'$ with $M'\subseteq \Mst$.
First, we show that for \Iadd  the problem can be solved in linear time. 
Second, we argue that, in contrast to all other 
manipulative actions, the \Iexact-\Iex question is 
computationally hard for the manipulative action \Idelete. This is at first sight
surprising since \Idelete is the only manipulative action 
for 
which the 
\Iconst-\Iex question is solvable in polynomial time. However, it can 
be easily explained by the fact that we need to use the adapted definition 
of the \Iexact-\Iex problem here. 

We start by showing that \exact-\ex-\add is solvable in 
linear time in the input size. On an intuitive level, 
this is due to the fact that, for 
\Iadd, it is already determined by the instance which agents we have to insert
to allow for the existence of a stable matching~$M'\subseteq \Mst$, as in case 
some agent $a$ blocks the matching $M'$ in the instance consisting of agents 
$U_{\orig}\cup W_{\orig}\cup 
X_A$ for some $X_A \subseteq U_{\addag} \cup W_{\addag}$ the 
only possibility to resolve this is to add $M'(a)$ to $X_A$. Following 
this 
idea, we prove that \exact-\ex-\add is linear-time solvable. 

\begin{proposition} \label{th:ex-ex-add}
    \exact-\ex-\add can be solved in
    $\mathcal{O}(n^2)$ time.
\end{proposition}
\begin{proof}
	
	\begin{algorithm}[t]
		
		\KwData{An SM instance~$\mathcal{I}$, a complete matching~$\Mst$, a
			budget~$\ell$, and two sets~$U_{\addag}$ and~$W_{\addag}$.}
		Set~$X_A:=\{w\in W_{\addag}: \exists m\in U\setminus U_{\addag} \text{ 
		with } 
		\Mst(u)=w\}$\; \label{line:xainit}
		\While{there exists a lonely woman~$w\in   W_{\orig} \cup X_A$
			and some man~$m \in  U_{\orig} \cup X_A$ with~$w\succ_m 
			\Mst(m)$}{Add~$\Mst(w)$ to 
			$X_A$\;\label{line:xa}}
		
		\eIf{$\Mst|_{U_{\orig}\cup W_{\orig}\cup X_A}$ is stable  
		and~$|X_A|\leq 
			\ell$}{\Return~$X_A$\;}{\Return False\;}
		\caption{Linear-time algorithm for 
			\exact-\ex-\add.}\label{al:exact-ex-add}
	\end{algorithm} 

	Assume we are given an instance of \exact-\ex-\add consisting of an SM 
	instance~$(U, W, \mathcal{P})$ together with two 
	subsets~$U_{\addag}
	\subseteq
	U$ and~$W_{\addag} \subseteq W$, a matching $\Mst$, and budget $\ell$.
    For a set~$X_A \subseteq U_{\addag} \cup W_{\addag}$, we call an 
agent~$a\in 
U_{\orig}\cup W_{\orig}\cup 
    X_A$  
    \textit{lonely} if~$\Mst(a)\notin
    U_{\orig}\cup W_{\orig}\cup X_A$.
    Note that for a solution~$X_A$ of agents to be added, there cannot exist 
both a
    lonely man and a lonely woman, as they otherwise would form a blocking 
    pair. We assume without loss of generality
    that the instance admits a solution~$X_A$ without a lonely man or no solution
    at all (we can do this by applying the following algorithm twice (once with the role of men and women switched) and then taking the smaller solution).
We show that \Cref{al:exact-ex-add} solves the problem in~$\mathcal{O} 
(n^2)$ time.
    As there exists no lonely man, for each man~$m\in U_{\orig}\cup X_A$ also~$\Mst(m)$ needs to 
    be
    contained in~$W_{\orig} \cup X_A$.
    Thus, every woman added to $X_A$ in Line~\ref{line:xainit} needs to be contained in every solution.
    Moreover, there cannot exist a lonely 
    woman~$w$ and a man~$m\in U_{\orig}\cup X_A$ which prefers~$w$ 
to~$\Mst(m)$, as otherwise $w$ and $m$ would form a blocking pair. This implies 
that all agents added to~$X_A$ in Line~\ref{line:xa} are 
    necessary to create a stable matching which is a subset of~$\Mst$. By 
    adding more agents to~$X_A$, it is never possible to 
    resolve any blocking pairs for~$\Mst|_{U_{\orig}\cup W_{\orig}\cup X_A}$, 
    as we have already ensured 
    that no lonely woman is part of a blocking pair. 
    Thereby, if the instance admits a solution, then~$X_A$ computed by
    \Cref{al:exact-ex-add} is a solution of minimum size. This proves the 
    correctness of the algorithm.
    
    All parts of \Cref{al:exact-ex-add} except for the \textbf{while}-loop can be clearly
    performed in $\mathcal{O} (n^2) $ time.
    To see that the \textbf{while}-loop can be executed in $\mathcal{O} (n^2)$
    time overall, we compute the set of lonely women once before entering the
    \textbf{while}-loop.
    The \textbf{while}-loop can be executed at most $n$ times since there
    are only $n$ women.
    In each execution of the \textbf{while}-loop, we update the set of critical 
lonely
    women in $\mathcal{O} (n)$ time by checking for each woman~$w' \in W_{\orig} \cup (X_A\cap W)$ 
whether the man $\Mst (w)$ added in the last execution of the 
\textbf{while}-loop prefers $w'$ to
    $\Mst (\Mst (w))$ and adding $w'$ to the set of lonely women if this is
    the case.
\end{proof}

In contrast to this, for \Idelete, our modified goal definition allows for more 
flexibility in the 
problem to encode 
computationally hard problems, as it is possible to decide which agents one 
wants to delete from the instance to resolve all initially present blocking 
pairs.
The intuitive reason why this problem is NP-hard is the following:
First, one can ensure that for each deleted agent $a$, agent~$M^*(a)$ needs to 
be deleted as well.
By ensuring 
this, selecting the $n-\ell$ agents that remain after the modifications (and 
thereby also implicitly the $\ell$ agents to be deleted) corresponds to 
finding an independent set of size~$n-\ell$ in the 
``underlying graph'', where each vertex corresponds to a pair in $\Mst$ and two 
vertices 
are 
connected if the two corresponding pairs cannot be part of the same stable 
matching. Note that in
contrast to \Cref{th:ex-ex-reor}, this graph is no longer bipartite. As 
\textsc{Independent Set} is NP-complete~\citep{Karp72}, it follows that
\exact-\ex-\delete (and, in fact, also \exact-\uni-\delete, where the 
problem is to make a given matching the unique stable matching) is also 
NP-complete. 
In the following, we present this hardness result in more detail before 
showing that 
\exact-\ex-\delete 
parameterized 
by the budget~$\ell$ is fixed-parameter tractable. 

\begin{proposition} \label{th:ex-ex-del}
    \exact-\ex/\uni-\delete is 
    NP-complete. This also holds if one is only allowed to delete pairs from 
the given 
    matching~$\Mst$.
\end{proposition}
\begin{proof}
    For \exact-\ex-\delete, membership in NP is obvious, as it is possible to 
determine in polynomial time whether a matching is stable. Further, for 
\exact-\uni-\delete, membership in NP follows from the fact that it is 
possible to determine in polynomial 
    time whether a stable matching is unique, e.g., by running the Gale-Shapley 
    algorithm to compute a stable matching $M$ and afterwards the Gale-Shapley 
    algorithm with roles of women and
    men swapped to compute a stable matching $M'$ and checking whether $M$ 
and~$M'$ are identical; then and only then $M$ is the unique stable matching.
    
    We show the NP-hardness of \exact-\ex/\uni-\delete  by a 
reduction from the NP-complete 
    ~\textsc{Independent Set} problem \citep{Karp72}.
    Given an undirected graph~$G$ and an integer $k$, \textsc{Independent Set} 
    asks whether there
    are $k$~pairwise non-adjacent vertices in~$G$. 
    Given an \textsc{Independent Set} instance~$G=(V,E)$ and integer~$k$, we 
    denote by~$u_v^1,\dots u_v^{d_v}$
    the list of all neighbors of a vertex~$v\in V$. The general 
    idea of the reduction is to introduce for each vertex $v\in V$ a man-woman 
pair who 
    are matched to each other in the given matching $\Mst$ and a penalizing 
    gadget that ensures that if one of the two agents from this pair is 
deleted, then the 
    other one 
    needs to be deleted as well.
    We construct the preferences of the agents in such a way that for every 
    edge~$\{v, v'\} \in E$, the agents corresponding to~$v$ prefer the agent 
    corresponding to~$v'$ of opposite gender to its partner in~$\Mst$.
    Thus, two pairs from $\Mst$ can be part of the same stable matching
    if and only if they are non-adjacent in the given graph.
    Hence, finding a solution to the manipulation problem of size~$\ell$
    corresponds to finding an independent set of size $|V| -\ell$. 
    
    Formally, the construction of the 
    corresponding \exact-\ex-\delete instance works as follows. In the SM
    instance~$\mathcal{I}$, we introduce for each~$v\in V$ a gadget consisting 
    of one vertex man
    $m_v$, one vertex woman~$w_v$, and dummy men and women~$\widetilde{m}^i_v$ 
    and
    $\widetilde{w}^i_v$ for~$i\in [2|V|]$. For all~$v\in V$, the vertex man $m_v$ 
    and the vertex
    woman $w_v$ have the preferences 
    \begin{align*}
    m_v&\colon w_{u_v^{1}}\succ \dots \succ w_{u_v^{d_v}} \succ w_v \succ
    \widetilde{w}_{v}^1 \succ \dots \succ \widetilde{w}_{v}^{2|V|} \pend,  \\
    w_v &\colon m_{u_v^{1}}\succ \dots \succ m_{u_v^{d_v}} \succ
    m_v\succ 
    \widetilde{m}_{v}^1 \succ \dots \succ \widetilde{m}_{v}^{2|V|} \pend
    \end{align*}
    
    \noindent and the dummy men and women~$\widetilde{m}^i_v$ and 
    ~$\widetilde{w}^i_v$ 
    for~$i\in 
    [2|V|]$ have the preferences 
    \begin{align*}
        \widetilde{m}_v^i&\colon w_v \succ \widetilde{w}_v^i \pend, 
        & \widetilde{w}_v^i&\colon m_v \succ \widetilde{m}_v^i \pend.
    \end{align*}
    We set~$\Mst:=\{\{m_v,w_v\}: v\in V\} \cup 
    \{\{\widetilde{m}^i_v,\widetilde{w}^i_v\}: v\in V, i\in [2|V|]\}$ and 
    $\ell:=2(|V|-k)$.
    We now prove the correctness of our construction.
    
    \smallskip
    ($\Rightarrow$) Let~$V'\subseteq V$ be an independent set of size $k$ 
    in~$G$. We claim that deleting the agent set~$A=\{\{m_v,w_v\}: v\in 
    V\setminus 
    V' 
    \}$, which is of size~$2(|V|-k) = \ell$, is a solution to the constructed 
    \exact-\ex-\delete  
    instance, i.e.,~$M'=\Mst\setminus \{\{m_v,w_v\}: v\in V\setminus 
    V' 
    \}$ is a stable matching in the resulting instance. For the sake 
    contradiction, assume that there exists a blocking pair for~$M'$ in 
    $\mathcal{I}\setminus A$. However, no vertex man or vertex woman can be part
    of 
    such a blocking pair, as they are all matched to their top-choice among the 
    remaining agents. Every dummy agent is matched to the best non-vertex agent and thus does not form a blocking pair 
    for~$M'$. Thus, $M'$ is stable (in fact, $M'$ is even the unique stable 
matching in  
    $\mathcal{I}\setminus A$). 
    
    ($\Leftarrow$) Assume that there exists a subset $A'\subseteq 
    U\cup W$ of agents of size at 
most~$\ell =2(n-k)$ such that some matching~$M'\subseteq 
    \Mst$ 
    is stable in~$\mathcal{I}\setminus A'$. First of all note that it is
    never possible to delete for some~$v\in V$ all corresponding dummy men or 
    to delete all dummy women, as the number of 
    both dummy men and dummy women for each 
    vertex exceeds the given budget. From this it follows that if~$A'$ 
    contains~$m_v$ for some~$v\in V$, 
    then it also has to contain~$w_v$. The reason for this is that 
    otherwise~$w_v$ 
    together with some non-deleted dummy agent~$\widetilde{m}_v^j$ forms a blocking pair for
    $M'$. Similarly, if~$A'$ contains~$w_v$ for some~$v\in V$, then it also
    contains~$m_v$. We now claim that~$V':=\{v\in V : m_v\notin A' \}$
    forms an independent set of size at least~$k$. First of all note that~$V'$
    has size at most~$k$, as there exist~$n$ vertices,~$|A'|\leq 2(|V|-k) = 
\ell$, 
    and~$m_v\in A'$ implies~$w_v\in A'$.
    For all~$v,v' \in V'$ we have that $m_v$ and $w_{v'}$ are still present
    in $\mathcal{I}\setminus A'$ by the definition of $V'$.
    Thus, $\{v, v'\} \not \in E$ for all $v,v'\in V'$, as otherwise
    $\{m_v,w_{v'}\}$ forms a blocking pair for~$M'$ in $\mathcal{I}\setminus 
A'$. 
\end{proof}
In contrast to the W[1]-hardness results for the other manipulative actions for 
\Iconst-\Iex, 
\exact-\ex-\delete parameterized by~$\ell$ is 
fixed-parameter tractable. 
The algorithm is based on a simple search tree. 
We pick a blocking pair and branch over which endpoint of the blocking pair gets deleted.
After deleting the selected endpoint, we recompute the set of blocking pairs
and decrease $\ell $ by 1 (see \Cref{al:exact-ex-del}).

\begin{algorithm}[t]
    
    \KwData{An SM instance~$\mathcal{I}$, a complete matching~$\Mst$, and
    a budget~$\ell$.}
    Set~$A'=\emptyset$ and let~$P$ be the set of blocking pairs for~$\Mst$ in 
    $\mathcal{I}$\;
    \While{there exists a pair in~$P$}{
    \If{$\ell\leq0$}{\Return False\;}
    Pick a pair~$\{w, m\} \in P$
    and branch over its endpoints $v = m$ or $v= w$\;
    Set~$A' :=A'\cup \{v\}$ and~$\ell :=\ell-1$\; 
    Set~$P$ to be the set of blocking pairs for
    $\Mst \setminus \{e : e\cap A' \neq \emptyset\}$ in $\mathcal{I}\setminus A'$\;
    }
\Return~$A'$\;
\caption{FPT algorithm wrt.~$\ell$ for 
\exact-\ex-\delete.}\label{al:exact-ex-del}
\end{algorithm}
    
\begin{proposition} \label{th:ex-ex-delFPT}
    \exact-\ex-\delete can be solved in
    $\mathcal{O}(n^2 2^\ell)$ time.
\end{proposition}
\begin{proof}
 We claim that \Cref{al:exact-ex-del} solves \exact-\ex-\delete in the stated 
 running time.
 The correctness follows directly from the fact that for each blocking pair
 one of its endpoints needs to be deleted.
 The running time follows from the fact that the set of blocking pairs can be 
 determined in $O(n^2)$ time and the search tree has depth~$\ell$ and branches 
 into two children at each node.
\end{proof}

\section{\Iexact-\Iuni} \label{ex-uni}
In this section, we turn from the task of making a given matching stable to the task of 
making the given matching the unique stable matching. We show that 
this change makes the considered computational problems significantly more 
demanding in the 
sense that the \Iexact-\Iuni problem is W[2]-hard with respect to~$\ell$ for  
\Ireor and \Iadd and NP-complete for \Iswap. In contrast, the problem for \IaccD is 
solvable in 
polynomial time. Recall that we have already proven in \Cref{th:ex-ex-del} 
that \Iexact-\Iuni is NP-complete for \Idelete.

\subsection{Hardness Results} \label{ex-uni-hard}

Both the W[2]-hardness result for the manipulative action \Ireor and the 
NP-completeness for \Iswap 
follow from the same parameterized reduction from the NP-complete and 
W[2]-complete  
\textsc{Hitting Set} 
problem parameterized by solution size \citep{DBLP:series/mcs/DowneyF99} with 
small modifications. In 
an instance of \textsc{Hitting Set}, we are given a universe~$Z$, a 
family~$\mathcal{F}=\{F_1,\dots, F_p\}$ of subsets of~$Z$, and an 
integer~$k$, and the task is to decide whether there exists a hitting set 
of size at most~$k$, i.e., a set~$X\subseteq Z$ with 
~$|X| \le k$ and~$X\cap F \neq \emptyset~$ for all~$F\in 
\mathcal{F}$. The general 
idea of the construction is as follows:
For each set 
$F\in \mathcal{F}$, we add a \emph{set gadget} consisting of two men and two women, and, 
for 
each element $z\in Z$, we add an \emph{element gadget} consisting of a man-woman pair. We 
connect all set gadgets to the element gadgets corresponding to the elements in 
the set. The preferences are constructed in a way such that in each set gadget 
where none of the element gadgets connected to it is manipulated, the two women can 
switch their partners and the resulting matching is still stable.
In contrast, when an element gadget connected to the set gadget is manipulated, then 
this switch creates a blocking pair and every stable matching contains the same edges in this gadget. 
Thereby, the given matching~$\Mst$  is 
the unique stable matching in the altered instance if and only if the 
manipulated element-gadgets form a hitting set. 
Note that in the 
following reduction, rather unintuitively, we manipulate agents to rank their 
partner in~$\Mst$ worse to make~$\Mst$ the unique stable matching.
\begin{theorem} \label{th:ex-uni-reor}
    \exact-\uni-\reor parameterized by~$\ell$ is W[2]-hard, and this also holds if the given 
    matching $\Mst$ is already stable in the original instance and one is only 
    allowed to modify the preferences of agents of one gender.
\end{theorem}

\begin{proof}
    We give a parameterized reduction from \textsc{Hitting Set}, 
    which is known 
    to be W[2]-complete 
    parameterized by the solution size~$k$ \citep{DBLP:series/mcs/DowneyF99}.
    Given a \textsc{Hitting Set} instance $((Z,\mathcal{F}=\{F_1,\dots , 
F_q\}),k)$, for each element~$z\in Z$, we add a man~$m_z$ and a woman~$w_z$, 
which are 
    the top-choices of each other (the preferences of both $m_z$ and $w_z$ are
    extended arbitrarily to include all agents of opposite sex).
    For each set~$F = \{z_1, \dots, z_q\}\in \mathcal{F}$, we add two 
    men~$m_F^1$ 
    and~$m_F^2$ and two women~$w_F^1$ and~$w_F^2$ with the following 
    preferences:
    \begin{align*}
    m_F^1 &: w_F^1 \succ w_{z_1} \succ w_{z_2} \succ \dots \succ w_{z_q} \succ 
    w_F^2 \pend, &m_F^2 & : w_F^2 \succ w_F^1 \pend,\\
    w_F^1 & : m_F^2 \succ m_F^1 \pend, &w_F^2 & : m_F^1 \succ m_F^2 \pend.	
    \end{align*}
    We set~$\Mst:= \{\{m_z, w_z\} : z\in Z\} \cup \{\{m_F^1, w_F^1\}, \{m_F^2, 
    w_F^2\}: 
    F\in \mathcal{F}\}$ to be the man-optimal matching, and~$\ell:= 
    k$ (see 
    \Cref{fig:ex-uni-reor} for a visualization). 
    
    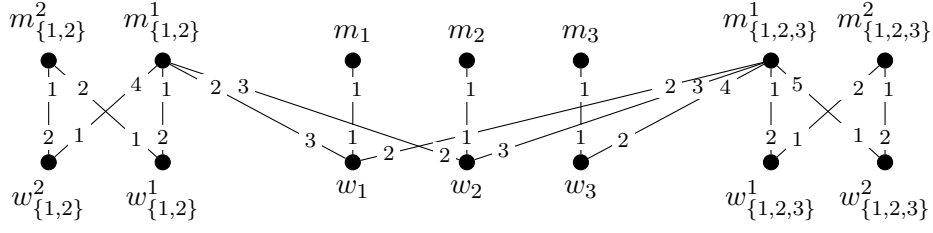
\begin{figure}[t]
    \begin{center}
        \begin{tikzpicture}[xscale =2 , yscale = 0.9]
        \node[vertex, label=-90:$w_1$] (w1) at (0, 0) {};
        \node[vertex, label=90:$m_1$] (m1) at (0, 1.5) {};
        
        \node[vertex, label=-90:$w_2$] (w2) at (0.75, 0) {};
        \node[vertex, label=90:$m_2$] (m2) at (0.75, 1.5) {};
        
        \node[vertex, label=-90:$w_3$] (w3) at (1.5, 0) {};
        \node[vertex, label=90:$m_3$] (m3) at (1.5, 1.5) {};
        \node[vertex, label=-90:$w_{\{1,2,3\}}^1$] (wF1) at (2.75, 0) {};
        \node[vertex, label=90:$m_{\{1,2,3\}}^1$] (mF1) at (2.75, 1.5) {};
        \node[vertex, label=-90:$w_{\{1,2,3\}}^2$] (wF2) at (3.5, 0) {};
        \node[vertex, label=90:$m_{\{1,2,3\}}^2$] (mF2) at (3.5, 1.5) {};
        
        \node[vertex, label=-90:$w_{\{1,2\}}^1$] (wF1b) at (-1.25, 0) {};
        \node[vertex, label=90:$m_{\{1,2\}}^1$] (mF1b) at (-1.25, 1.5) {};
        \node[vertex, label=-90:$w_{\{1,2\}}^2$] (wF2b) at (-2, 0) {};
        \node[vertex, label=90:$m_{\{1,2\}}^2$] (mF2b) at (-2, 1.5) {};
        
        \draw (w1) edge node[pos=0.2, fill=white, inner sep=2pt] {\scriptsize 
            $1$}  node[pos=0.76, fill=white, inner sep=2pt] {\scriptsize~$1$} 
            (m1);
        \draw (w3) edge node[pos=0.2, fill=white, inner sep=2pt] {\scriptsize 
            $1$}  node[pos=0.76, fill=white, inner sep=2pt] {\scriptsize~$1$} 
            (m3);
        
        \draw (wF1) edge node[pos=0.2, fill=white, inner sep=2pt] {\scriptsize 
            $2$}  node[pos=0.76, fill=white, inner sep=2pt] {\scriptsize~$1$} 
            (mF1);
        \draw (wF2) edge node[pos=0.2, fill=white, inner sep=2pt] {\scriptsize 
            $2$}  node[pos=0.76, fill=white, inner sep=2pt] 
            {\scriptsize~$1$}(mF2);
        \draw (wF1) edge node[pos=0.2, fill=white, inner sep=2pt] {\scriptsize 
            $1$}  node[pos=0.76, fill=white, inner sep=2pt] 
            {\scriptsize~$2$}(mF2);
        \draw (mF1) edge node[pos=0.2, fill=white, inner sep=2pt] {\scriptsize 
            $5$}  node[pos=0.76, fill=white, inner sep=2pt] 
            {\scriptsize~$1$}(wF2);
        
        \draw (w3) edge node[pos=0.2, fill=white, inner sep=2pt] {\scriptsize 
            $2$}  node[pos=0.76, fill=white, inner sep=2pt] {\scriptsize~$4$} 
        (mF1);	
        \draw (w2) edge node[pos=0.1, fill=white, inner sep=2pt] {\scriptsize
            $3$}  node[pos=0.76, fill=white, inner sep=2pt] {\scriptsize~$3$} 
        (mF1);
        \draw (w1) edge node[pos=0.07, fill=white, inner sep=2pt] {\scriptsize 
            $2$}  node[pos=0.76, fill=white, inner sep=2pt] {\scriptsize~$2$} 
        (mF1);
        
        \draw (wF1b) edge node[pos=0.2, fill=white, inner sep=2pt] {\scriptsize 
            $2$}  node[pos=0.76, fill=white, inner sep=2pt] {\scriptsize~$1$} 
        (mF1b);
        \draw (wF2b) edge node[pos=0.2, fill=white, inner sep=2pt] {\scriptsize 
            $2$}  node[pos=0.76, fill=white, inner sep=2pt] 
            {\scriptsize~$1$}(mF2b);
        \draw (wF1b) edge node[pos=0.2, fill=white, inner sep=2pt] {\scriptsize 
            $1$}  node[pos=0.76, fill=white, inner sep=2pt] 
            {\scriptsize~$2$}(mF2b);
        \draw (mF1b) edge node[pos=0.2, fill=white, inner sep=2pt] {\scriptsize 
            $4$}  node[pos=0.76, fill=white, inner sep=2pt] {\scriptsize 
            $1$}(wF2b);		
        
        \draw (w2) edge node[pos=0.05, fill=white, inner sep=2pt] {\scriptsize
            $2$}  node[pos=0.76, fill=white, inner sep=2pt] {\scriptsize~$3$} 
        (mF1b);
        \draw (w1) edge node[pos=0.2, fill=white, inner sep=2pt] {\scriptsize 
            $3$}  node[pos=0.76, fill=white, inner sep=2pt] {\scriptsize~$2$} 
        (mF1b);

        \draw (w2) edge node[pos=0.2, fill=white, inner sep=2pt] {\scriptsize
            $1$}  node[pos=0.76, fill=white, inner sep=2pt] {\scriptsize~$1$}
            (m2);
        \end{tikzpicture}
    \end{center}
    \caption{Example of the hardness reduction from
        \Cref{th:ex-uni-reor} for the \textsc{Hitting Set} 
        instance~$Z=\{1,2,3\}$,
        $\mathcal{F}=\{\{1,2,3\},\{1,2\} \}$.}
    \label{fig:ex-uni-reor}
\end{figure}
        
    We now show that the given \textsc{Hitting Set} instance admits a solution 
    of 
    size~$k$ if and only if it is possible to make~$M^*$ the unique stable
    matching in the constructed SM instance by reordering the preferences of at 
    most $\ell$ agents. 
    
    ($\Rightarrow$)
    Let~$X \subseteq Z$ be a hitting set.
    For $z\in X$, we modify the preferences of~$w_{z}$ to be the following:
    ~$w_{z} : m_{F_1}^1 \succ m_{F_2}^1\succ \dots \succ m_{F_q}^1 \succ 
    m_{z} 
    \pend$.
    Matching~$\Mst$ is still a stable matching, as each man is matched to his 
    top-choice.
    To show that~$\Mst$ is the unique stable matching, we utilize that if a 
    second stable matching~$M'$ exists, then the union~$\Mst \cup M'$ needs 
    to 
    contain at least one cycle which consists of alternating edges from the two 
    matchings. Moreover, as $\Mst$ is man-optimal (as every man is matched to 
his top-choice), each woman contained in the 
    cycle needs to prefer the man matched to her in $M'$ to the man matched to
    her in $\Mst$. We now argue that such a cycle cannot exist and thereby that 
    $\Mst$ 
    is the unique stable matching. 
    First of all, note that this cycle cannot contain an agent~$m_z$ (and thus neither~$w_z$) for some~$z\in Z$, as no woman 
    prefers such 
    a man to her partner in~$\Mst$.
    Thus, $M'$ contains $\{m_z, w_z\}$ for every~$z\in Z$.
    Next, we show that the cycle does not contain agent~$m_F^1$ for all $F\in 
    \mathcal{F}$.  
    Because $X$ contains at least one~$z_F\in F$, 
    the preferences of at least one woman $w_{z_F}$ with $z_F\in F$ have been 
    reordered 
    such that $w_{z_F}$ now prefers $m_F^1$ to $m_{z_F}$. As no woman $w_{z}$ for $z\in 
Z$ can be part of a cycle in $\Mst \cup M'$, a cycle in~$\Mst \cup M'$ 
containing 
    $m_F^1$ would imply that~$M'$ matches~$m_F^1$ to some 
    woman from a set-gadget that is not his top-choice. However, 
then 
    $M'$ is blocked by the pair~$\{m_F^1,  w_z\}$, a contradiction.
    Hence, the cycle cannot contain an agent~$m_F^1$.
    Therefore, the cycle contains an agent~$m_F^2$.
    Since~$m_F^1$ and 
    therefore also~$\Mst(m_F^1 ) = w_F^1$ are not contained in the 
cycle, this implies that $M'$ matches $m_F^2$ to a woman to which he 
prefers $w_F^1$. Thus, $M'$ is blocked by the 
    pair~$\{m_F^2, w_F^1\}$ in this case, a contradiction. As we have exhausted all cases, it 
follows that $\Mst$ is the unique stable matching after the bribery.
    
    ($\Leftarrow$)
    Let~$S$ be the set of at most $\ell$ agents such that modifying their preferences
    can make~$\Mst$ the unique stable matching. From 
    this, 
    we construct a solution $X$ to the given \textsc{Hitting Set} problem as 
    follows.
    For each agent~$m_z$ or~$w_z$ contained in~$S$, we add the element~$z$ 
    to~$X$, 
    and for each agent~$m_F^i$ or~$w_F^i$ contained in~$S$, we add an arbitrary 
    element~${z\in F}$ to~$X$.
    Clearly,~$|X| \le |S| \le \ell = k$, so it remains to show that~$X$ is a hitting 
    set.
    
    Assume that~$F\in \mathcal{F}$ does not intersect~$X$.
    Then the preferences of~$m_F^i$,~$w_F^i$, and all~$w_z$ for~$z\in F$ 
    are 
    unchanged after the bribery.
    We claim that~$M' := \bigr( \Mst \setminus \{\{m_F^1, w_F^1\}, \{m_F^2, 
    w_F^2\}\}\bigr) \cup \{\{m_F^1, w_F^2\}, \{m_F^2, w_F^1\}\}$ is then a 
stable 
    matching, 
    contradicting the fact that~$\Mst$ is the unique stable matching.
    
    As $\Mst$ is stable and $\Mst$ and $M'$ only differ in $m_F^1$, $w_F^2$, 
    $m_F^2$, and 
    $w_F^1$, any blocking pair for~$M'$ must contain an agent~$m_F^i$ 
    or~$w_F^i$ 
    for some 
    ~$i \in \{1, 2\}$.
    Agents~$w_F^1$  and $m_F^2$ are matched to their top-choice and thus are 
not part of 
    a 
    blocking pair. The only agent which~$m_F^2$ prefers to~$w_F^1$ is 
    $w_F^2$. However, $w_F^2$ is matched to her top-choice in~$M'$. Similarly, 
as after the bribery 
    all agents which~$m_F^1$ prefers to~$w_F^2$ are matched to 
    their 
    top-choices, and thus do not participate in a blocking pair. Thus, no 
    blocking pair for $M'$ exists.
\end{proof}

We now adapt the parameterized reduction presented in the proof of 
\Cref{th:ex-uni-reor} to show that the \Iexact-\Iuni problem parameterized 
by~$\ell$ is also W[2]-hard for the manipulative action \Iadd.
\begin{proposition} \label{th:ex-uni-add}
    \exact-\uni-\add parameterized by~$\ell$ is W[2]-hard, and this also holds if the given 
    matching $\Mst$ is already stable in the instance restricted to $U_{\orig} \cup W_{\orig}$ and one is only 
    allowed to add agents of one gender.
\end{proposition}
\begin{proof}
    We show \Cref{th:ex-uni-add} by adapting the reduction from 
    the proof of \Cref{th:ex-uni-reor}. We  adapt the reduction 
slightly by changing the preferences of
    $w_z$ for all~$z\in Z$ to~$w_{z} \colon m_{F_1}^1 \succ m_{F_2}^1 \succ  
    \dots \succ 
    m_{F_p}^1 \succ m_{z} \pend$. Moreover, we set 
    $\Mst:= \{\{m_z, w_z\} : z\in Z\} \cup \{\{m_F^1, w_F^1\}, \{m_F^2, w_F^2\} : F\in 
    \mathcal{F}\}$, we set~$\ell:= k$, and the agents that can be added 
    to~${U_{\addag}:=\emptyset}$ and ~$W_{\addag}:=\{w_z: z\in Z\}$. 
    
    ($\Rightarrow$)
    Let~$X\subseteq Z$ be a hitting set. We set~$X_A=\{w_z: z\in X\}$ to be 
the set of 
    added agents. We claim that 
    $M'=\{\{m_z, w_z\} : z\in X\}\cup\{\{m_F^1, w_F^1\}, \{m_F^2, w_F^2\}: F\in 
    \mathcal{F}\} \subseteq \Mst$ is the unique stable matching after adding 
    the 
    agents from~$X_A$. 
    First of all note that~$M'$ is stable, as all matched men are matched to their top 
    choice and no woman prefers one of the unassigned man to their assigned 
    partner. The argument why~$M'$ is the unique stable matching is analogous 
    to the proof of
    \Cref{th:ex-uni-reor}.
    
    ($\Leftarrow$)
    Let~$X_A$ be the set of agents added in a solution to the constructed 
    \exact-\uni-\add instance and let 
    $M':=\Mst|_{U_{\orig}\cup W_{\orig}\cup X_A}$ be the unique stable matching. 
    We 
    claim 
    that~$X=\{z\in Z: w_z\in X_A\}$ is a hitting set.
    
    Assume that~$F\in \mathcal{F}$ does not intersect~$X$. Then~$M'' := \bigr( 
    M' \setminus \{\{m_F^1, w_F^1\}, \{m_F^2, 
    w_F^2\}\}\bigr) \cup \{\{m_F^1, w_F^2\}, \{m_F^2, w_F^2\}\}$ is also a stable 
    matching by an argument analogous to the proof of \Cref{th:ex-uni-reor}, 
contradicting 
    the fact that~$\Mst$ is the unique stable matching. 
\end{proof}

Finally, we adapt the reduction from \Cref{th:ex-uni-reor} to prove
NP-hardness for the manipulative action \Iswap. Here, we utilize the 
fact that \Ireor operations can be modeled by (up to $n^2+n$) \Iswap operations
(see \Cref{se:relManip} for a high-level discussion of this fact).
To do so, we
adapt the reduction such that it is only possible to modify the preferences of 
women $w_z$ and add an ``activation cost'' to modifying the preferences of 
$w_z$ such that only the preferences of a fixed number of women can be modified 
but  for these we can modify them arbitrarily.
\begin{proposition} \label{th:ex-uni-swap}
    \exact-\uni-\swap is NP-complete, and this also holds if the given 
    matching~$\Mst$ is already stable in the original instance and we are only 
    allowed to modify the preferences of agents of one gender.
\end{proposition}
\begin{proof}
Membership in NP is obvious, as it is possible to determine in polynomial 
    time whether a stable matching is unique (as discussed in 
\Cref{th:ex-ex-del}).
    
    We adapt the reduction from \textsc{Hitting Set} to \exact-\uni-\reor 
    from~\Cref{th:ex-uni-reor} as 
    follows.
    Let $(Z, \mathcal{F}, k)$ be an instance of \textsc{Hitting Set}, and let 
    $(\mathcal{I}, \Mst, \ell)$ be the instance of \exact-\uni-\reor 
constructed in
    the reduction described in the proof of \Cref{th:ex-uni-reor}.
    We assume that $\ell \le 2n$, as otherwise $\mathcal{I}$ is a trivial 
    YES-instance.
    Furthermore, we assume $n \ge 3$, since otherwise the \textsc{Hitting Set}
    instance can be solved by brute force.
    In the 
    following we modify~$(\mathcal{I}=(U, W, \mathcal{P}), \Mst, \ell)$ as 
follows. 
    We add $n^5$ men $m_1^d$, \dots, $m_{n^5}^d$ and $n^5 $ women $w_1^d$, 
    $\dots$, $w_{n^5}^d$, where, for all $i\in [n^5]$, the preferences of 
$m_i^d$ and $w_i^d$ are as 
    follows (indices are taken modulo $n^5$):
    \begin{align*}
    m_i^d &: w_i^d \succ w^d_{i+1} \succ w^d_{i+2} \succ \dots \succ w^d_{n^5 + 
        i-1} \pend, \\
    w^d_i &: m^d_i \succ m^d_{i+1} \succ m^d_{i+2} \succ \dots \succ m^d_{n^5 + 
        i-1} \pend.
    \end{align*}
    Now, for each $z\in Z$, we add in $w_z$'s preference list $n^2$ dummy men 
after 
    $m_z$, and the remaining dummy men at the end of $w_z$'s preference list.
    For all other women~$w\in W\setminus\{w_z: z\in Z\}$, we insert $n^4$ 
    dummy men between each two neighboring agents in $w$'s preferences.
    For each men $m\in U$, we insert $n^4$ dummy women between any two neighboring 
    agents in $m$'s preferences.
    We set $M':= \Mst \cup \{\{m_i^d, w_i^d\} : i\in [n^5]\}$ and the overall 
    budget 
    to~$\ell':=k(n^2+n)$.
    The reduction clearly runs in polynomial time, so it remains to show its 
    correctness.
    
    $(\Rightarrow)$
    Let $X$ be a hitting set. For each $z\in X$, we 
    change 
    the preferences of~$w_{z}$ by swapping~$m_{z}$ down  $n^2+n$ times such 
    that after the modification~$w_z$ prefers, for all $F\in \mathcal{F}$, the 
    agent 
    $m_F^1$ to $m_{z}$.
    Thus, the overall number of performed swaps is at most~$\ell'$.
    By \Cref{lem:dummy-agents}, any stable matching in the modified instance 
    contains the edges~$\{m_i^d, w_i^d\}$ for each $i\in [n^5]$.
    It follows by the same arguments as in the proof of 
\Cref{th:ex-uni-reor} 
    that $M'$ is the unique stable matching after the bribery.
    
    $(\Leftarrow)$
    Since $\ell' < n^4$, we can swap pairs containing  two non-dummy agents 
    only 
    for agents~$w_z$ for some $z\in Z$.
    Note that swapping the agent $m_z$ with any non-dummy agent in the 
    preferences of $w_z$ requires at least $n^2$ swaps, and thus, this happens 
    for at most~$k $~such agents.
    The corresponding elements of $Z$ now form a hitting set by
    arguments analogous
    to the proof of \Cref{th:ex-uni-reor}.
\end{proof}

\begin{remark}
	In \Iexact-\Iuni, we are given a matching $\Mst$ and want to make this 
	matching stable.
	However, another possible objective might be to just ensure that there 
	exists a 
	unique stable 
	matching in the manipulated instance, irrespective of which matching 
	is the stable one.
	We remark that the reductions presented in this section also show 
	W[2]-hardness
	parameterized by the budget respectively NP-hardness for this objective.
\end{remark}

\subsection{Algorithms}
Contrasting  the hardness results for all other manipulative actions,
\exact-\uni-\accD  turns out to be solvable in polynomial time. 
On an intuitive level, one 
reason for this is that we can only delete agents from the preferences of other agents, but we cannot swap the order of agents in the preferences.
In particular, we cannot change whether an agent~$a$ is before or after~$\Mst (a')$  in the preferences of $a'$.
Recall that the hardness
reductions for \Ireor and \Iswap crucially rely on this feature. We start this 
subsection by giving definitions and facts about rotations 
\citep{DBLP:books/daglib/0066875} that we will use afterwards to construct a 
polynomial-time algorithm for \exact-\uni-\accD and an XP-algorithm for 
\exact-\uni-\reor parameterized by~$\ell$.
For more details on rotations, we refer to the monograph of 
\citet{DBLP:books/daglib/0066875}.
 
For a stable matching~$M$ and a man $m\in U$, let~$s_M(m)$ 
denote the first woman~$w$ succeeding~$M(m)$ in~$m$'s preference list who 
prefers~$m$ to~$M(w)$.
If no such woman exists, then we set~$s_M (m) := \emptyset$.
A \emph{rotation exposed in a stable matching~$M$} is a
sequence~$\rho=(m_{1},w_{1}),\dots,(m_{{r}},w_{{r}})$ such that for 
each~$k\in 
[r]$ it holds that~$\{m_{k}, w_{k}\}\in M$ and 
$w_{{k+1}}=s_M(m_{{k}})$, where indices are taken modulo~$r$. We call 
such a rotation a \textit{man-rotation} and~$s_M(m)$ the \textit{rotation 
    successor} of~$m$.
    There is a close relation between rotations and stable matchings (see 
    e.g.~\citep{DBLP:books/daglib/0066875}).
    It is easy to see that given a rotation $(m_{1}, w_{1}), \dots, (m_{r}, 
w_{r})$ exposed in a stable matching~$M$, matching~$M' := \bigl(M \setminus 
\{\{m_k, w_k\} : k\in [r] \}\bigr) \cup \{ \{m_k, w_{k+1}\} : k 
\in  [r] \}$ is again a stable matching.
    We will mainly use the ``reverse direction'' of this statement, namely that 
    the absence of rotations can be used to prove the uniqueness of a stable 
    matching.
    
    \begin{example}
    Consider the SM instance depicted in \cref{fig:ex-uni-reor} in 
    \Cref{ex-uni-hard} and let $M$ be 
    the matching in which every man is matched to his top-choice. Then, the 
    rotation 
    successor of~$m^1_{\{1,2\}}$ in~$M$ is $w^2_{\{1,2\}}$, i.e., 
    $s_M(m^1_{\{1,2\}})=w^2_{\{1,2\}}$. Moreover, the rotation successor 
    of~$m^2_{\{1,2\}}$ is $w^1_{\{1,2\}}$, i.e., 
    $s_M(m^2_{\{1,2\}})=w^1_{\{1,2\}}$. Thus, 
    $(m^1_{\{1,2\}},w^1_{\{1,2\}}),(m^2_{\{1,2\}},w^2_{\{1,2\}})$ is a rotation 
    exposed in 
    the stable matching $M$, which proves that $M$ is not the unique stable 
    matching. 
\end{example}

We define the rotation successor~$s_W(w)$ of~$w\in W$ 
analogously and call a rotation where the roles of men and women are switched 
\textit{woman-rotation}. As a matching is unique if and only if it
exposes neither a man-rotation nor a woman-rotation 
\citep{DBLP:books/daglib/0066875}, we can
reformulate the goal of \exact-\uni-\accD: 
Modify the given SM instance by deleting the acceptability of at most~$\ell$
pairs such that neither a man-rotation nor a woman-rotation is exposed 
in~$\Mst$.

We start by making two straightforward observations:
\begin{observation}\label{obs:rotation1}
    To determine whether a stable matching~$M$ is the unique stable matching, 
    it is enough to know the rotation successors of all agents. 
\end{observation}
This observation gives rise to a simple algorithm to check whether a stable 
matching~$M$ exposes a man-rotation.
We create a sink~$t$ and for each man~$m\in U$ a vertex~$v_m$.
We insert an arc from~$v_m$ to~$v_{m'}$ if~$M(m')$ is~$m$'s 
rotation successor in~$M$, i.e.,~$M(m')=s_M(m)$, and an arc from~$v_m$ to~$t$ 
if 
$m$ 
does not have a rotation successor. Then, checking whether~$M$ exposes a 
man-rotation reduces to checking whether there exists a cycle in the 
constructed directed graph. 

In the following, for an agent~$a\in A$, we refer to its preferences induced by 
the 
set of all agents it prefers to~$M(a)$ as the \textit{first part} of its 
preferences and to  its 
preferences induced by the set of all agents to which 
it prefers~$M(a)$ as the \textit{second 
    part} of its preferences. Using this notation, we can make the following 
observation, which holds also if the roles of women and men are switched: 
\begin{observation} \label{obs:rotation}
    To determine whether a stable matching~$M$ exposes a man-rotation, it is 
    enough to know the first part of the women's preferences and the second 
    part of the men's preferences.  
\end{observation}

A simple 
approach to solve the \exact-\uni-\accD problem would be to start by computing 
the set of 
rotations in the given matching and delete the acceptability of one pair in 
each rotation. However, thereby, we would change the rotation successors 
of some agents, which could lead to new rotations. That is why we need to be 
careful when 
choosing the pair within each rotation which one wants to delete. 

To circumvent this issue, we first observe that, by \Cref{obs:rotation}, it is never 
beneficial to delete 
the acceptability of some pair~$\{m,w\}$ with $m\in U$ and $w\in W$ if $m$ and~$w$ 
both appear in the 
same part of each others preferences, as this implies that none of them can be 
the rotation successor of the other. Moreover, it is possible to separately 
solve the 
problem of ensuring that the given matching~$\Mst$ does not expose a man-rotation
and the problem of ensuring that the given matching~$\Mst$ does not expose a woman-rotation.
To solve the former problem, we only care about the rotation successors of all men. Thereby, 
we only delete the acceptability of pairs~$\{m,w\}$ where~$w$ appears in the 
second part of~$m$'s preferences and~$m$ appears in the first part of~$w$'s 
preferences. For woman-rotations, the situation is 
symmetric. We solve both 
problems by reducing them to the \textsc{Minimum Weight Spanning 
Anti-Arborescence} problem. In an 
instance of the \textsc{Minimum Weight Spanning Anti-Arborescence} problem, we 
are given a directed graph $G=(V,E)$ with arc costs and a budget $k\in 
\mathbb{N}$. 
The question is whether there exists a spanning anti-arborescence, i.e., an 
acyclic 
subgraph of~$G$ such that all vertices of~$G$ but one have out-degree 
exactly one, of cost at most~$k$.  \textsc{Minimum Weight Spanning 
Anti-Arborescence} can be 
solved in~$\mathcal{O}(|E|+|V|\log |V|)$
time
\citep{edmonds1967optimum,DBLP:journals/combinatorica/GabowGST86}.

Given an SM instance $\mathcal{I}=(U, W, \mathcal{P})$ and a matching 
$\Mst$, the basic idea of the algorithm is the following (we present the 
algorithm
for excluding man-rotations; excluding woman-rotations can be done 
symmetrically):
For a set~$F$ of deleted 
acceptabilities, let~$s_{\Mst}^F (m)$ denote the rotation successor 
of~$m$ after the deletion of~$F$.
To ensure that there is no man-rotation, we need to find a set~$F$ of deleted acceptabilities such that the graph where 
the agents form the vertex set and the arc set is~$\{ (w, m) : m\in U \land 
w\in W\land \{m, w\} \in \Mst \} \cup 
\{(m, s_{\Mst}^F (m)): m\in U\text{ with } s_{\Mst}^F (m) \neq \emptyset\}$ 
is 
acyclic.
Note that we can change~$s_{\Mst} (m)$ only by deleting the pair~$\{m, 
s_{\Mst}(m)\}$. In this case, the new rotation successor becomes the first 
woman~$w'$ 
succeeding~$s_{\Mst}(m)$ in~$m$'s preference list which prefers~$m$ 
to~$\Mst(w')$.
Thus, the costs of making a woman~$w'$ the rotation successor of~$m$ is the 
number of women~$w$ such that $m$ prefers $\Mst(m)$ to $w$ and $w$ to $w'$ and 
$w$ prefers $m$ to $\Mst(w)$.
We now argue that the problem of making~$\{(w, m) : m\in U \land w\in W \land 
\{m, w\} \in \Mst\} \cup \{(m, s_{\Mst}^F (m)): m\in U\text{ with } 
s_{\Mst}^F (m) 
\neq \emptyset\}$ acyclic can be expressed as an instance of \textsc{Minimum 
Weight Spanning 
    Anti-Arborescence}.

\begin{theorem}\label{thm:exact-uni-accD}
    \exact-\uni-\accD can be solved in~$\mathcal{O}(n^2)$ time.
\end{theorem}
\begin{proof}
    Clearly, any solution needs to delete all blocking pairs.
    Thus, we assume without loss of generality that~$\Mst$ is a stable matching.
    
    Given an instance $(\mathcal{I}=(U, W, \mathcal{P}), \Mst, \ell)$ of 
\exact-\uni-\accD, we reduce the problem to two instances of the
    \textsc{Minimum Weight Spanning Anti-Arborescence} problem.
    The first instance of this problem that is responsible for deleting all 
man-rotations is constructed as follows.
    The graph contains a vertex~$v_m$ for each edge~$\{m,w\}\in 
\Mst$ as well as a
sink~$t$.
    We add an arc~$(v_m, v_{m'})$ if~$\Mst(m')$ prefers~$m$ to~$m'$ and~$m$ 
    prefers~$\Mst (m)$ to~$w'$.
    The weight of this arc is the number of women~$w^*$ such that~$m$ 
    prefers~$w^*$ to~$\Mst (m')$ and~$\Mst (m)$ to~$w^*$, and~$w^*$ prefers~$m$ 
    to~$\Mst(w^*)$ (i.e., the number of acceptabilities which need to be deleted
    to make $\Mst (m')$ the rotation successor of $m$).
    Furthermore, there is an arc $(v_m, t)$ for all $\{m,w\} \in \Mst$.
    The weight of this arc is the number of women~$w^*$ such that~$m$ accepts $w^*$ and 
    prefers~$\Mst (m)$ to~$w^*$, and~$w^*$ prefers~$m$ 
    to~$\Mst(w^*)$ (i.e., the number of acceptabilities which need to be deleted
    to make $\emptyset$ the rotation successor of $m$).
    We call this graph~$H_U$. Similarly, we construct a graph~$H_W$ (where the 
    roles of men and women are
    exchanged). 
    
    We claim that~$\Mst$ can be made the unique stable matching after the 
    deletion 
    of~$\ell$ pairs if and only if the minimum weight anti-arborescences in~$H_U$
    and~$H_W$ together have weight at most~$\ell$. 
        
    ~$(\Rightarrow)$
    Let~$F\subseteq \{\{m, w\} : m \in U , w\in W\}$ be a set of at most~$\ell$ 
pairs whose deletion make~$\Mst$ 
    the 
    unique stable 
    matching.
    Let~$F_W:= \{\{m, w\} \in F: w\succ_m \Mst(m)\}$ and~$F_U := \{\{m, w\} \in 
F: 
    m\succ_w \Mst(w)\}$.
    For any man~$m$, let~$s_m := v_{\Mst(w')}$, where~$w'$ is the woman
    best-ranked by~$m$ succeeding~$\Mst(m)$ such that~$w$ prefers~$m$ 
    to~$\Mst(w)$ 
    after 
    the manipulation, i.e,.~$w'$ is the rotation successor of~$m$ after the
    manipulation.
    If no such woman exists, then we set~$s_m := t$.
    We construct an anti-arborescence in~$\mathcal{A}_U$ of cost at 
most~$|F_U|$ by 
    adding for
    each pair~$\{m, w\}\in \Mst$ the arc~$(v_m, s_m)$ to the anti-arborescence.
    We claim that~$\mathcal{A}_U$ is an anti-arborescence.
    Every vertex but~$t$ has exactly one outgoing arc, so it is enough to show 
    that there does not exist a cycle.
    As we have inserted for each man an arc from the node including him to the 
    node including his rotation successor, there cannot exist any cycle in the 
    anti-arborescence, as such a cycle would induce an exposed man-rotation in the 
    modified 
    SM instance which would
    contradict the uniqueness of $\Mst$ in the modified SM instance.
    
    In the same way one can construct an anti-arborescence of cost~$|F_W|$ 
    in~$H_W$. The constructed anti-arborescences together 
    have
    weight at most~$|F_W| + |F_U| \le |F| \le \ell$, as any arc in~$F_W\cap F_U$ would 
    be 
    a blocking pair for $\Mst$.
    
    ~$(\Leftarrow)$
    Let~$\mathcal{A}_U~$ be an anti-arborescence in~$H_U$, and~$\mathcal{A}_W$ 
    be 
    an anti-arborescence in~$H_W$.
    For every arc~$(v_m, t)\in \mathcal{A}_U$, we delete the acceptability of 
all pairs~$\{m, w'\}$ with $m$ preferring~$\Mst (m)$ to~$w' $ and $w'$ 
preferring~$m$ to $\Mst (w')$.
    For every
    arc~$( v_m, v_{\widetilde{m}}) \in \mathcal{A}_U$, we 
    delete 
    the acceptability of all
    pairs~$\{m, w'\}$ with~$m$ 
    preferring~$w'$ to~$\Mst (\widetilde{m})$, and~$\Mst (m)$ to~$w'$, and~$w'$ 
    preferring~$m$ 
    to~$\Mst(w')$.
    After these deletions,~$\widetilde{w}$ is the rotation successor of~$m$. 
    Let~$F_U$ denote the set of pairs deleted.
    We proceed with $\mathcal{A}_W$ analogously, and denote as~$F_W$ the set 
    of deleted pairs.
    By construction, $\mathcal{A}_U$~has cost~$|F_U|$, and $\mathcal{A}_W$ has 
    cost~$|F_W|$.
    
    Assume that~$\Mst$ is not the unique stable matching after deleting the 
    pairs from $F_U\cup F_W$.
    Then, without loss of generality, a man-rotation is exposed 
    in~$\Mst$:~$(m_{i_1}, 
    w_{j_1}), 
    \dots, (m_{i_{r}},
    w_{j_{r}})$.
    As we already observed, the anti-aborescence~$\mathcal{A}_U$ contains all 
    arcs~$(v_m, v_{\widetilde{m}})$ where~$\widetilde{w}$ 
    is~$m$'s 
    rotation-successor (after the deletion of~$F_U$).
    Thus,~$\mathcal{A}_U$ contains the arcs~$(v_{m_{i_k}}, 
    v_{m_{i_{k+1}}})$ for
    all~$k\in [r]$ (all indices are taken modulo~$r$). This implies 
    that~$\mathcal{A}_U$
    contains a cycle, a contradiction to~$\mathcal{A}_U$ being an 
    anti-arborescence.
\end{proof}

We conclude this section by constructing an XP algorithm for \exact-\uni-\reor 
parameterized by~$\ell$ which runs in $\mathcal{O}(2^\ell n^{2\ell+2})$ time. 
This algorithm is described in \Cref{alg:xp-unique-reor} and requires as input 
an instance of \exact-\uni-\reor consisting of an SM instance 
$\mathcal{I}=(U, W, \mathcal{P})$, a complete matching $\Mst$, and a budget 
$\ell$. 
\begin{algorithm}
    \Input{An SM instance~$\mathcal{I}=(U, W, \mathcal{P})$, a complete 
matching~$\Mst$, and
    a budget~$\ell$.}
   Guess sets $X_U \subseteq U$ and $X_W \subseteq W$ with $|X_U \cup X_W| \le \ell$\;\label{line:guessx}
  For each $m \in X_U$, guess $s^{\reor}_{\Mst} (m) \in W \cup \{ \emptyset\}$\;\label{line:guesssm}
  For each $w\in X_W$, guess $s^{\reor}_{\Mst} (w) \in M \cup 
\{\emptyset\}$\;\label{line:guesssw}
  \If{there exists a blocking pair~$\{m, w\}$ for~$\Mst$ with $m\notin X_U$ and 
  $w\notin X_W$}{Reject this guess\;\label{line:checkbp}}
  \If{there exists an agent $a\in X_U\cup X_W$ with $s^{\reor}_{\Mst} (a) 
\notin 
X_U \cup X_W$ and $s^{\reor}_{\Mst} (a)$ preferring $M(s^{\reor}_{\Mst} (a))$ 
to $a$}{Reject this guess\;\label{line:checksuccs}}
  \If{there exists $m\in X_U$ and a woman $w\in W \setminus X_W$ such that $s^{\reor}_{\Mst} (m) = \emptyset$ and $w$ prefers $m$ to $\Mst (w)$ or there exists $w\in X_W$ and $m\in U\setminus X_U$ such that $s^{\reor}_{\Mst} (w) = \emptyset$ and $m$ prefers $w$ to $\Mst (m)$}{Reject the guess\;\label{line:checkrotation}}
  \If{there exists $m\in X_U$ and $w\in X_W$ such that $m = s_{\Mst}^{\reor} (w) $ and $w = s^{\reor}_{\Mst} (m)$}{Reject the guess\;\label{line:checkmutualsuccessor}}
  Let~$H$ be an empty directed graph\;\label{line:begin}
   Add a vertex~$t$ to $H$\; \label{line:addt}
    \ForEach{$\{m, w\} \in \Mst$}
    {
    Add a vertex~$v_m$ to $H$\;\label{line:addvm}
    }
    \ForEach{$m \in X_U$}{
     Add arc $(v_m, v)$ to $H$, where $v:=t$ if $s^{\reor}_{\Mst} 
     (m)=\emptyset$ 
     and $v:=v_{m'}$ with $m': = \Mst ({s^{\reor}_{\Mst} (m)})$ 
     otherwise\;\label{line:arcxu}
     }
    \ForEach{$m\in U\setminus X_U$\label{line:beginm}}{
      Let $\widetilde{s}_{\Mst} (m) \in W \setminus X_W$ be the woman from $W\setminus X_W$ which $m$ likes most such that
        $m$ prefers $\Mst (m)$ to $\widetilde{s}_{\Mst} (m)$, and
        $\widetilde{s}_{\Mst} (m)$ prefers $m$ to $\Mst 
        \big(\widetilde{s}_{\Mst} 
        (m)\big)$\;\label{line:computewm}

      \If{no such $\widetilde{s}_{\Mst} (m)$ exists}{
        Add $(v_m, t)$ to $H$\;\label{line:sinkedge}
        \ForEach{$w\in X_W$ such that $m$ prefers $\Mst (m)$ to $w$ and $s^{\reor}_{\Mst} (w) \neq m$}{
        Add arc $(v_m, v_{m'})$ to $H$, where $m' := \Mst 
        (w)$\;}\label{line:arc}
      }

      \Else{
        Add arc $(v_m, v_{m'})$ to $H$, where $m' := \Mst (\widetilde{s}_{\Mst} (m))$\;\label{line:arcsuccessor}
        \ForEach{$w\in X_W$ such that $m $ prefers $\Mst (m)$ to $w$ to $\widetilde{s}_{\Mst} (m)$ and $s^{\reor}_{\Mst} (w) \neq m$}{
          Add arc $(v_m, v_{m'})$ to $H$, where $m' := \Mst (w)$\;\label{line:arcxw}
        }
      }
  }\label{line:endconstruction}
  
  \If{$H$ does not contain a spanning anti-arborescence}{
    Reject this guess\;\label{line:reject1}}\label{line:end}
  
  Repeat Lines~\ref{line:begin}--\ref{line:end} with roles of women and men 
swapped\;\label{line:sym}
  
  Accept this guess\;

\caption{An XP algorithm wrt.~$\ell$ for 
\exact-\uni-\reor.}\label{alg:xp-unique-reor}
\end{algorithm}

The algorithm starts by guessing the
subsets of men~$X_U\subseteq U$ and women 
$X_W\subseteq W$ of summed size $\ell$ whose preferences we reorder 
(Line~\ref{line:guessx}). We reject a
guess if there exists a blocking pair~$\{m,w\}\in \bp(\Mst,\mathcal{I})$ such 
that~$m\notin X_U$ and~$w\notin X_W$, as in this case it is not possible to 
resolve this blocking pair using the guessed agents (Line~\ref{line:checkbp}). Moreover, for each~$m\in
X_U$, we guess his rotation-successor~$s^\reor_{\Mst}(m)\in W\cup 
\{\emptyset\}$ after 
the reorderings and for each~$w\in X_W$,
we guess her rotation successor~$s^\reor_{\Mst}(w)\in U\cup \{\emptyset\}$ 
(Lines~\ref{line:guesssm} and~\ref{line:guesssw}). We reject the guess if it is 
impossible for an agent $a\in X_U\cup X_W$ to make $s^\reor_{\Mst}(a)$ $a$'s 
rotation successor.
This is the case if $s^\reor_{\Mst}(a)\notin X_U\cup X_W$  
and~$s^\reor_{\Mst}(a)$ prefers~$M(s^\reor_{\Mst}(a))$ to $a$  
(Line~\ref{line:checksuccs}). 
Further, if there exists some~$m\in X_U$ for which we have guessed 
that~$s^\reor_{\Mst}(m)=\emptyset$ and there exists some~$w\in W\setminus X_W$ 
that 
prefers~$m$ to~$\Mst(w)$, then we reject
the guess,
as in this case if~$m$ ranks~$w$ above~$\Mst(m)$, we create a blocking pair and 
if $m$ ranks~$\Mst(m)$ 
above~$w$, then~$m$ has a rotation successor; the same holds with roles of 
women and 
men swapped (Line~\ref{line:checkrotation}).
We also check whether there is a man-woman pair~$(m, w)$ with $m \in X_U$ and 
$w\in X_W$ such that we guessed that they are their mutual rotation successor 
and reject a guess in this case (Line~\ref{line:checkmutualsuccessor}), because 
if $w$ is the rotation successor of $m$, then $m $ prefers $\Mst (m)$ to $w$, 
and if $m $ is the rotation successor of~$w$, then $m$ prefers $w$ to $\Mst 
(m)$ 
and these two conditions clearly cannot be satisfied at the same time.

In the end, for all agents~$a\in X_U\cup X_W$, we will reorder their 
preferences such that
their rotation successor is ranked directly after~$\Mst(a)$.
The only influence that the preferences of~$a$ can have on the rotation 
successor of another agent~$a'$ is whether $a$ prefers $a'$ to $\Mst 
(a)$ or not.
By
\Cref{obs:rotation}, it follows that in order to make $\Mst$ the unique stable matching, we only need to decide which 
agents are in the first part of~$a$'s preferences (and can then order them arbitrarily before $\Mst (a)$ in the preferences of $a$). Again, by
\Cref{obs:rotation}, 
this can be 
solved for the
men in~$X_U$ and women in~$X_W$ separately (with the exception that we need to 
ensure that there is no man-woman pair which are mutual rotation successors; 
however, this may happen only if both agents are contained in $X_U \cup X_W$, 
and we exclude that this happens in Line~\ref{line:checkmutualsuccessor}): 
Selecting the first part of the 
preferences of~$w\in X_W$ only influences the rotation successors of all 
men. Similarly, selecting the first part of the preferences of
$m\in X_U$ only influences whether there exists
a woman-rotation. Consequently, it is possible to split the problem into 
two parts. We describe how to determine the preferences of all~$w\in X_W$,
thereby, resolving all man-rotations. The woman-rotations can be resolved 
symmetrically (Line~\ref{line:sym}).

To determine how to reorder the preferences of all~$w\in X_W$, we reduce 
the problem to an 
instance of \textsc{Spanning Anti-Arborescence} (Lines~\ref{line:begin}--\ref{line:endconstruction}). We construct the directed graph
as follows. For each pair~$\{m,w\}\in \Mst$, we introduce a vertex~$v_{m}$ (Line~\ref{line:addvm}). 
Moreover, we add a sink~$t$ (Line~\ref{line:addt}). For 
all~$m\in X_U$, we add an arc from~$v_{m}$
to the vertex corresponding to the man matched to the guessed rotation 
successor of $m$ in $\Mst$, i.e.,~$v_{\Mst\textbf{}(s^\reor_{\Mst}(m))}$, or to $t$ if 
$s^\reor_{\Mst}(m)=\emptyset$  (Line~\ref{line:arcxu}). Now, we add for each 
two vertices~$v_m$ and 
$v_{m'}$ with~$m \in U \setminus X_U$ and $m'\in X_U$ an arc from~$v_m$ 
to~$v_{m'}$ if we can reorder the guessed 
agents' preferences such that~$\Mst(m')$ is~$m$'s rotation successor (Lines~\ref{line:beginm}--\ref{line:endconstruction}).

More formally, for 
each $m\in U\setminus X_U$, we denote as $\widetilde{s}_{\Mst}(m)$ 
agent~$m$'s most-preferred 
woman~$w\in W\setminus X_W$ who prefers~$m$ to
$\Mst(w)$ and is ranked after $\Mst(m)$ by~$m$ (independent of how the 
preferences of the guessed agents are reordered, $m$ cannot have a rotation 
successor to which he prefers $\widetilde{s}_{\Mst}(m)$). We deal with the case 
that $\widetilde{s}_{\Mst}(m)$ is not well-defined separately below. Now, for 
each
$w'\in X_W$, who is ranked between $\Mst(m)$ and 
$\widetilde{s}_{\Mst}(m)$ in the preferences of~$m$ and fulfills ${s}^{\reor}_{\Mst} (w') \neq m$, we add an arc from $v_{m}$ to 
$v_{\Mst(w')}$ (for those women, we can decide whether they rank $m$ before or 
after $\Mst(w')$ and thus whether they become $m$'s rotation successor or not). 
Note that we require ${s}^{\reor}_{\Mst}(w') \neq m$ as otherwise~$w'$ must 
prefer $\Mst (w')$ to $m$ and thus cannot be the rotation successor of $m$.
Moreover, we add an arc from $v_{m}$ to 
$v_{\Mst(\widetilde{s}_{\Mst}(m))}$ (Lines~\ref{line:arcsuccessor} 
and~\ref{line:arcxw}). As mentioned above, it may happen 
that~$\widetilde{s}_{\Mst}(m)$ is undefined for some~$m\in 
U\setminus X_U$. In this case, we add an arc between $v_m$ and $v_{\Mst(w')}$ 
for each~$w'\in 
X_W$ which~$m$ ranks below $\Mst(m)$ and fulfills ${s}^{\reor}_{\Mst} (w') \neq m$ and an arc from $v_m$ to~$t$ (Lines~\ref{line:sinkedge} and~\ref{line:arc}).
We call the resulting graph~$H_U$ and the graph constructed
using the same algorithm with
the roles of men and women switched~$H_W$.

We compute an anti-arborescence in $H_U$. In the anti-arborescence, for each 
$v_m$, the end point of its outgoing 
arc corresponds to the 
rotation successor of $m$ in the modified instance. To ensure this, we 
construct 
the	preferences of all women $w\in X_W$ as follows. For each $w\in X_W$, 
we rank all men $m\in U$ 
such that there is an arc from $v_m$ to $v_{\Mst(w)}$ in the anti-arborescence 
in 
an arbitrary order before $\Mst(w)$, while we place the guessed
man~$s^\reor_{\Mst}(w)$ directly after~$\Mst(w)$ and add the remaining agents in an
arbitrary order after~$s^\reor_{\Mst}(w)$.
If $s^{\reor}_{\Mst} (w) = \emptyset$, then we place all men $m \in U$ such 
that there is an arc from $v_m$ to $v_{\Mst (w)}$ in the anti-arborescence 
before $\Mst (w)$ and all other men in an arbitrary order after $\Mst (w)$.

We use the same procedure to determine the preferences of all $m\in X_U$. 
Thereby, if there exist
anti-arborescences in $H_U$ and $H_W$, we are able to reorder the preferences 
of 
the guessed agents such that $\Mst$ becomes the unique stable 
matching. Thus, we return YES in this case.
Otherwise, we reject this guess (Lines~\ref{line:reject1} and~\ref{line:sym}),
continue with the next guess and return NO after rejecting the last 
guess.

It remains to prove the correctness of the algorithm: 
\begin{lemma} \label{le:ex-uni-reor-XP1}
    If the algorithm accepts a guess, then there exists a solution to the given 
    instance of \exact-\uni-\reor.
\end{lemma}
\begin{proof}
    We now prove that for every pair $(\mathcal{A}_U, \mathcal{A}_W)$ of 
    anti-arborescences found in
    the graphs~$H_U$ and~$H_W$, the resulting reorderings of the preferences
    make $\Mst$ the unique stable matching.
    First we show that the preferences are well-defined:
    The only case in which we require for two agents~$a$ and $a'$ that $a $ 
    prefers~$a'$ to $\Mst (a')$ and $a$ prefers $\Mst (a')$ to $a'$ is when $a' 
    $ 
    is the rotation successor of~$a$ and $a$ is the rotation successor of~$a'$.
    However, due to the check in Line~\ref{line:checkmutualsuccessor}, this is 
    not possible.
    For the sake of contradiction, let 
    us 
    assume that $\Mst$ is not the unique stable matching. There are two 
    possibilities, either $\Mst$ is not a stable matching or $\Mst$ is not the 
    unique 
    stable matching. 
    
    First, we show that $\Mst$ is stable.
    Since we rejected each guess containing a blocking pair~$\{m, w\}$ with $m 
\in U\setminus 
    X_U$ and $w\in W \setminus X_W$, each blocking pair involves at least one 
    agent
    from $X_U\cup X_W$.
    Fix a blocking pair $\{m, w\}$.
    We assume without loss of generality
    that~$m\in X_U$.
    The algorithm
    constructs the preferences of $m$ such that $m$ only prefers a
    woman~$w$ to $\Mst(m)$ if the arc from $v_w$ to $v_{\Mst(m)}$ is part of 
    the 
    anti-arborescence $\mathcal{A}_W$.
    However, such an arc only exists in~$H_W$ if $w$ ranks $m$ below $\Mst(w)$,
    which implies that~$\{m,w\}$ cannot be blocking.
    
    Now we show that $\Mst$ is the unique stable matching.
    To do so, we show that for each man $m$ with $(v_m, v_{\Mst(w)})\in 
    \mathcal{A}_U$, the rotation successor of $m$ is the women $w$.
    If~$\mathcal{A}_U$ contains the arc $(v_m, t)$, then $m$ has no rotation 
    successor.
    A symmetric statement also holds for each woman.
    From this, the uniqueness of $\Mst$ easily follows, as for any rotation 
    (without loss
    of generality a man-rotation)
    $(m_{i_1},w_{j_1}),\dots,(m_{i_{r}},w_{j_{r}})$ exposed in $\Mst$, the 
    anti-arborescence~$\mathcal{A}_U$ then contains the arcs 
    $(v_{m_{i_{j}}},v_{m_{i_{j+1}}})$ for all $j\in
    [r]$, and therefore contains a cycle, a contradiction to $\mathcal{A}_U$ 
being an 
    anti-arborescence.
    
    So consider a man $m$ with $(v_m, v_{\Mst(w)})\in \mathcal{A}
    _U$.
    By the definition of $H_U$, man~$m$ prefers~$\Mst(m) $ to $w$.
    We claim that for every woman $w'\in W$ whom $m$ ranks between~$\Mst(m)$ and 
    $w$, 
    it holds that $w'$ prefers $\Mst(w')$ to $m$ after the modifications. If 
    $w'\in X_W$, then we reorder the preferences of $w'$ in this way; 
    otherwise, this follows since~$H_U$ contains edge~$(v_m, v_{\Mst (w)})$. It remains to show that 
    $w$ prefers $m$ to $\Mst(w)$. 
    If $w\in X_W$, then we reordered the preferences of $w$ such that $w$ 
    prefers $m$ to $\Mst(w)$.
    Otherwise, woman $w$ prefers $m$ to $\Mst(w)$ since~$H_U$ contains edge~$(v_m, v_{\Mst (w)})$.
    If $(v_m, t) \in \mathcal{A}_U$, then analogous arguments show that any 
    woman $w$ after~$\Mst(m) $ in $m$'s preferences does not prefer $m$ to~$\Mst(w)$.
\end{proof}

\begin{lemma} \label{le:ex-uni-reor-XP2}
    If the algorithm rejects every guess, then there exists no solution to the given 
    instance of \exact-\uni-\reor.
\end{lemma}
\begin{proof}
    For the sake of contradiction, let us assume that the given instance of
    \exact-\uni-\reor admits a solution. We claim that
    there exists a guess of~$X_U\cup X_W$ and their rotation successors for 
    which 
    $H_U$ and
    $H_W$ 
    both 
    admit anti-arborescences. This leads to a contradiction, thereby proving 
    the lemma. 
    
    Assume that there exists a set $Y_A= Y_U\cup Y_W$ with
    $Y_U\subseteq U$ and $Y_W\subseteq W$ of $\ell$~agents and a reordering of 
    the preferences of these agents such that $\Mst$ is the unique stable 
    matching 
    in the resulting instance. Let $\mathcal{I'}$ denote the manipulated 
    instance and let $s_{\Mst}^{\mathcal{I}'}(m)$ denote the rotation successor 
    of 
    some $m\in U$ in $\mathcal{I'}$ and $s_{\Mst}^{\mathcal{I}'}(w)$ the 
rotation 
    successor of some~$w\in W$. Then, there exists a guess where $X_U=Y_U$, 
    $X_W=Y_W$, and for all $m\in X_U$, his rotation successor is~$s_{\Mst}^{\mathcal{I}'}(m)$, i.e., $s_{\Mst}^\reor (m) =  
    s_{\Mst}^{\mathcal{I}'}(m)$,
    and for all $w\in X_U$, her rotation successor is $s_{\Mst}^{\mathcal{I}'}(w)$. First of all 
    note that the guess is not immediately rejected, as for each blocking pair 
    one of the involved agents needs to be part of $X_U\cup X_W$, and no agent 
    from $X_A$ without a rotation successor can be preferred by an unmodified 
    agent of opposite gender to its partner in $\Mst$.
    Furthermore, there cannot be a man-woman pair~$(m, w)$ with $m \in X_U$ and $w\in X_W$ such that they are to be their mutual rotation successor:
    If $w$ is the rotation successor of $m$, then $m $ prefers $\Mst (m)$ to $w$, and if $m $ is the rotation successor of $w$, then $m$ prefers $w$ to $\Mst (m)$ and these two conditions clearly cannot be satisfied at the same time.
    
    We describe how to construct an 
    anti-arborescence $\mathcal{A}_U$ for $H_U$ while the construction 
    for~$\mathcal{A}_W$ works
    analogously. 
    For 
    each $m\in U$, we include the arc 
    $(v_m,v_{\Mst(s_{\Mst}^{\mathcal{I}'}(m))})$ in 
    $\mathcal{A}_U$ and the arc $(v_m,t)$ if
    $s_{\Mst}^{\mathcal{I}'}(m)=\emptyset$.
    As there is no rotation exposed in $\Mst$ in $\mathcal{I}'$, the resulting 
    graph $\mathcal{A}_U$ is acyclic and every vertex but $t$ has out-degree 
    exactly 1, i.e., $\mathcal{A}_U$ is indeed an anti-arborescence.
    
    It remains to show that $\mathcal{A}_U$ is a subgraph of $H_U$.
    Fix an arc $(v_m, v_{\Mst(s_{\Mst}^{\mathcal{I}'} (m))})\in \mathcal{A}_U$ 
    (where~$v_{\Mst(\emptyset)} := t$).
    For all $m\in X_U$, this arc is contained in $H_U$, as we have already 
    guessed~$s_{\Mst}^{\mathcal{I}'}(m)$ and added the arc
    $(v_m,v_{\Mst(s_{\Mst}^{\mathcal{I}'}(m))})$ 
    to $H_U$. For all $m\in U\setminus X_U$, the 
    woman~$s_{\Mst}^{\mathcal{I}'}(m)$ 
    needs
    to be ranked below $\Mst(m)$ in $m$'s preferences. Moreover, by
    definition,~$s_{\Mst}^{\mathcal{I}'}(m)$ is the first woman after $\Mst(m)$ 
    in 
    $m$'s preferences who prefers~$m$ over~$\Mst(w)$. Thereby, $w$ cannot be 
    ranked after the first
    woman $\widetilde{s}_{\Mst} (m) \in W\setminus X_W$ who prefers~$m$ 
    to~$\Mst(\widetilde{s}_{\Mst} (m))$. Thus, if $w$
    is not~$\widetilde{s}_{\Mst} (m)$, then $w$ is contained in $X_W$.
    In fact, for
    all such women $w$ there exists an arc from $v_m$ to~$v_{\Mst(w)}$ in 
    $H_U$.
\end{proof}

The developed algorithm runs in $\mathcal{O}(2^\ell n^{2\ell+2})$
time since we iterate over up to~$\binom{2n}{\ell}$ guesses for~$X_A$ and for each of 
these guesses, we iterate over $\mathcal{O}(n^\ell)$ guesses for the rotation successors. 
For each guess of~$X_A$ and the rotation successors, graph~$H$ and the anti-arborescence can be computed in 
$\mathcal{O}(n^2)$ time. 
Altogether, the following theorem results from \Cref{le:ex-uni-reor-XP1} 
and 
\Cref{le:ex-uni-reor-XP2}:

\begin{theorem} \label{pr:ex-uni-reor-XP}
    \exact-\uni-\reor is solvable in $\mathcal{O}(2^\ell n^{2\ell+2})$ time.
\end{theorem} 

\section{Conclusion} \label{se:conc}
We provided a first comprehensive study of the computational complexity
of several manipulative actions and goals in the context of the 
\textsc{Stable Marriage} problem. 
Our diverse set of 
computational complexity results is surveyed in \cref{ta:sum}.

Several challenges for future research remain.
In contrast to the constructive setting considered here, there is also a 
destructive view on 
manipulation, where the goal is to prevent a
certain constellation (see 
\Cref{subsub:const-reor} for a discussion which of our results translate).
Moreover, for the \Iconst-\Iuni scenario not presented here (where one edge shall be contained in every stable matching), our 
hardness 
results for \Iconst-\Iex carry over, as the stable matchings constructed in the proofs of \Cref{th:const-ex-swap,th:const-ex-reor,cor:const-ex-add} are indeed unique.
However, the algorithm for \const-\ex-\delete as well as the 2-approximation 
algorithm for \const-\ex-\reor do not work for \Iconst-\Iuni.
A very specific open question is whether the \exact-\uni-\swap 
problem is fixed-parameter tractable 
when parameterized by the budget.
Additionally, there is clearly a lot of room for investigating more 
manipulative actions. For instance, a manipulator might be able to divide the 
set 
of agents into two parts, where a separate matching for each part needs to be 
found and agents from different parts cannot form a blocking pair. Further, in 
the 
presence of ties, 
assuming a less powerful manipulator, a 
manipulator might only be able to break ties in the preferences  (whether this 
manipulation can have an 
impact depends however on the concrete 
stability concept and the manipulation goal considered). Notably, in several 
matching mechanisms used in practice, e.g., in the context of matching students 
to schools in Estonia \citep{estreport} and in several larger US cities 
\citep{erdil2008s} and in the context of assigning residents to hospitals in 
Scotland 
\citep{scotreport}, ties in the 
agent's preferences are broken uniformly at random. Clearly, one may also
extend the study 
of external 
manipulation 
to stable 
matching problems beyond
\textsc{Stable Marriage}.
A further natural extension of our work would be to consider
weighted manipulation, where one assumes that each possible manipulative action 
comes at a specific cost. All our results 
for the two \Iexact settings as well as the hardness results for \Iconst carry over to the weighted case in a straightforward way, while the 
polynomial-time 
algorithm for \const-\ex-\delete does not work anymore. 

Lastly, it might also be interesting to analyze the power of different 
manipulative actions in real-world scenarios. We already did some very 
preliminary experiments 
for all polynomial-computable cases 
on synthetic data having between 30 and 200 agents, where the 
preferences of agents were drawn uniformly at 
random 
from all possible preferences. The manipulation goal was also set uniformly at 
random. The following two observations were particularly surprising to us: In 
the \Iconst-\Iex setting, \Idelete 
operations seem to be quite 
powerful, as most of the time deleting a 
moderately 
low number of agents (around 10\%) sufficed. In the \Iexact-\Iex setting, 
\Ireor operations are not 
as powerful as one might intuitively suspect, as,  on average, close to half of 
the agents needed to be 
modified---note that there always exists a trivial solution where the 
preferences of 
all agents from one gender are reordered. 

\bigskip

\subsection*{Acknowledgments}
An extended abstract of the paper appeared in the proceedings of the 
\emph{13th International Symposium on Algorithmic Game Theory (SAGT~2020)}, 
Springer LNCS~12283, pages~163--177, 2020. Niclas Boehmer was 
supported by the DFG project MaMu  (NI 
369/19). Klaus 
Heeger was supported by the DFG Research Training Group 2434 ``Facets of 
Complexity''.

\bibliographystyle{plainnat}

\end{document}